\documentclass[a4paper,USenglish]{lipics-v2016}

\usepackage{microtype}
\usepackage{xspace}
\usepackage{amsfonts}
\usepackage{amssymb}
\usepackage{mathtools}
\usepackage{booktabs}
\usepackage{tikz}
\usepackage{csquotes}
\usepackage{stmaryrd}
\usepackage{colonequals}
\usepackage{apxproof}

\theoremstyle{plain}
\def\appendixbibliographystyle{alphaabbr}
\newtheoremrep{theorem}{Theorem}
\newtheoremrep{proposition}[theorem]{Proposition}
\newtheoremrep{lemma}[theorem]{Lemma}

\newcommand{\OR}{\Gamma} %

\newcommand{\PosRAaccbold}{\textrm{PosRA}$^{\textrm{acc}}$\xspace}
\newcommand{\Plexbold}{\textrm{PosRA}$_{\lexbold}$\xspace}
\newcommand{\Plexaccbold}{\textrm{PosRA}$_{\lexbold}^{\textrm{acc}}$\xspace}
\newcommand{\PosRA}{{\textrm{PosRA}}\xspace}
\newcommand{\Plex}{{\textrm{PosRA}}$_{\lex}$\xspace}
\newcommand{\Pnoprod}{{\textrm{PosRA}}$_{\mathrm{no}\times}$\xspace}
\newcommand{\Pgen}{{\textrm{PosRA}}$_{\gen}$\xspace}
\newcommand{\PosRAagg}{{\textrm{PosRA}}$^{\mathrm{acc}}$\xspace}
\newcommand{\PosRAagggby}{{\textrm{PosRA}}$^{\mathrm{accGBy}}$\xspace}
\newcommand{\Plexacc}{{\textrm{PosRA}}$_{\lex}^{\mathrm{acc}}$\xspace}

\newcommand{\Pnoprodacc}{{\textrm{PosRA}}$^{\mathrm{acc}}_{\mathrm{no}\times}$\xspace}
\newcommand{\posRAagg}{\PosRAagg}
\newcommand{\PosRAacc}{\PosRAagg}
\newcommand{\PosRAaccgby}{\PosRAagggby}

\newcommand{\Nat}{\mbox{$\mathbb{N}$}}

\newcommand{\ptw}{\dir}
\newcommand{\dir}{{\texttt{DIR}}\xspace}
\newcommand{\gen}{\ptw}
\newcommand{\lexbold}{\texttt{LEX}\xspace}
\newcommand{\lex}{{\texttt{LEX}}\xspace}
\newcommand{\cat}{{\texttt{CAT}}\xspace}

\newcommand{\poss}{{\texttt{POSS}}\xspace}
\newcommand{\cert}{{\texttt{CERT}}\xspace}

\newcommand{\cupgen}{\cup}
\newcommand{\cupcat}{\cup_\cat}

\newcommand{\deft}[1]{\emph{#1}}

\makeatletter
\newcommand*{\defeq}{\mathrel{\rlap{%
  \raisebox{0.3ex}{$\m@th\cdot$}}%
  \raisebox{-0.3ex}{$\m@th\cdot$}}%
  =}
\makeatother

\renewcommand{\leq}{\leqslant}
\renewcommand{\geq}{\geqslant}
\renewcommand{\phi}{\varphi}
\renewcommand{\epsilon}{\varepsilon}

\newcommand\restr[2]{{%
  \kern-\nulldelimiterspace %
  #1 %
  _{|#2} %
  }}

\newcommand{\singleton}[1]{[#1]}
\newcommand{\ordern}[1]{[{\leq}#1]}

\newcommand{\card}[1]{\left|#1\right|}

\newcommand{\ID}{\mathit{ID}}
\newcommand{\id}{\mathit{id}}

\newcommand{\calD}{\mathcal{D}}

\newcommand{\calF}{\mathcal{F}}

\newcommand{\calI}{\mathcal{I}}

\newcommand{\calM}{\mathcal{M}}

\newcommand{\calS}{\mathcal{S}}

\newcommand{\arity}[1]{\mathop{\mathrm{a}}(#1)}

\newcommand{\axiom}[1]{\textsf{#1}}

\newcommand{\dupelim}{\mathrm{dupElim}}
\newcommand{\pw}{\mathit{pw}}

\newcommand{\langlem}{}
\newcommand{\ranglem}{}

\newcommand{\accum}{\mathrm{accum}}
\newcommand{\accumgby}{\mathrm{accumGroupBy}}
\newcommand{\concat}{\mathrm{concat}}
\newcommand{\pr}{\mathrm{pr}}

\renewcommand{\a}{\mathsf{a}}
\renewcommand{\b}{\mathsf{b}}
\newcommand{\cn}{\mathsf{c}}
\newcommand{\n}{\mathsf{n}}
\newcommand{\e}{\mathsf{e}}
\newcommand{\f}{\mathsf{f}}
\renewcommand{\i}{\mathsf{i}}
\renewcommand{\l}{\mathsf{l}}
\renewcommand{\r}{\mathsf{r}}
\newcommand{\s}{\mathsf{s}}

\newcommand{\NN}{\mathbb{N}}
\newcommand{\ZZ}{\mathbb{Z}}

\newcommand{\tikzm}[1]{\tikz[overlay,remember picture]{\node (#1) {};}}
\newcommand{\tikzmd}[2]{\tikz[overlay,remember picture,thick]{\draw[->] (#1) -- (#2);}}

\newcommand{\bag}{\mathrm{bag}}

\newcommand{\BALG}{\text{BALG}}
\newcommand{\BQL}{\mathcal{BQL}}

\newcommand{\NRLaggr}{\mathcal{NRL}^{\mathrm{aggr}}}
\newcommand{\PomAlg}{\mathcal{P}\text{om-}\mathcal{A}\text{lg}}
\newcommand{\PomAlgEps}{\PomAlg_{\epsilon_n}}

\title{Possible and Certain Answers for Queries\protect\\ over Order-Incomplete Data}
\titlerunning{Possible and Certain Answers for Queries over Order-Incomplete Data}

\hypersetup{
    colorlinks,
    linkcolor={red!50!black},
    citecolor={blue!50!black},
    urlcolor={blue!30!black}
}

\keywords{certain answer; possible answer; partial order; uncertain data}
\subjclass{H.2.1 Database Management -- Logical Design}

\EventEditors{Sven Schewe, Thomas Schneider, and Jef Wijsen}
\EventNoEds{3}
\EventLongTitle{24th International Symposium on Temporal
Representation and Reasoning (TIME 2017)}
\EventShortTitle{TIME 2017}
\EventAcronym{TIME}
\EventYear{2017}
\EventDate{October 16--18, 2017}
\EventLocation{Mons, Belgium}
\EventLogo{}
\SeriesVolume{90}
\ArticleNo{4}

\author[1]{Antoine Amarilli}
\author[2]{Mouhamadou Lamine Ba}
\author[3]{Daniel Deutch}
\author[4,5]{Pierre~Senellart}
\affil[1]{LTCI, T\'el\'ecom ParisTech, Universit\'e Paris-Saclay;
Paris, France}
\affil[2]{University Alioune Diop of Bambey; Bambey, Senegal}
\affil[3]{Blavatnik School of Computer Science, Tel Aviv University;
Tel Aviv, Israel}
\affil[4]{DI ENS, ENS, CNRS, PSL Research University; Paris, France}
\affil[5]{Inria Paris; Paris, France}

\Copyright{Antoine Amarilli, Mouhamadou Lamine Ba, Daniel Deutch, and
Pierre Senellart}

\begin{document}

\maketitle

\begin{abstract}
To combine and query ordered data from multiple sources, one needs to
handle uncertainty about the possible orderings. Examples of such
  ``order-incomplete'' data include integrated event sequences such as
  log entries; lists of properties (e.g., hotels and
  restaurants) ranked by an unknown function reflecting relevance or
  customer ratings; and documents edited concurrently with an uncertain order
  on edits. This
  paper introduces a query language for order-incomplete data, based on
  the positive relational algebra with order-aware accumulation. We use
  partial orders to represent order-incomplete data, and study possible
  and certain answers for queries in this context. We show that these
  problems are respectively NP-complete and coNP-complete, but identify
  tractable cases depending on the query operators or input partial
  orders.
\end{abstract}

\section{Introduction}
\label{sec:introduction}
Many applications need to combine and transform ordered data (e.g., temporal
data, rankings, preferences) from
multiple sources. Examples include sequences of readings from
multiple sensors, or log entries from different applications or
machines, that must be combined to form a complete picture of
events; rankings of restaurants and hotels published by
different websites, their ranking function being often
proprietary and unknown; and concurrent edits of shared documents,
where the order of contributions made by different users needs to be
merged. Even when the order of items from each individual source is
known, the order across sources is often {\em uncertain}.
For instance, even when sensor readings or log entries have
time\-stamps, these may be ill-synchronized across sensors or
machines; different websites may follow different rules
and rank different hotels, so there are multiple ways to create a
unified ranked list; concurrent document editions may be
ordered in multiple ways. We say that the resulting
information is \emph{order-incomplete}.

This paper studies query evaluation over order-incomplete data in
a relational setting \cite{AHV-1995}. Our running example is that of restaurants
and hotels from travel websites, ranked according to proprietary
functions. An example query could compute the union of ranked lists of restaurants from distinct
websites, or ask for a ranked list of pairs of a restaurant and a
hotel 
in the same district. As we do not know how the proprietary order is
defined, the query result may become \emph{uncertain}:
there may be multiple reasonable orderings of
restaurants in the union result, or multiple orderings of
restaurant--hotel pairs. We also study the application of
order-aware \emph{accumulation} to
the query result, where each possible order may yield a different value:
e.g., extracting only the highest ranked pairs,
concatenating 
their names, or assessing the
attractiveness of a district based on its
best restaurants and hotels.

Our approach is to handle this uncertainty through the classical notions
of {\em possible and certain answers}. First, whenever there is a
\emph{certain answer} to the query -- i.e., there is only one possible order on query results or one accumulation result -- 
which is obtained no matter the order on the input and in intermediate
results, we should present it to the user, who can then browse through
the ordered query results (as is typically done in absence of
uncertainty, using constructs such as SQL's \texttt{ORDER BY}). 
Certain answers can arise even in non-trivial cases 
where the combination of input data admits 
many possible
orders: consider user queries that select only a small interesting subset of the
data (for which the ordering happens to be certain), or a short summary
obtained through accumulation over large data. In many other cases, the
different orders on input data or the uncertainty caused by the query may
lead to several \emph{possible answers}. In this case, it is still of
interest (and non-trivial) to verify whether an answer is
possible, e.g., to check whether a given ranking of hotel--restaurant pairs is consistent with a combination of other rankings (the latter done through a query). Thus, we study the problems of deciding whether a given answer is
\emph{certain}, and whether it is \emph{possible}.

As users may wish to focus on the position of some tuples of
interest (e.g., ``is it possible/certain that a particular restaurant--hotel pair 
is ranked first?'', or ``is it possible/certain that restaurant
$A$ is ranked above restaurant $B$?), we show that these questions may be expressed in our framework through proper choices of accumulation functions.

\subparagraph*{Main contributions.}
We introduce a query
language with accumulation
for order-incomplete data, 
which 
 generalizes the positive relational algebra \cite{AHV-1995} with aggregation as
 the outermost
 operation.
We define a bag semantics for this language, without assuming that a single choice of order can be
made (unlike, e.g., rank aggregation~\cite{Fagin}): we use \emph{partial orders}
to represent
all orders that are consistent with the  input data.
 We then
undertake 
the first general study of the
\emph{complexity of possible and certain answers for queries over
such data}.
We show that these problems are respectively NP-complete and coNP-complete,
the main difficulties being the existence of
duplicate tuple values in the data and the use of order-aware accumulation.
Fortunately, we can show realistic tractable cases: certainty is in PTIME
without accumulation, and both problems are tractable under reasonable
restrictions on the input and on the query.

The rest of this paper is organized as follows.
In Section~\ref{sec:model}, we introduce our data model and our query language.
We define and exemplify the problems of possible and certain answers in Section~\ref{sec:posscertdef}.
We then study their complexity, first in the general case (Section~\ref{sec:posscert}),
then in restricted settings that ensure tractability 
(Sections~\ref{sec:fpt2} and~\ref{sec:fpt}).
We study extensions to the language, namely duplicate elimination and group-by, in Section~\ref{sec:extensions}.
We compare our model and results with related work in Section~\ref{sec:compare}, and conclude in Section~\ref{sec:conclusion}. 

Full proofs of all results are given in an extensive appendix, for lack
of space. Please note that this version of the paper removes some erroneous
results relative to an earlier arXiv version and the conference proceedings
version: see Appendix~\ref{apx:errata} for details.

\section{Data Model and Query Language}\label{sec:model}
We fix a countable set of values $\calD$
that includes $\Nat$ and infinitely many values not in $\Nat$.
A \emph{tuple~$t$ over~$\calD$} of \emph{arity}~$\arity{t}$ is an
element of~$\calD^{\arity{t}}$, denoted $\langle v_1, \dots,
v_{\smash{\arity{t}}}\rangle$. The simplest notion of ordered relations are then
\emph{list relations}~\cite{colby1994query,colby1994concepts}: a list relation
of arity~$n \in \Nat$
is an ordered list of tuples over~$\calD$ of arity~$n$ (where
the same tuple value may appear multiple times).
List relations impose a single order over tuples, but when one
combines (e.g., unions) them, there may be multiple plausible ways
to order the results. 

We thus introduce \emph{partially
ordered relations} (\emph{po-relations}). A po-relation  $\OR =
(\ID, T, <)$ of arity~$n \in \Nat$ consists of a finite set of
\emph{identifiers} $\ID$ (chosen from some infinite set
closed under product), a \emph{strict partial order} $<$ on~$\ID$, and a (generally non
injective) mapping $T$ from $\ID$ to~$\calD^n$.
The actual identifiers do not matter, but
we need them to
refer to occurrences of the
same tuple value. Hence, we always
consider po-relations \emph{up to isomorphism}, where $(\ID, T, <)$ and $(\ID',
T', {<'})$ are \emph{isomorphic} iff there is a bijection $\phi: \ID \to \ID'$
such that $T'(\phi(\id)) = T(\id)$ for all $\id \in \ID$, and $\phi(\id_1) {<'}
\phi(\id_2)$ iff $\id_1 < \id_2$ for all $\id_1, \id_2 \in \ID$.

A special case of po-relations
are \emph{unordered po-relations} (or \emph{bag relations}), where $<$ is empty:
we write them $(\ID, T)$. The \emph{underlying bag relation} of 
$\OR = (\ID, T, <)$ is $(\ID, T)$.

The point of po-relations is to represent \emph{sets} of list relations.
Formally, a \emph{linear extension} $<'$ of~$<$ is a total order on~$\ID$ such
that for each $x<y$ we have $x <' y$. The 
{\em possible worlds} $\pw(\OR)$ of $\OR$ are then defined as follows:
for each linear extension ${<'}$ of~$<$, writing $\ID$ as $\id_1 <' \cdots <'
\id_{\card{\ID}}$,
the list relation $(T(\id_1), \ldots, T(\id_{\card{ID}}))$ is
in~$\pw(\OR)$. As $T$ is generally not injective, two different linear extensions may yield the
same list relation.
Po-relations can thus model uncertainty over the \emph{order} of
tuples (but not on their \emph{value}: the underlying bag relation is always
certain).

\begin{figure}
\noindent\begin{minipage}[b]{.62\linewidth}
{\renewcommand{\tabcolsep}{4pt}
{\small\noindent\begin{tabular}{l@{\quad}l@{\quad}l}
$\!\!$\begin{tabular}[b]{l@{\,\,}c}
\toprule
\textit{restname} & \textit{distr}   \\
\midrule
Gagnaire & 8\,\,\,\,\,\raisebox{0.9em}{\tikzm{froma}} \\
TourArgent & 5\,\,\,\,\,\raisebox{-0.3em}{\tikzm{toa}}\tikzmd{froma}{toa} \\
\bottomrule\\[-.8em]
\multicolumn{2}{c}{{(a) $\mathit{Rest}$ table}}\\
\end{tabular}
&
\begin{tabular}[b]{l@{~}c}
\toprule
\textit{hotelname} & \textit{distr}   \\
\midrule
Mercure & 5\phantom{2}\quad\raisebox{0.9em}{\tikzm{fromb}} \\
Balzac & 8\phantom{2}\quad\raisebox{0.9em}{\tikzm{dummy}} \\
Mercure & 12\quad\raisebox{-0.3em}{\tikzm{tob}}\tikzmd{fromb}{tob} \\
\bottomrule\\[-.8em]
\multicolumn{2}{c}{{(b) $\mathit{Hotel}$ table}}
\end{tabular}
&
\begin{tabular}[b]{l@{~}c}
\toprule
\textit{hotelname} & \textit{distr}   \\
\midrule
Balzac & 8\phantom{2}\quad\raisebox{0.9em}{\tikzm{fromc}} \\
Mercure & 5\phantom{2}\quad\raisebox{0.9em}{\tikzm{dummyd}} \\
Mercure & 12\quad\raisebox{-0.3em}{\tikzm{toc}}\tikzmd{fromc}{toc} \\
\bottomrule\\[-.8em]
\multicolumn{2}{c}{{(c) $\mathit{Hotel}_2$ table}}
\end{tabular}
\end{tabular}}
} \vspace{-.6em}\caption{Running example: Paris restaurants and hotels} \label{fig:examplerels}
\end{minipage}\hfill\begin{minipage}[b]{.25\linewidth}
\noindent$\!\!\!\!\!$\begin{tikzpicture}[scale=0.98,yscale=1.05]
  \node (GM) at (0,-2) {$\langle\textup{G},8,\textup{M},5\rangle$};
  \node (TAM) at (1,-1) {$\langle\textup{TA},5,\textup{M},5\rangle$};
  \node (GB) at (-1,-1) {$\langle\textup{G},8,\textup{B},8\rangle$};
  \node (TAB) at (0,0) {$\langle\textup{TA},5,\textup{B},8\rangle$};
  \draw[->] (GM) -- (GB);
  \draw[->] (GB) -- (TAB);
  \draw[->] (GM) -- (TAM);
  \draw[->] (TAM) -- (TAB);
\end{tikzpicture}\null
\vspace{-.4em}
\caption{Example~\ref{exa:simplegen}}
  \label{fig:example}
\end{minipage}
\hfill\begin{minipage}[b]{.12\linewidth}
\noindent\begin{tikzpicture}[yscale=.8,every node/.style={align=center,outer sep=0,inner
  sep=2pt}]
  \node (n13) at (0, 0) {fr\\[-.6em]{\scriptsize a}};
  \node (n20) at (1, 0) {it\\[-.6em]{\scriptsize b}};
  \node (n37) at (0, 1) {fr\\[-.6em]{\scriptsize c}};
  \node (n42) at (1, 1) {it\\[-.6em]{\scriptsize d}};
  \node (n100) at (0, 2) {jp\\[-.45em]{\scriptsize e}};
  \node (n102) at (1, 2) {jp\\[-.45em]{\scriptsize f}};
  \draw[->] (n13) -- (n37);
  \draw[->] (n20) -- (n37);
  \draw[->] (n37) -- (n100);
  \draw[->] (n42) -- (n100);
  \draw[->] (n42) -- (n102);
\end{tikzpicture}\null
\vspace{-.4em}
\caption{Example~\ref{exa:notposet}}
  \label{fig:notposet}
\end{minipage}
\vspace{-.9em}
\end{figure}

\subparagraph*{Query language.}

We now define a bag semantics for
{\em positive relational algebra} operators, to manipulate po-relations with queries. The positive relational
algebra,
written \PosRA, is a standard query language for relational
data~\cite{AHV-1995}.
We will extend \PosRA later in this section with
{\em accumulation},
and add further extensions in Section~\ref{sec:extensions}.
Each \PosRA operator applies to po-relations and computes a new po-relation; we
present them in turn.

\smallskip

The \textsf{selection} operator restricts
the relation to a subset of its tuples, and the order is
the restriction of the input order. The \emph{tuple predicates}
allowed in selections are Boolean combinations of
equalities and inequalities, which can use tuple attributes and values in
$\calD$.
\begin{description}
  \item[\axiom{selection}:] For any po-relation
  $\OR = (\ID, T,<)$
  and tuple predicate $\psi$,
  we define the selection $\sigma_{\psi}(\OR) \defeq (\ID', T_{|\ID'}, <_{|\ID'})$
  where $\ID' \defeq \{\id \in \ID \mid \psi(T(\id))\text{~holds}\}$.
\end{description}
\smallskip
The \textsf{projection} operator changes tuple values in the usual way, but
keeps the original tuple ordering in the result, and retains all copies of duplicate
tuples (following our \emph{bag semantics}):
\begin{description}
  \item[\axiom{projection}:] For a po-relation
    $\OR = (\ID,T,<)$ and attributes $A_1, \ldots, A_n$,
    we define the projection $\Pi_{A_1, \ldots, A_n}(\OR) \defeq
    (\ID,T', {<})$ where $T'$ maps each $\id \in \ID$ to
    $\Pi_{A_1, \ldots, A_n}(T(\id))$.
\end{description}
\smallskip
As for \textsf{union}, we impose the minimal order constraints that are
compatible with those of the inputs.
We use the
{\em parallel composition}~\cite{posets} of two partial orders $<$
and $<'$ on disjoint sets $\ID$ and $\ID'$, i.e., the partial order $<''
\defeq ({<}
\mathbin{\parallel} {<'})$
on $\ID \cup \ID'$ defined by:
  every $\id \in \ID$ is incomparable for $<''$ with every
  $\id' \in \ID'$;
  for each $\id_{1}, \id_{2}\in ID$, we have $\id_{1} <'' \id_{2}$ iff $\id_{1} < \id_{2}$;
  for each $\id_{1}', \id_{2}' \in ID'$, we have $\id_{1}'<'' \id_{2}'$ iff
  $\id_{1}'<' \id_{2}'$.
\begin{description}
  \item[\axiom{union}:] Let $\OR = (\ID,T,<)$ and $\OR' = (\ID',T',<')$ be two po-relations of
    the same arity. We assume that the identifiers of~$\OR'$ have been renamed
    if necessary to ensure that $\ID$ and~$\ID'$ are disjoint.
    We then define $\OR \cup \OR'
    \defeq (\ID \cup \ID', T'', {<} \mathbin{\parallel} {<'})$, where $T''$
    maps
    $\id \in \ID$ to $T(\id)$ and
    $\id' \in \ID'$ to $T'(\id')$.
\end{description}
\smallskip
The union result $\OR \cup \OR'$ does not depend on how we renamed~$\OR'$,
i.e., it is unique up to isomorphism.
Our definition also
implies that $\OR \cup \OR$ is different from~$\OR$, as per bag semantics.
In particular, when $\OR$ and $\OR'$ have only one possible world,
$\OR \cup \OR'$ usually does not.

We next introduce two possible \textsf{product} operators. First, the
\emph{direct product}~\cite{stanley1986enumerative} ${<_{\ptw}} \defeq
({<} \times_\ptw {<'})$ of two partial orders $<$ and $<'$ on
sets $\ID$ and $\ID'$ is defined by $(\id_{1}, \id_{1}') <_{\ptw} (\id_{2},
\id_{2}')$ for
each $(\id_{1},  \id_{1}'), (\id_{2}, \id_{2}') \in \ID \times \ID'$ iff
$\id_{1} < \id_{2}$ and $\id_{1}' <'
\id_{2}'$.
We define the \textsf{direct product} operator over po-relations
accordingly: two identifiers in the product are comparable only
if \emph{both components} of both identifiers compare in the same
way.
\begin{description}
  \item[\axiom{direct product}:] For any po-relations $\OR = (\ID, T,<)$ and
    $\OR' = (\ID', T',<')$,
    remembering that the sets of possible identifiers is closed under product,
    we let $\OR \times_\gen \OR' \defeq
    (\ID \times \ID', T'', \allowbreak {< \times_{\ptw} <'})$,
    where $T''$ maps each $(\id, \id') \in \ID \times \ID'$ to the
    \emph{concatenation} $\langle T(\id), T'(\id') \rangle$.
\end{description}
\smallskip
Again, the direct product result often has multiple possible worlds even when
inputs do not.

The second product operator uses the \textsf{lexicographic product} (or \emph{ordinal
product}~\cite{stanley1986enumerative})
${<_\lex} \defeq ({<} \times_\lex {<'})$
of two partial orders $<$ and $<'$,
defined by $(\id_{1}, \id_{1}') <_{\lex} (\id_{2}, \id_{2}')$ for all $(\id_{1},
\id_{1}'), (\id_{2}, \id_{2}') \in \ID \times \ID'$ iff
    either $\id_{1} < \id_{2}$, or $\id_{1} = \id_{2}$ and $\id_{1}' <' \id_{2}'$.
\begin{description}
  \item[\axiom{lexicographic product}:]
    For any po-relations $\OR = (\ID, T,<)$ and $\OR' = (\ID', T',<')$,
    we define
    $\OR \times_\lex \OR'$ as $(\ID
    \times \ID', T'', < \times_\lex <')$ with $T''$ defined like for direct
    product.
\end{description}
\smallskip
Last, we define the \emph{constant expressions} that we allow:
\begin{description}
  \item[\axiom{const}:] \textbullet
    for any tuple~$t$, the singleton po-relation $\singleton{t}$
    has only one tuple with value~$t$;
    \newline \null\hspace{.61cm}\textbullet for any $n \in \Nat$, the po-relation $\ordern{n}$ has arity~$1$
    and has
    $\pw(\ordern{n}) = \{(1, \ldots, n)\}$.
\end{description}
\begin{toappendix}
\subsection{Proof of Theorem~\ref{thm:incomparable}}
  \label{apx:incomparableproof}
\end{toappendix}
A natural question is then to determine whether any of our operators is subsumed
by the others, but we show that this is not the case:

\begin{theoremrep}\label{thm:incomparable}
  No \PosRA\ operator can be expressed through a combination of the
others.
\end{theoremrep}

\begin{toappendix}
  We actually prove a stronger result, namely
  Theorem~\ref{thm:incomparable2}, where we add the $\dupelim$ operator
  to \PosRA operators. We 
consider each operator in
turn, showing it cannot be expressed through a combination of the others.

We first consider constant expressions. We will show differences in
expressiveness even when setting the input po-database to be empty.

\begin{itemize}
  \item For $\singleton{t}$, consider the query $\singleton{\langle0\rangle}$.
    The value $0$ is not in the database,
    and cannot be produced by the $\ordern{n}$ constant
    expression, and so
    this query has no equivalent that does not use the $[t]$ constant
    expression.
  \item For $\ordern{n}$, observe that $\ordern{2}$ is a po-relation with
    a non-empty order, while any query involving the other operators
    will have empty order (none of
    our unary and binary operators turns unordered po-relations into an ordered
    one, and the $\singleton{t}$ constant expression produces an
    unordered po-relation).
\end{itemize}

Moving on to unary and binary operators, all operators but products are easily shown to be
non-expressible:

\begin{description}
  \item[selection.]
    For any constant $a$ not
    in~$\mathbb{N}$, consider the po-database $D_a$ consisting of a
    single unordered po-relation with name~$R$ formed of
    two unary tuples $\langle 0\rangle$ and
    $\langle a\rangle$. Let
    $Q=\sigma_{.1\neq\text{``0''}}(R)$. Then, $Q(D_a)$ is the po-relation
    consisting only of the tuple $\langle a\rangle$. No PosRA query
    without selection has the same semantics, as no other operator than
    selection can create a po-relation containing the constant~$a$ for
    any input~$D_a$, unless it also contains the constant~$0$.
    
  \item[projection.] $\Pi$ is the only operator that can decrease the arity of an
    input po-relation.

  \item[union.] $[\langle 0\rangle]\cup[\langle 1\rangle]$ (over the empty
    po-database) cannot be simulated by any combination of operators, as
    can be simply shown by induction: no other operator will produce a
    po-relation which has in the same attribute the two elements $0$ and
    $1$.

  \item[duplicate elimination.]
    For any constant $a$ not
    in~$\mathbb{N}$, consider the po-database $D_a$ consisting of a
    single unordered po-relation with name~$R$ formed of
    two identical unary tuples $\langle a\rangle$ and
    $\langle a\rangle$. Let
    $Q=\dupelim(R)$. Then, $Q(D_a)$ is the po-relation
    consisting of the single tuple $\langle a\rangle$. No PosRA query
    without duplicate elimination has the same semantics, as no other operator than
    duplicate elimination can create a po-relation containing only once the constant~$a$ for
    any input~$\Gamma_a$.
\end{description}

Observe that product operators are the only ones that can
increase arity, so taken together they are non-redundant with the other
operators. 
There remains to prove that each of $\times_\gen$ and $\times_\lex$ is not
redundant.
As in Section~\ref{sec:fpt2}, we use the name \Pgen for the fragment of
\PosRA\ where $\times_\lex$ is not used; and \Plex for the fragment of
\PosRA\ where $\times_\gen$.

\subsubsection{Transformations Not Expressible in \Plex + dupElim}
We rely on Propositions~\ref{prp:lexwidth} and~\ref{prp:wdupe}: the
result of any \Plex
query (possibly with $\dupelim$), when it does not completely fail, has a
\emph{width} (see Definition~\ref{def:width} in Section~\ref{sec:fpt2}) bounded by a function of
the width of the original po-database. On
the other hand, consider the query $Q=R\times_\gen R$ and an input
po-database $D_n$ where $R$ is mapped to $\ordern{n}$ (an input relation
of width $1$) for an arbitrary $R_n$.
Then $Q(D_n)$ is a po-relation of width $n$, which shows $Q$ is not
expressible with the operators of \Plex and $\dupelim$.

\subsubsection{Transformations Not Expressible in \Pgen + dupElim}
We now show the converse, that \Plex expresses some transformations that
cannot be expressed in \Pgen. To do this, we introduce the
\emph{concatenation} of po-relations:

\begin{definition}
  \label{def:concat}
  The \emph{concatenation}
  $\OR \cupcat \OR'$ of two po-relations $\OR$ and $\OR'$
is the series composition of their two partial
orders. Note that $\pw(\OR \cupcat \OR') = \{L \cupcat L' \mid L \in \pw(\OR),
L' \in \pw(\OR')\}$, where $L \cupcat L'$ is the concatenation of two list
  relations in
the standard sense.
\end{definition}

  We show that concatenation can be captured with \Plex.

  \begin{lemma}
    \label{lem:lexconcat}
    For any arity $n \in \mathbb{N}$ and distinguished relation names $R$ and $R'$,
    there is a \Plex query $Q_n$ such that, for any two
    po-relations $\OR$ and $\OR'$ of arity~$n$, letting $D$ be the
    database mapping $R$ to $\OR$ and $R'$ to $\OR'$, $Q_n(D)$ is
    $\OR \cupcat \OR'$.
  \end{lemma}

  \begin{proof}
    For any $n \in \mathbb{N}$ and names $R$ and $R'$,
    consider the following query (using again numerical attribute names for
    simplicity):
    \[Q_n(R, R') \defeq
\Pi_{3\dots n+2} \left(\sigma_{.1= .2} \left(\ordern{2}
\times_\lex ((\singleton{1}\times_\lex R) \cupgen
   (\singleton{2}\times_\lex R'))\right)\right)\]
   It is easily verified that $Q_n$ satisfies the claimed property.
  \end{proof}

  By contrast, we show that concatenation cannot be captured with \Pgen
  and $\dupelim$.
\begin{lemma}\label{lem:noconcat}
  For any arity $n \in \mathbb{N}_+$ and distinguished relation names $R$ and~$R'$,
  there is no \Pgen query $Q_n$ (possibly with $\dupelim$) such that, for any
  po-relations $\OR$ and $\OR'$ of arity~$n$, letting $D$
  be the
  po-database that maps $R$ to $\OR$ and $R'$ to $\OR'$, the query result $Q_n(D)$
  is $\OR \cupcat \OR'$.
\end{lemma}

  To prove Lemma~\ref{lem:noconcat}, we first introduce the following concept:

  \begin{definition}
    Let $v \in \calD$. We call a
    po-relation $\OR = (\ID, T, <)$ \emph{$v$-impartial} if, for any
    two identifiers $\id_1$ and $\id_2$ and $1 \leq i \leq \arity{\OR}$ such that exactly
    one of $T(\id_1).i$, $T(\id_2).i$ is $v$, the following holds: $\id_1$ and
    $\id_2$ are \emph{incomparable}, namely, neither $\id_1 < \id_2$ nor $\id_2
    < \id_1$ hold.
  \end{definition}

\begin{lemma}
  \label{lem:incomp}
  Let $v \in \calD \backslash \mathbb{N}$ be a value. For any \Pgen query
  $Q$, possibly with $\dupelim$,
  for any
  po-database $D$ of $v$-impartial
  po-relations, the po-relation $Q(D)$ (when duplicate elimination does
  not completely fail) is
  $v$-impartial.
\end{lemma}

\begin{proof}
  Let $v \in \calD \backslash \mathbb{N}$ be such a value.
  We show the claim by induction on the query $Q$.

  The base cases are the following:

  \begin{itemize}
    \item For the base relations, the claim is vacuous by our hypothesis on~$D$.
    \item For the singleton constant expressions, the claim is trivial as
      they contain less than two tuples.
    \item For the $\ordern{i}$ constant expressions, the claim is
      immediate as $v \notin \mathbb{N}$.
  \end{itemize}

  We now prove the induction step:

  \begin{itemize}
    \item For selection, the claim is shown by noticing that, for any
      $v$-impartial
      po-relation $\OR$,
      letting $\OR'$ be the image of $\OR$ by any selection, $\OR'$
      is itself $v$-impartial. Indeed, considering two identifiers $\id_1$ and
      $\id_2$ in
      $\OR'$ and $1 \leq i \leq \arity{\OR}$ satisfying the condition, as $\OR$ is
      $v$-impartial, $\id_1$ and $\id_2$ are incomparable in $\OR$, so they are also
      incomparable in $\OR'$.
    \item For projection, the claim is also immediate as the property to prove
      is maintained when reordering, copying or deleting attributes. Indeed,
      considering again two identifiers $\id_1'$ and $\id_2'$ of $\OR'$ and $1
      \leq i' \leq \arity{\OR'}$, the respective preimages $\id_1$ and $\id_2$ in~$\OR$ of $\id_1'$ and
      $\id_2'$ before the
      projection satisfy the same condition for some different $1 \leq i \leq
      \arity{\OR}$ which is the
      preimage of $i'$, so we again use the impartiality of the original
      po-relation to conclude.
    \item For union, the property is preserved. Indeed, for $\OR'' \defeq \OR \cupgen
      \OR'$, writing $\OR'' = (\ID'', T'', {<''})$, assume by contradiction the existence of two identifiers $\id_1,
      \id_2 \in
      \OR''$ and $1 \leq i \leq \arity{\OR''}$ such that exactly one of
      $T''(\id_1).i$ and
      $T''(\id_2).i$ is $v$ but (without loss of generality) $\id_1 < \id_2$ in
      $\OR''$.
      It is easily seen that, as $\id_1$ and
      $\id_2$
      are not incomparable, they must come from the same relation; but then, as
      that relation was $v$-impartial, we have a contradiction.
    \item For duplicate elimination, the property is preserved as
      duplicate elimination (when it does not fail) results in a
      po-relation where the order between tuples with different values is
      preserved.
    \item We now show that the property is preserved for $\times_\gen$. Consider
      $\OR'' \defeq \OR\times_\gen \OR'$ where $\OR$ and $\OR'$ are
      $v$-impartial, and write $\OR'' = (\ID'', T'', <'')$ as above.
      Assume that there are two identifiers $\id_1''$ and $\id_2''$ of~$\ID''$
      and $1 \leq i \leq \arity{\OR''}$ that violate the $v$-impartiality of~$\OR''$.
      Let $(\id_1, \id_1'), (\id_2, \id_2') \in \ID \times \ID'$ be the
      pairs of identifiers used to create $\id_1''$ and $\id_2''$.
  We distinguish on whether $1 \leq i \leq \arity{\OR}$ or $\arity{\OR} < i \leq
  \arity{\OR} + \arity{\OR'}$. In the first case, we deduce that exactly one of
      $T(\id_1).i$ and $T(\id_2).i$ is $v$, so that in particular $\id_1 \neq
      \id_2$.
      Thus, by definition of the order in $\times_\gen$, it is easily seen that,
  because $\id''_1$ and $\id''_2$ are comparable in~$\OR''$,
  $\id_1$ and $\id_2$ must compare in the same way in~$\OR$, contradicting the
  $v$-impartiality of~$\OR$.
  The second case is symmetric.\qedhere
  \end{itemize}
\end{proof}

We now conclude with the proof of Lemma~\ref{lem:noconcat}:

\begin{proof}
  Let us assume by way of contradiction that there is $n \in \mathbb{N}_+$ and a
  \Pgen query~$Q_n$, possibly with $\dupelim$ that captures $\cupcat$. 
  Let $v \neq v'$ be two distinct
  values in $\calD \backslash
  \mathbb{N}$, and consider the singleton
  po-relation $\OR$ containing one identifier of value~$t$ and $\OR'$ containing
  one identifier of value $t'$, where
  $t$ (resp.\ $t'$) are tuples of arity $n$ containing
  $n$ times the value $v$ (resp.\ $v'$). Consider the
  po-database $D$ mapping
  $R$ to $\OR$ and $R'$ to $\OR'$.
  Write $\OR'' \defeq Q_n(D)$.
  By our assumption, as $\OR'' = (\ID'', T'', {<''})$ must be $\OR \cupcat \OR'$,
  it must contain an
  identifier $\id \in \ID''$ such that $T''(\id) = t$ and an identifier $\id'
  \in \ID''$ such that $T''(\id') = t'$.
  Now, as $\OR$ and $\OR'$ are (vacuously) $v$-impartial,
  we know by Lemma~\ref{lem:incomp} that $\OR''$ is $v$-impartial.
  Hence, as $n > 0$, taking $i = 1$, as $t \neq t'$ and exactly one of $t.1$ and
  $t'.1$ is $v$, we know that $\id$ and $\id'$ must be incomparable in~$<''$,
  so there is a possible world of $\OR''$ where $\id'$ precedes $\id$.
  This contradicts the fact that, as we should have $\OR'' = \OR \cupcat \OR'$,
  the po-relation $\OR''$ should have exactly one possible world, namely, 
  $(t, t')$.
\end{proof}

Lemma~\ref{lem:lexconcat} and Lemma~\ref{lem:noconcat} conclude the
proof of Theorem~\ref{thm:incomparable}.
\end{toappendix}

We have now defined a semantics on po-relations for each \PosRA
operator. We define a \emph{\PosRA query} in the expected way, as a
query built from these operators and from relation names. Calling
\emph{schema} a set $\calS$ of relation names and arities, with an
attribute name for each position of each relation, we define a
\emph{po-database} $D$ as having a po-relation $D[R]$ of the correct arity
for each relation name $R$ in~$\calS$. For a po-database $D$ and a \PosRA
query $Q$ we denote by $Q(D)$ the po-relation obtained by evaluating
$Q$ over $D$.

\begin{toappendix}
  \subsection{Proof of Proposition~\ref{prp:repsys}}
\end{toappendix}

\begin{example}
\label{exa:simplegen}

  The po-database $D$ in
Figure~\ref{fig:examplerels} contains information
  about restaurants and hotels in Paris: each po-relation has a total order
  (from top to bottom) according to customer
ratings from a given travel website, and for brevity we do not represent
  identifiers.

  Let $Q \defeq \mathit{Rest} \times_\gen
(\sigma_{\mathit{distr}\neq
\text{``12''}}(\textit{Hotel}))$.
  Its result $Q(D)$ has two possible worlds:\\{%
  \footnotesize\(
(\langle \textup{G}, 8, \textup{M}, 5\rangle , \langle \textup{G}, 8,
\textup{B}, 8\rangle , \langle \textup{TA}, 5, \textup{M}, 5\rangle
, \langle
\textup{TA}, 5, \textup{B},8\rangle ),\quad
(\langle \textup{G}, 8, \textup{M}, 5\rangle , \langle \textup{TA},
5, \textup{M}, 5\rangle ,\langle \textup{G}, 8, \textup{B}, 8\rangle
, \langle \textup{TA},5, \textup{B},8\rangle ).\)}
  In a sense, these \emph{list relations} of hotel--restaurant pairs are
  \emph{consistent} with the order in~$D$: we
do not know how to order two pairs, except when both the hotel
and restaurant compare in the same way.
  The \emph{po-relation} $Q(D)$ is represented in
  Figure~\ref{fig:example} as a Hasse diagram 
  (ordered from bottom to top), again writing tuple values instead of tuple
  identifiers for brevity.

Consider now $Q'\defeq\Pi (\sigma_{\mathit{Rest}.\mathit{distr} =
\mathit{Hotel}.\mathit{distr}} (Q))$, where $\Pi$
  projects out $\mathit{Hotel}.\mathit{distr}$. The possible worlds of~$Q'(D)$
are {\footnotesize \((\langle \textup{G},\textup{B},8\rangle,$ $\langle
\textup{TA},\textup{M},5\rangle)\)} and {\footnotesize \((\langle
\textup{TA},\textup{M},5\rangle,$ $\langle
\textup{G},\textup{B},8\rangle)\)}, intuitively reflecting two
different opinions on the order of restaurant--hotel pairs
in the same district.
  Defining $Q''$ similarly to $Q'$ but replacing $\times_\gen$ by
  $\times_\lex$ in~$Q$, we have $\pw(Q''(D)) =
  {\footnotesize (\langle\textup{G},\textup{B},8\rangle,
  \langle\textup{TA},\textup{M},5\rangle)}$.
\end{example}
We conclude by observing
that we can efficiently evaluate \PosRA queries on po-relations:

\begin{propositionrep}\label{prp:repsys}
  For any fixed \PosRA\ query $Q$, given a po-database $D$, we
  can construct the
  \mbox{po-relation}
  $Q(D)$ in polynomial time in the size of~$D$ (the polynomial degree depends on~$Q$).
\end{propositionrep}

\begin{proof}
We show the claim by a simple induction on the query~$Q$.

\begin{itemize}
  \item If $Q$ is a relation name~$R$, $Q(D)$ is obtained 
in linear time.

\item If $Q$ is a constant expression, $Q(D)$ is obtained in constant
  time.

\item If $Q=\sigma_\psi(Q')$ or $Q=\Pi_{k_1\dots k_p}(Q')$, $Q(D)$ is
  obtained in time linear in $|Q'(D)|$, and we conclude by the induction
    hypothesis.

  \item If $Q=Q_1\cup Q_2$ or $Q=Q_1 \times_\lex Q_2$ or
    $Q=Q_1\times_\gen Q_2$, $Q(D)$ is obtained in time linear in
    $|Q_1(D)|\times|Q_2(D)|$ and we conclude by the induction hypothesis.
    \qedhere
\end{itemize}
\end{proof}

\subparagraph*{Accumulation.}
We now enrich \PosRA with order-aware {\em
accumulation} as the outermost operation, inspired by
\emph{right accumulation} and \emph{iteration} in list programming,
and \emph{aggregation} in
relational databases. We fix a \emph{monoid} $(\calM, \oplus, \epsilon)$ for
accumulation and define:

\begin{definition}
  \label{def:aggregationb}
  For $n \in \mathbb{N}$,
  let
  $h: \calD^n \times
  \mathbb{N}^* \to \calM$ be a function
  called an arity-$n$ \deft{accumulation map}.
  We call $\accum_{h, \oplus}$ an \deft{arity-$n$ accumulation operator};
  its result $\accum_{h, \oplus}(L)$
  on an arity-$n$
  list relation $L = (t_1, \ldots, t_n)$ is 
  $h(t_1, 1) \oplus \cdots \oplus h(t_n,
  n)$, and it is $\epsilon$
  on an
  empty~$L$.
  For complexity purposes, we \emph{always} require accumulation operators to be
  \emph{PTIME-evaluable}, i.e.,
  given any list relation $L$, we can compute $\accum_{h,
  \oplus}(L)$ in PTIME.
\end{definition}

The accumulation operator
maps the tuples with $h$ to~$\calM$, where accumulation is performed
with~$\oplus$.
The map~$h$ may use its second argument to take into
account the absolute position of tuples in~$L$. In what
follows, we
omit the arity of accumulation
when clear from context.

\subparagraph*{The \PosRAaccbold language.} We define the
language \posRAagg that contains all queries of the form
$Q=\accum_{h,\oplus}(Q')$, where $\accum_{h,\oplus}$ is an accumulation
operator and $Q'$ is a \PosRA\ query. 
The \emph{possible results} of $Q$ on a
po-database $D$, denoted $Q(D)$, is the set of results obtained by applying
accumulation to each possible world of $Q'(D)$, namely:

\begin{definition}
  \label{def:aggreg-po}
  For a po-relation $\OR$,
  we define:
  \(\accum_{h, \oplus}(\OR)
  \defeq\{\accum_{h, \oplus}(L) \mid L \in
  \pw(\OR)\}\).
\end{definition}

Of course, accumulation has exactly one result whenever the operator $\accum_{h,
\oplus}$ does not depend on the order of input tuples:
this covers, e.g., the standard sum, min, max, etc.
Hence,
we focus
on accumulation operators which \emph{depend on the order of tuples} (e.g.,
defining
$\oplus$ as concatenation), so there may be more than one accumulation result:

\begin{example}
  \label{exa:aggreg}

  As a first example, let
  $\mathit{Ratings}(\mathit{user},\mathit{restaurant},\mathit{rating})$ be an
  \emph{unordered} po-relation describing
  the numerical ratings given by users to restaurants, where each user rated each restaurant at most
  once. Let $\mathit{Relevance}(\mathit{user})$ be a po-relation
  giving a partially-known ordering of users to indicate the relevance of
  their reviews. We wish to compute a \emph{total rating} for each restaurant
  which is given by the
  sum of its reviews weighted by a
  PTIME-computable weight
  function~$w$. Specifically, $w(i)$ gives a nonnegative weight to the rating of the $i$-th most relevant user.
  Consider 
  \(
  Q_1 \defeq
\accum_{h_1,+}(\sigma_\psi(\mathit{Relevance}\times_{\lex}\mathit{Ratings}))\)
  where we set \(h_1(t,n) \defeq t.\mathit{rating} \times w(n)\), and where
  $\psi$ is the tuple predicate:
  \(
    \mathit{restaurant}=\text{``Gagnaire''}\land\mathit{Ratings}.\mathit{user}=
    \mathit{Relevance}.\mathit{user}
\).
  The query $Q_1$ gives the total rating of \textquote{Gagnaire}, and each
possible world of $\mathit{Relevance}$ may lead to a different accumulation result.

As a second example, consider an unordered po-relation
$\mathit{HotelCity}(\mathit{hotel}, \mathit{city})$ indicating in which
city each hotel is located, and consider a po-relation
$\mathit{City}(\mathit{city})$ which is (partially) ranked by a
criterion such as interest level, proximity, etc. Now consider the
query
  \(
Q_2 \defeq \accum_{h_2,\mathrm{concat}}(
  \Pi_{\mathit{hotel}}(
Q_2'))\), where \(Q_2' \defeq \sigma_{\mathit{City}.\mathit{city} =
\mathit{HotelCity}.\mathit{city}}(
\mathit{City}\times_{\lex}\mathit{HotelCity})\), where
$h_2(t,n)\defeq t$, and where ``$\mathrm{concat}$'' denotes standard string concatenation.
$Q_2$ concatenates the hotel names according to the preference order on the city
where they are located, allowing any possible order between hotels of the same city and
between hotels in incomparable cities.
\end{example}

\section{Possibility and Certainty}\label{sec:posscertdef}
Evaluating a \PosRA or \PosRAacc query $Q$ on a po-database $D$ yields a {\em set of possible
results}: for \PosRAacc, it yields an explicit set of accumulation results, and
for \PosRA, it yields a po-relation that represents a set of possible worlds (list
relations).
The uncertainty among the results may be due to the order of the input relations being partial, due to
uncertainty yielded by the query, or both. In some cases, there is only one
possible result, i.e., a \emph{certain} answer. In other cases, we may wish
to examine multiple \emph{possible} answers. We thus define:

\begin{definition}[(Possibility and Certainty)]
  \label{def:posscert}
  Let $Q$ be a \PosRA query, $D$ be a po-database, and $L$ a list
  relation. The \emph{possibility problem} (\poss) asks 
  if $L \in \pw(Q(D))$, i.e.,
  if $L$ is a possible result.
  The \emph{certainty
  problem} (\cert) asks if $\pw(Q(D)) = \{L\}$, 
  i.e., if $L$ is the only possible result.

  Likewise,
  if $Q$ is a \PosRAacc query with 
 accumulation monoid $\calM$, for a result $v
 \in \calM$, the \poss problem asks whether $v \in Q(D)$, and \cert asks
 whether $Q(D) = \{v\}$.
\end{definition}

\subparagraph*{Discussion.}
For \PosRAacc, our definition follows the usual notion of
possible and certain answers in data integration
\cite{DBLP:conf/pods/Lenzerini02} and incomplete information \cite{Libkin06}.
For \PosRA,
we ask for possibility or certainty of an \emph{entire} output list
relation,
i.e., 
\emph{instance possibility and certainty}~\cite{KO06}. 
We now justify that these notions are useful and discuss more ``local'' alternatives.

First, as we exemplify below,
the output of a query may be certain even for complex
queries and uncertain input. 
It is important to identify such cases and present the user with the certain
answer in full,
like order-by query results in current DBMSs.
Our \cert problem is useful for this task, because we can use it to decide if a
certain output exists, and if yes, we can compute it in PTIME (by choosing any
linear extension).
However, \cert is a challenging problem to solve, because of duplicate values 
(see ``\textsf{Technical difficulties}'' below).

\begin{example}
	\label{ex:certexample}
  Consider the po-database $D$ of Figure~\ref{fig:examplerels} with the po-relations
  $\mathit{Rest}$ and $\mathit{Hotel}_2$.
  To find recommended pairs of hotels and restaurants in the same
  district, the user can write
  $Q \defeq \sigma_{\mathit{Rest}.\mathit{distr}=\mathit{Hotel}_2.\mathit{distr}}
  (\mathit{Rest} \times_{\dir} \mathit{Hotel}_2)$. Evaluating $Q(D)$ yields only
  one possible world, namely, the list relation $(\langle G,8,B,8\rangle,
  \langle \mathit{TA},5,M,5\rangle)$, which is a \emph{certain} result.

  This could also happen with larger input relations. Imagine for example that
  we join hotels and restaurants to find pairs of a hotel and a restaurant
  located in that hotel. The result can be certain if the relative ranking of the
  hotels and of their restaurants agree.
\end{example}

If there is no certain answer, deciding possibility of an instance may be
considered as ``best effort''. It can be useful, e.g., to check if a
list relation (obtained from another source) is consistent with a query result. For
example, we may wish to check if a website's 
ranking of 
hotel--restaurant pairs is
\emph{consistent} with the preferences expressed in its rankings for hotels and
restaurants, to detect when a pair is ranked higher than its components would
warrant.

When there is no overall certain answer, or when we want to check the
possibility of some aggregate property of the relation, we can use a \PosRAacc
query.
In particular, in addition to the applications of Example~\ref{exa:aggreg},
accumulation allows us to encode alternative notions of \poss and \cert for {\em
\PosRA} queries, and to express them as \poss and \cert for \PosRAacc. For example,
instead of possibility or certainty for a full relation, we can express
possibility or certainty of the \emph{location}\footnote{Remember that the {\em existence} of a tuple is not
order-dependent and thus vacuous in our setting.}
of particular tuples of interest:

\begin{example}
\label{ex:variants}
 
With accumulation we can model
  {\emph{position-based selection}} queries. Consider for instance a \emph{top-$k$}
operator on list relations, which retrieves a list relation of the first~$k$~tuples. For a
po-relation, the set of results is all possible such list relations. We can
implement {\em top-$k$} as $\accum_{h_3,\concat}$
with $h_3(t, n)$ being $(t)$ for $n \leq k$ and $\epsilon$ otherwise,
and with $\text{concat}$ being list concatenation. We can similarly compute \emph{select-at-$k$}, i.e., return the
tuple at position $k$, via $\accum_{h_4,\concat}$
with $h_4(t, n)$ being $(t)$ for $n=k$ and $\epsilon$ otherwise.

  Accumulation can also be used for a {\emph{tuple-level comparison}}.
  To check whether 
the first occurrence of a tuple $t_1$ precedes any occurrence of
$t_2$, we define $h_5$ for all $n\in\NN$ by $h_5(t_1,n) \defeq \top$, $h_5(t_2,n) \defeq \bot$ and
$h_5(t,n)\defeq\epsilon$ for $t \neq t_1,t_2$, and a monoid operator $\oplus$ such
that $\top\oplus\top=\top\oplus\bot=\top$, $\bot\oplus
\bot=\bot\oplus\top=\bot$:
  assuming that $t_1$ and $t_2$ are both present, then the result is~$\top$ if
  the first occurrence of $t_1$ precedes any occurrence of $t_2$, and it
  is~$\bot$ otherwise.
\end{example}

We study the complexity of these variants in Section \ref{sec:fpt}. We now
give examples of their use:

\begin{example}
  \label{exa:accumul}
  Consider $Q=\Pi_{\mathit{distr}}(
  \sigma_{\mathit{Rest}.\mathit{distr}=\mathit{Hotel}.\mathit{distr}}
  (\mathit{Rest} \times_{\dir} \mathit{Hotel}))$, which computes ordered
  recommendations of districts including both hotels and restaurants. Using
  accumulation as in Example~\ref{ex:variants}, 
  the user can compute the best district to stay in with
  $Q'=\text{top-}1(Q)$.
  If $Q'$ has a certain answer,
  then 
  there is a dominating
  hotel--restaurant pair in this district, which answers the user's need.
  If there is no certain answer, \poss allows the user to 
  determine 
  the \emph{possible} top-$1$ districts.
	
  We can also use \poss and \cert for \PosRAacc queries to restrict attention to
  \emph{tuples} of interest. If the user hesitates between districts $5$ and $6$, they
  can apply tuple-level comparison  to see
  whether the best pair of district~$5$ may be better (or is always better) than
  that of~$6$.
\end{example}

\subparagraph*{Technical difficulties.} The main challenge to solve \poss and
\cert for a \PosRA query $Q$ on an input po-database $D$ is that the tuple values of
the desired result~$L$ may
occur multiple times in the po-relation $Q(D)$, making it hard to match $L$ and
$Q(D)$.
In other words, even though we may compute the po-relation $Q(D)$ in PTIME (by
Proposition~\ref{prp:repsys}) and present it to the user, they still cannot
easily ``read'' possible and certain answers out of the po-relation:

\begin{example}
  \label{exa:notposet}
  Consider a po-relation $\OR = (\ID, T, {<})$ with
  $\ID = \{\id_{\mathrm{a}}, \allowbreak\id_{\mathrm{b}},
  \allowbreak\id_{\mathrm{c}},
  \allowbreak\id_{\mathrm{d}}, \allowbreak\id_{\mathrm{e}},
  \allowbreak\id_{\mathrm{f}}\}$, with
  $T(\id_{\mathrm{a}}) \defeq \langle\text{Gagnaire}, \text{fr}\rangle$,
  $T(\id_{\mathrm{b}}) \defeq
  \langle\text{Italia}, \text{it}\rangle$,
  $T(\id_{\mathrm{c}}) \defeq \langle\text{TourArgent}, \text{fr}\rangle$,
  $T(\id_{\mathrm{d}}) \defeq
  \langle\text{Verdi}, \text{it}\rangle$, $T(\id_{\mathrm{e}}) \defeq
  \langle\text{Tsukizi}, \text{jp}\rangle$,
  $T(\id_{\mathrm{f}}) \defeq \langle\text{Sola}, \text{jp}\rangle$,
  and with $\id_{\mathrm{a}} < \id_{\mathrm{c}}$, $\id_{\mathrm{b}} <
  \id_{\mathrm{c}}$, $\id_{\mathrm{c}} <
  \id_{\mathrm{e}}$, $\id_{\mathrm{d}} < \id_{\mathrm{e}}$, and
  $\id_{\mathrm{d}} < \id_{\mathrm{f}}$.
  Intuitively, $\OR$ describes a preference relation over restaurants,
  with their name and the type of their cuisine.
  Consider the \PosRA query $Q \defeq \Pi(\OR)$ that projects $\OR$ on type; we
  illustrate the result (with the original identifiers) in Figure~\ref{fig:notposet}.
  Let $L$ be the list
  relation $(\text{it}, \text{fr}, \text{jp}, \text{it}, \text{fr},
  \text{jp})$, and consider
  \poss for $Q$, $\OR$, and $L$.
 
  We have that $L \in \pw(Q(\OR))$, as shown by the linear extension
  $\id_{\mathrm{d}} <' \id_{\mathrm{a}} <' \id_{\mathrm{f}} <' \id_{\mathrm{b}} <'
  \id_{\mathrm{c}} <' \id_{\mathrm{e}}$
  of~$<$. However, this is hard to see, because 
  each of \text{it}, \text{fr}, \text{jp} appears more than once in the
  candidate list as well as in the po-relation; there are thus multiple
  ways to ``map'' the elements of the candidate list to those of the po-relation, and only some of these mappings lead to the existence of a corresponding linear extension.
  It is also challenging to check if $L$ is a
  certain answer: here, it is not, as there are other possible answers,
  e.g.: $(\text{it},
  \text{fr}, \text{fr}, \text{it}, \text{jp}, \text{jp})$.
\end{example}
For \PosRAacc queries, this technical difficulty is even accrued because of the
need to figure out the possible ways in which the desired accumulation result
can be obtained.

\section{General Complexity Results}\label{sec:posscert}
\begin{toappendix}
  We summarize the complexity results of
  Sections~\ref{sec:posscert}--\ref{sec:fpt} in Table~\ref{tab:complexity}.
  \begin{table*}
  \footnotesize
  \caption{Summary of complexity results for possibility and
  certainty}
  \vspace{-1em}
  \label{tab:complexity}
  {
    {
      \renewcommand{\tabcolsep}{4pt}
  \begin{tabularx}{\linewidth}{Xllll@{~~}l}
    \toprule
    &
    {\bf Query} &
    {\bf Restrict.\ on accum.} &
    {\bf Input po-relations} &
    \multicolumn{2}{l}{\bf Complexity} \\
    \midrule
    \poss &
    \PosRA/\posRAagg &
    ---
      & arbitrary
      & NP-c. & (Thm.~\ref{thm:posscomp1}) \\
\cert & \posRAagg & --- & arbitrary
      & coNP-c. & (Thm.~\ref{thm:certcomp}) \\
    \cert & \PosRA & --- &
arbitrary
      & PTIME & (Thm.~\ref{thm:certptime}) \\
    \midrule
    \poss & \Plex & ---
      & totally ordered
      & PTIME & (Thm.~\ref{thm:aggregwCorr}) \\
    \poss & \Plex & ---
      & width $\leq k$
      & PTIME & (Thm.~\ref{thm:aggregw}) \\
    \poss & \Pgen & ---
      & totally ordered
      & NP-c. & (Thm.~\ref{thm:posscompextend1}) \\
    \poss & \Pnoprod & ---
      & ia-width or width $\leq k$
      & PTIME & (Thm.~\ref{thm:aggregnoprod}) \\
    \poss & \Plex/\Pgen & ---
      & 1 total.\ ord., 1 unord.
      & NP-c. & (Thm.~\ref{thm:posscompextended}) \\
    \midrule
\cert & \posRAagg & cancellative &
arbitrary
      & PTIME & (Thm.~\ref{thm:certaintyptimec}) \\
    \poss & \PosRAacc & finite and pos.-invar.
    & totally ordered
      & NP-c. & (Thm.~\ref{thm:possfrihypoposs}) \\
    \cert & \PosRAacc & finite and pos.-invar.
    & totally ordered
      & coNP-c. & (Thm.~\ref{thm:possfrihypocert}) \\
    both & \Plexacc & finite
      & width $\leq k$
    & PTIME & (Thm.~\ref{thm:aggregwa}) \\
    both & \Pnoprodacc & finite and pos.-invar.
      & ia-width or width $\leq k$
      & PTIME & (Thm.~\ref{thm:aggregnoproda}) \\
    \poss & \Pnoprodacc & pos.-invar.
      & unordered
      & NP-c. & (Thm.~\ref{thm:possgri}) \\
    \bottomrule
  \end{tabularx}
}
}
\end{table*}

\end{toappendix}

We have defined the \PosRA and \PosRAacc query languages, and defined and
motivated the problems \poss
and \cert. We now start the study of their complexity, which is the main technical
contribution of our paper. We will always study their \emph{data complexity}\footnote{In \emph{combined complexity}, with $Q$ part of
  the input, \poss and \cert are easily seen to be respectively NP-hard and coNP-hard, by reducing from
  the evaluation of Boolean conjunctive queries (which is NP-hard in data
complexity \cite{AHV-1995}) even without order.},
where the query $Q$ is fixed: in particular, for \PosRAacc, the accumulation map and
monoid, which we assumed to be PTIME-evaluable, is fixed as part of the query,
though it is allowed to be infinite. The input to \poss and \cert for the fixed
query $Q$ is the
po-database $D$ and the candidate result (a list relation for \PosRA, an
accumulation result for \PosRAacc).

\subparagraph*{Possibility.} We start with \poss, which we show
to be NP-complete in general.

\begin{toappendix}
\subsection{Proofs of Theorems~\ref{thm:posscomp1} and~\ref{thm:certcomp}}
\end{toappendix}

\begin{theoremrep}\label{thm:posscomp1} The \poss\ problem is in NP for any \PosRA
  or \PosRAagg query. Further, there exists a \PosRA query and a \PosRAagg query for which the \poss problem is
  NP-complete.
\end{theoremrep}

\begin{proofsketch}
 The membership for \PosRA in NP is clear: guess a linear extension and check that
  it realizes the candidate possible result. 
  For hardness, as in previous work
  \cite{warmuth1984complexity}, we reduce from the UNARY-3-PARTITION
problem~\cite{garey-johnson}: given a number $B$ and $3m$ numbers written in unary,
  decide if they can be partitioned in triples that all sum to~$B$.
We reduce this to \poss for the identity \PosRA query, on an arity-1 input
  po-relation where each input number~$n$ is represented as a chain of $n+2$
  elements. The first and last elements of each chain are respectively called
  start and end markers, and elements of distinct
  chains are pairwise
  incomparable. The candidate
  possible world $L$ consists of $m$ repetitions of the following pattern: three start markers, $B$
  elements, three end markers. A linear extension achieves $L$ iff the
  triples matched by~$<$ to each copy of the pattern are a solution to
  UNARY-3-PARTITION, hence \poss for $Q$ is NP-hard. 
  This implies hardness for \PosRAagg, when accumulating
  with the identity map and concatenation (so that any list relation is mapped
  to itself).
\end{proofsketch}

In fact, as we will later point out,
hardness holds even for quite a restrictive setting, with a more intricate proof: see
Theorem~\ref{thm:posscompextend1}.

\subparagraph*{Certainty.} We show that \cert is coNP-complete for \PosRAacc:

\begin{theoremrep}\label{thm:certcomp}
	The \cert problem is in coNP for any \PosRAagg query, and there is a
        \PosRAagg query for which it is coNP-complete.
\end{theoremrep}

\begin{proofsketch}
	Again, membership is immediate. We show hardness of \cert by studying a
        \PosRAacc query $Q_{\a}$ that checks if two input
        po-relations $\OR$ and $\OR'$ have some common possible world: $Q_{\a}$
        does so
        so by testing if one can alternate between elements of $\OR$ and
        $\OR'$ with the same label, using accumulation in the 
	transition monoid of a deterministic finite automaton.
        We show hardness of \poss for $Q_{\a}$ (as in the previous result), and
        further ensure that $Q_{\a}$ always has at most two possible
        accumulation results, no matter the input. Hence, \poss for~$Q_{\a}$ reduces
        to the negation of \cert for $Q_{\a}$, so that \cert is also hard.
\end{proofsketch}

\begin{toappendix}
  We first show the upper bounds:

\begin{proposition}
  For any \PosRA or \PosRAagg{} query $Q$, \poss for $Q$ is in NP and \cert for $Q$ is in co-NP.
\end{proposition}

\begin{proof}
  We show the results for \PosRAagg{} queries, as the same clearly holds for
  \PosRA queries.
To show the NP membership of \poss, evaluate in PTIME the query without
accumulation using Proposition~\ref{prp:repsys}, yielding a
po-relation~$\OR$.
Now, guess a total order of $\OR$, checking in PTIME
that it is compatible with the comparability relations of $\OR$.
If there is no accumulation function, check that it achieves the
candidate result. Otherwise, evaluate the accumulation (in PTIME as the
accumulation operator is PTIME-evaluable),
and check that the correct result is obtained.
  
To show the co-NP membership of \cert, follow the same reasoning but guessing
an order that achieves a result different from the candidate result.
\end{proof}

We now show the lower bounds. We first show the lower bound of
Theorem~\ref{thm:posscomp1} for \poss on a \PosRA query. In fact, when arbitrary
po-relations are allowed, \poss is already hard for a trivial query: we will
use non-trivial \PosRA queries later to show hardness of \poss on restricted input
po-relations (cf.\ Theorem~\ref{thm:posscompextend1} and
Theorem~\ref{thm:posscompextended}).

\begin{proposition}
  \label{prp:posshardsimple}
  There is a \PosRA query $Q$ such that the \poss problem for $Q$ is NP-hard.
\end{proposition}

This result can also be shown from existing
work~\cite{warmuth1984complexity} about the complexity of the so-called
\emph{shuffle problem}: given a string $w$ and a tuple of strings $s_1, \ldots,
s_n$ on the fixed alphabet $A=\{a, b\}$, decide whether there is an interleaving
of $s_1, \ldots, s_n$ which is equal to~$w$. It is easy to see that there is a
reduction from the shuffle problem to the \poss problem, by representing each string
$s_i$ as a totally ordered relation $L_i$ of tuples labeled~$a$ and~$b$ that
code the string, letting $\OR$ be the po-relation which is the union of the
$L_i$, and asking if the totally ordered relation that codes~$w$ is a possible
world of the identity query on the po-relation~$\OR$. Hence, as the shuffle
problem is shown to be NP-hard in~\cite{warmuth1984complexity}, this implies the
same for \poss. We nevertheless give a self-contained proof of
Proposition~\ref{prp:posshardsimple}, because we will be extending this proof 
to show different results in Theorem~\ref{thm:posscompextend1}.
We note that our proof is in fact very similar to the hardness
proof of~\cite{warmuth1984complexity}; see specifically Lemma~3.2
of~\cite{warmuth1984complexity}.

\begin{proof}
  The reduction is from the UNARY-3-PARTITION
  problem, which is NP-hard~\cite{garey-johnson}: given $3m$ integers $E = (n_1, \ldots, n_{3m})$
  written in unary (not necessarily distinct) and a number~$B$, decide if the integers can be partitioned
  in triples such that the sum of each triple is~$B$. We reduce an instance
  $\mathcal{I} = (E, B)$ of UNARY-3-PARTITION to a \poss instance in PTIME. We
  use the trivial identity query $Q \defeq R$, where $R$ is a relation name of
  arity~$1$. We will use an input po-database $D$ that maps the relation name $R$ to
  a po-relation $\OR$, and we now
  describe how to construct the input relation $\OR = (\ID, T, <)$ in PTIME from
  the UNARY-3-PARTITION instance.

  We set $\ID$ to be $\{\id_i^j \mid 1 \leq i \leq 3m, 1 \leq j \leq n_i + 2 \}$:
  this is constructible in PTIME, because the input to UNARY-3-PARTITION is
  written in unary. The relation $\OR$ will have arity~$1$ and domain $\{\s, \n,
  \e\}$, where $\s$, $\n$ and $\e$ are three arbitrary distinct values chosen
  from~$\calD$ (standing for ``start'', ``inner'', and ``end''). We set
  $T(\id_i^1) \colonequals \s$ and
  $T(\id_i^{n_i+2}) \colonequals \e$
  for all $1 \leq i \leq 3m$, and set
  $T(\id_i^j) \colonequals \n$ in all other cases, i.e., for all $1 \leq i \leq 3m$ and all
  $2 \leq j \leq n_i + 1$. Last, we define the order relation~$<$ by letting
  $\id_i^j < \id_i^{j'}$ for all $1 \leq i \leq 3m$ and $1 \leq j < j' \leq
  n_i+2$. This implies in particular that, for all $1 \leq i, i' \leq 3m$, for all $1 \leq
  j \leq n_i+2$ and $1 \leq j' \leq n_{i'}+2$, if $(i,j) \neq (i', j')$, then
  the elements $\id_i^j$ and $\id_{i'}^{j'}$ are comparable by~$<$ iff $i = i'$.
  
  Now, let $L'$ be the list relation $\s^3 \n^B \e^3$, where
  exponents denote repetition of tuples, and let $L$ be the list relation
  $(L')^m$, which we will use as a candidate possible world. We now claim that the UNARY-3-PARTITION
  instance defined by $E$ and $B$ has a solution iff $L \in \pw(\OR)$, which
  concludes the proof because the reduction is clearly in PTIME.

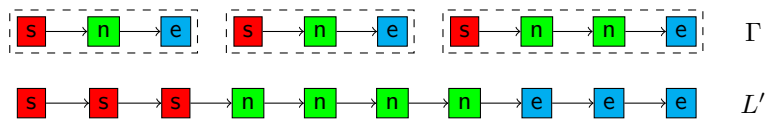
\begin{figure}
\centering
  \begin{tikzpicture}[scale=.95]
    \node (S') at (11, 1.5) {$\OR$};

    \node[draw,fill=red] (R1s) at (1, 1.5) {$\s$};
    \node[draw,fill=green] (R1n) at (2, 1.5) {$\n$};
    \node[draw,fill=cyan] (R1e) at (3, 1.5) {$\e$};
    \draw[->] (R1s) -- (R1n);
    \draw[->] (R1n) -- (R1e);

    \node[draw,fill=red] (R2s) at (4, 1.5) {$\s$};
    \node[draw,fill=green] (R2n) at (5, 1.5) {$\n$};
    \node[draw,fill=cyan] (R2e) at (6, 1.5) {$\e$};
    \draw[->] (R2s) -- (R2n);
    \draw[->] (R2n) -- (R2e);

    \node[draw,fill=red] (R3s) at (7, 1.5) {$\s$};
    \node[draw,fill=green] (R3n1) at (8, 1.5) {$\n$};
    \node[draw,fill=green] (R3n2) at (9, 1.5) {$\n$};
    \node[draw,fill=cyan] (R3e) at (10, 1.5) {$\e$};
    \draw[->] (R3s) -- (R3n1);
    \draw[->] (R3n1) -- (R3n2);
    \draw[->] (R3n2) -- (R3e);

    \draw[dashed] (0.7, 1.2) rectangle (3.3, 1.8);
    \draw[dashed] (3.7, 1.2) rectangle (6.3, 1.8);
    \draw[dashed] (6.7, 1.2) rectangle (10.3, 1.8);

    \node (L') at (11, .5) {$L'$};

    \node[draw,fill=red] (R1sl) at (1, .5) {$\s$};
    \node[draw,fill=red] (R1nl) at (2, .5) {$\s$};
    \node[draw,fill=red] (R1el) at (3, .5) {$\s$};
    \draw[->] (R1sl) -- (R1nl);
    \draw[->] (R1nl) -- (R1el);

    \node[draw,fill=green] (R2sl) at (4, .5) {$\n$};
    \node[draw,fill=green] (R2nl) at (5, .5) {$\n$};
    \node[draw,fill=green] (R2el) at (6, .5) {$\n$};
    \draw[->] (R2sl) -- (R2nl);
    \draw[->] (R2nl) -- (R2el);

    \node[draw,fill=green] (R3sl) at (7, .5) {$\n$};
    \node[draw,fill=cyan] (R3n1l) at (8, .5) {$\e$};
    \node[draw,fill=cyan] (R3n2l) at (9, .5) {$\e$};
    \node[draw,fill=cyan] (R3el) at (10, .5) {$\e$};
    \draw[->] (R3sl) -- (R3n1l);
    \draw[->] (R3n1l) -- (R3n2l);
    \draw[->] (R3n2l) -- (R3el);

    \draw[->] (R1el) -- (R2sl);
    \draw[->] (R2el) -- (R3sl);
  \end{tikzpicture}
  \caption{Example for the proof of Proposition~\ref{prp:posshardsimple}, with $E =
  (1, 1, 2)$ and $B = 4$.}
  \label{fig:gridpic2}
\end{figure}

  To see why the reduction is correct, we first show that, if $E$ is a positive
  instance of \mbox{UNARY-3-PARTITION}, then there is a linear extension $<'$ of~$<$
  which witnesses that $L \in \pw(\OR)$. Indeed, consider a
  3-partition $\mathbf{s} = (s^i_1, s^i_2, s^i_3)$ for $1 \leq i \leq m$, with $n_{s^i_1} +
  n_{s^i_2} + n_{s^i_3} = B$ for all $1 \leq i \leq m$, and each integer of $\{1, \ldots, 3m\}$ occurring exactly
  once in $\mathbf{s}$. We can realize $L$ from~$\mathbf{s}$, picking
  successively the following for $1
  \leq i \leq m$ to realize~$L'$: the tuples
  $\id^{s^i_p}_1$ for $1 \leq p \leq 3$ that are mapped to~$\s$ by~$T$; the tuples
  $\id^{s^i_p}_{j_p}$ for $1
  \leq p \leq 3$ and $2 \leq j_p \leq n_{s^i_p} + 1$ that
  are mapped to~$\n$ by~$T$ (hence, $B$ tuples in total, by the condition
  on~$\mathsf{s}$); the tuples
  $\id^{s^i_p}_{n_{s^i_p}+2}$ for $1 \leq p \leq 3$ that are mapped to~$\e$
  by~$T$.

  Conversely, we show that, if there is a linear extension $<'$ of~$<$ which
  witnesses that $L \in \pw(\OR)$, then we can build a 3-partition
  $\mathbf{s} = (s^i_1, s^i_2, s^i_3)$ for $1 \leq i \leq m$ which satisfies the
  conditions above. To see why, we first observe that, for each $1 \leq i \leq
  m$, for the $i$-th occurrence of the sublist $L'$ in~$L$, there
  must be three distinct values $s^i_1, s^i_2, s^i_3$, such that the elements
  of~$\ID$ which occur in~$<'$ at the positions of the
  value~$\n$ in this occurrence of~$L'$
  are precisely the elements of the form $\id_{s^i_p}^{j_p}$ for
  $1 \leq p \leq 3$ and $1 \leq j_p \leq n_{s^i_p}+1$. Indeed, we show this
  claim for increasing values of~$i$, from $i = 1$ to $i = m$. For the $i$-th
  occurrence of~$L'$ for some $1 \leq i \leq m$, we define $s^i_1,
  s^i_2, s^i_3$, such that the elements $\s^3$ in this occurrence of~$L'$
  are mapped to $\id^1_{s^i_1}, \id^1_{s^i_2}, \id^1_{s^i_3}$: they must indeed
  be mapped to such elements because they are the only ones mapped to~$\s$
  by~$T$. Now, the elements of the form $\id_{s^i_p}^{j_p}$ for
  $1 \leq p \leq 3$ and $1 \leq j_p \leq n_{s^i_p}+1$ are the only ones that can
  be enumerated, because are the only ones that have not been enumerated yet, and they have no
  ancestors mapped to~$\s$ by~$T$ that have not been enumerated. Further, all
  these elements must be enumerated, because this is the only possible way for~$<'$ to be
  able to enumerate $\e$-labeled elements, namely, the
  $\id_{s^i_p}^{n_{s^i_p}+2}$ for $1 \leq p \leq 3$. Now that we have defined
  the 3-partition $\mathbf{s}$, it is clear by definition of a linear extension
  that all numbers in~$\mathbf{s}$ must be distinct. Further, as $<'$
  achieves~$L'$, by considering each occurrence of~$L'$, we know that, for $1
  \leq i \leq m$, we have $s^i_1 + s^i_2 + s^i_3 = B$. Hence, $\mathbf{s}$
  witnesses that $E$ is a positive instance to the UNARY-3-PARTITION problem.
  
  This establishes the correctness of the reduction, and concludes the proof.
\end{proof}

To show the lower bound for \PosRAacc, we show a general lemma about reducing
\poss and \cert for \PosRA queries to the same problems for \PosRAacc queries:

\begin{lemma}
  \label{lem:addacc}
  For any arity $k \in \mathbb{N}$,
  there exists an infinite and cancellative monoid $(\calM_k, \oplus, \epsilon)$ (see
  Definition~\ref{def:cancellative}), a position-invariant and arity-$k$
  accumulation map $h_k$ (see Definition~\ref{def:ri}), and a PTIME-evaluable accumulation operator
  $\accum_{h_k,\oplus}$ such that, for any \PosRA query~$Q$ of arity~$k$, the \poss and \cert
  problems for $Q$ are respectively equivalent to the \poss and \cert problems for the
  \PosRAacc query $\accum_{h_k, \oplus} Q$.
\end{lemma}

\begin{proof}
  Fix $k \in \mathbb{N}$. We will use the identity accumulation
  operator. Consider the monoid $(\calM_k, \oplus, \epsilon)$ defined as
  follows: $\calM_k$ is
  the list relations on $\calD^k$, that is, the finite sequences of
  elements of $\calD^k$, the neutral element $\epsilon$ is the empty list, and
  the associative operation $\oplus$ is the concatenation of list relations. This clearly defines a
  monoid, and it is clearly cancellative. Let $h_k$ be the
  position-invariant accumulation map that maps any tuple $t$ to the singleton
  list relation $\singleton{t}$ containing precisely one tuple with that value.
  
  Now, consider the query $Q' \defeq
  \accum_{h_k, \oplus} (Q)$. Let $D$ be an po-database. It is clear
  that any list relation $L$ is a possible world of 
  $Q(D)$ iff $L$ is a possible result of $Q'(D)$: in other words, we have
  $\pw(Q(D)) = \pw(Q'(D))$. This clearly ensures that \poss and \cert for $Q$
  are respectively equivalent to \poss and \cert for~$Q'$.
\end{proof}

We deduce:

\begin{corollary}
  There is a \PosRAacc query $Q$ such that the \poss problem for~$Q$ is NP-hard.
\end{corollary}

What remains is to show the hardness result for \cert and \PosRAacc. This result
is more complex, and is presented (in a slightly stronger form) 
as
Theorem~\ref{thm:possfrihypocert}
in
Appendix~\ref{sec:restracc}.

\end{toappendix}

For \PosRA queries, however, we show
that \cert is in PTIME.
As we will see later, this follows from the tractability of \cert for \PosRAacc on
\emph{cancellative monoids} (Theorem~\ref{thm:certaintyptimec}).

\begin{toappendix}
\subsection{Proof of Theorem~\ref{thm:certptime}}
\end{toappendix}

\begin{theoremrep}\label{thm:certptime}
  \cert is in PTIME for any \PosRA\ query.
\end{theoremrep}

\begin{proof}
  Let $Q$ be the \PosRA query of interest,
  and let $k \in \mathbb{N}$ be its arity: let $(\calM_k, \oplus)$ be the
  cancellative monoid (see Definition~\ref{def:cancellative}) and $h_k$ be the accumulation map obtained from Lemma~\ref{lem:addacc}. By
  Lemma~\ref{lem:addacc}, we know that \cert for $Q$ is equivalent to \cert for
  the \PosRAacc query $Q' \defeq \accum_{h,\oplus}$, which is clearly
  constructible in PTIME. 

  Now, by Theorem~\ref{thm:certaintyptimec} (proven in
Appendix~\ref{sec:proof-cancellative-monoids}), we know that the \cert problem is in
PTIME for $Q'$, because it performs accumulation in a cancellative monoid.
  Hence, using the PTIME reduction above, we deduce that the \cert problem for $Q$ is in PTIME
as well.
\end{proof}

We next identify further tractable cases, first for \PosRA\ and then for \PosRAacc.

\section{Tractable Cases for POSS on PosRA Queries}\label{sec:fpt2}
We show that \poss is tractable for \PosRA queries if we restrict the allowed
operators and if we bound some order-theoretic parameters of the input po-database,
such as \emph{poset width}.

We call \Plex the fragment of \PosRA that disallows the
$\times_{\dir}$ operator, but allows all other operators (including
$\times_{\lex}$). We also define \Pgen that disallows $\times_{\lex}$ but
not $\times_{\dir}$.

\subparagraph*{Totally ordered inputs.}
\begin{toappendix}
\subsection{Totally Ordered Inputs}
\end{toappendix}
We start by the natural case where the individual
po-relations are \emph{totally ordered}, i.e., their order relation is a total
order (so they actually represent a list relation).
This applies
to situations where we integrate data from multiple sources that are
certain (totally ordered), and where uncertainty only results
from the integration query
(so that the result may still have exponentially many
possible worlds, e.g., the \emph{union} of two total orders has
exponentially many possible interleavings). 
In a sense, the $\times_\dir$ operator is the one introducing the most
uncertainty and ``complexity'' in the result, so we
consider the fragment \Plex of \PosRA queries without
$\times_{\dir}$, and show:

\begin{toappendix}
  \subsubsection{Tractability Result: Proof of
  Theorems~\ref{thm:aggregwCorr} and~\ref{thm:aggregw}}
  \label{sec:totaltract}
The point of
restricting to \Plex queries is that they can only make the
width increase in a way that depends on the \emph{width} of the input relations, but
not on their \emph{size}:
\begin{propositionrep}\label{prp:lexwidth}
  Let $k \geq 2$ and $Q$ be a \Plex query. Let $k' =
  k^{\card{Q}+1}$.
  For any po-database~$D$ of width~$\leq k$, the po-relation $Q(D)$ has width $\leq k'$.
\end{propositionrep}

\begin{proof}
  We prove by induction on the \Plex query $Q$ that one can compute a bound on
  the width of the output of the query as a function of the bound $k$ on the
  width of the inputs. For the base cases:

  \begin{itemize}
    \item Input po-relations have width $\leq k$.
    \item Constant po-relations (singletons and constant chains) have width~$1$.
  \end{itemize}

  For the induction step:

  \begin{itemize}
    \item Given two po-relations $\OR_1$ and $\OR_2$ with bounds $k_1$ and $k_2$, their
      union $\OR_1 \cupgen \OR_2$ clearly has bound $k_1 + k_2$, as any antichain in the union must be
      the union of an antichain of $\OR_1$ and of an antichain of $\OR_2$.
    \item Given a po-relation $\OR_1$ with bound $k_1$, applying a projection or
      selection to $\OR_1$ cannot make the width increase.
    \item Given two po-relations $\OR_1$ and $\OR_2$ with bounds $k_1$ and $k_2$, their
      product $\OR \defeq \OR_1 \times_{\lex} \OR_2$ has bound $k_1 \cdot k_2$. To show
      this, consider any set $A$ of $\OR$ containing strictly more than $k_1
      \cdot k_2$ identifiers, which we see as
      pairs of an identifier of $\OR_1$ and an identifier of $\OR_2$. It is immediate
      that one of the following must hold:

      \begin{enumerate}
        \item Letting $S_1 \defeq \{u \mid \exists v, (u, v) \in A\}$, we have
          $\card{S_1} > k_1$
        \item There exists $u$ such that, letting $S_2(u) \defeq \{v \mid (u, v) \in
          A\}$, we have $\card{S_2} > k_2$
      \end{enumerate}

      Informally, when putting $> k_1 \cdot k_2$ values in buckets (the value of
      their first component), either $> k_1$ different buckets are used, or
      there is a bucket containing $> k_2$ elements.

      In the first case, as $S_1$ is a subset of identifiers of $\OR_1$ of cardinality
      $> k_1$ and $\OR_1$ has width~$k_1$, it cannot be an antichain, so it must
      contain two comparable elements $u_1 < u_2$, so that, considering $v_1$
      and $v_2$ such that $a_1 = (u_1, v_1)$ and $a_2 = (u_2, v_2)$ are in~$A$,
      we have by definition of $\times_\lex$ that $a_1 <_\OR a_2$, so that
      $A$ is not an antichain of~$\OR$.

      In the second case, as $S_2(u)$ is a subset of identifiers of~$\OR_2$ of
      cardinality $> k_2$ and $\OR_2$ has width $k_2$, it cannot be an antichain,
      so it must contain two comparable elements $v_1 < v_2$. Hence, considering
      $a_1 = (u, v_1)$ and $a_2 = (u, v_2)$ which are in~$A$, we have
      $a_1 <_\OR
      a_2$, and again $A$ is not an antichain of $\OR$.

      Hence, we deduce that no set of cardinality $> k_1 \cdot k_2$
      of~$\OR$ is an
      antichain, so that $\OR$ has width $\leq k_1 \cdot k_2$, as desired.
    \end{itemize}

  Letting $o$ be the number of product operators in~$Q$ plus the number of
  union operators, it is now clear that we can take $k' = k^{o + 1}$.
  Indeed, po-relations with no product or union operators have width at most $k$
  (using that $k \geq 1$). As projections and selections do not change the
  width, the only operators to consider are product and union. If $Q_1$ has
  $o_1$ operators and $Q_2$ has $o_2$ operators, bounding by induction the width
  of $Q_1(D)$ to be $k^{o_1 + 1}$ and $Q_2(D) = k^{o_2 + 1}$, for $Q = Q_1 \cup
  Q_2$, the number of operators is $o_1 + o_2 + 1$, and the new bound is $k^{o_1
  + 1} + k^{o_2 + 1}$, which as $k \geq 2$ is less than $k^{o_1 + 1 + o_2 + 1}$,
  that is, $k^{(o_1 + o_2 + 1) + 1}$. For~$\times_\lex$, we proceed in the
  same way and directly obtain the $k^{(o_1 + o_2 + 1) + 1}$ bound. Hence, we
  can indeed take $k' = k^{\card{Q}+1}$.
\end{proof}

From this, we will deduce \poss is tractable for
\Plex queries when the input po-database consists of relations of
bounded width.
We
  now
  prove Theorem~\ref{thm:aggregw}, which clearly generalizes
  Theorem~\ref{thm:aggregwCorr}. We will prove both the result for \Plex queries
  and its extension to \Plexacc queries with finite accumulation
  (Theorem~\ref{thm:aggregwa}).
\end{toappendix}

\begin{theoremrep}\label{thm:aggregwCorr}
\poss is in PTIME for \Plex
queries if
input po-relations are totally ordered.
\end{theoremrep}

In fact, we can show tractability for
relations of bounded
\emph{poset width}:

\begin{definition}[\cite{schroder2003ordered}]
  \label{def:width}
  An \deft{antichain} in a po-relation~$\OR=(\ID,T,<)$ is a set
  $A\subseteq\ID$ of pairwise
  incomparable tuple identifiers.
  The
  \deft{width} of $\OR$ is the size of its largest
  antichain.
  The \deft{width} of a po-database is the maximal width of its po-relations.
\end{definition}

In particular, totally ordered po-relations have width~$1$, and
unordered po-relations have a width equal to their size (number of
tuples); the width of a po-relation can be computed in
PTIME~\cite{fulkerson1955note}.
Po-relations of low width are a common practical case: they cover, for instance,
po-relations that are totally ordered except for a few tied identifiers at each
level. We show:

\begin{theoremrep}\label{thm:aggregw}
  For any fixed $k \in \NN$ and fixed \Plex query $Q$,
  the \poss problem for~$Q$ is in PTIME when
  all po-relations of the input
  po-database have width $\leq k$.
\end{theoremrep}

\begin{proofsketch}
  As $\times_{\dir}$ is disallowed,
  we can show that the po-relation $\OR \defeq Q(D)$
  has width~$k'$ depending only on $k$ and the query~$Q$ (but not on~$D$). We can then compute in PTIME a
  \emph{chain partition} of $\OR$
  \cite{dilworth1950decomposition,fulkerson1955note}, namely, a decomposition
  of~$\OR$ in totally ordered chains, with additional order constraints between
  them. This allows us to 
  apply a dynamic algorithm to decide \poss: the state of the algorithm is the position on
  the chains. The number of states is polynomial with degree $k'$, which is
  a constant when $Q$ and $k$ are fixed.
\end{proofsketch}

\begin{toappendix}
  Let $\OR \defeq Q(D)$ be the po-relation obtained by evaluating the query $Q$
  of interest on the input po-database $D$, excluding the
  accumulation operator if any (so we are evaluating a \Plex query).
  We can compute this in PTIME using
  Proposition~\ref{prp:repsys}. Letting $k'$ be the constant (only depending on~$Q$ and $k$)
  given by Proposition~\ref{prp:lexwidth}, we know that
  $w(\OR) \leq k'$.

  We first show the tractability of \poss and \cert for \Plexacc queries with
  finite accumulation, which amounts to applying directly a finite accumulation
  operator to $\OR$. We then deal with \Plex queries, which amounts to solving
  directly \poss and \cert on the po-relation~$\OR$.

  \subparagraph*{\Plexaccbold queries with finite accumulation.}
  It suffices to show the following rephrasing of the result:

  \begin{theorem}
    \label{thm:aggregwb}
    For any constant $k' \in \mathbb{N}$,
    and accumulation operator $\accum_{h, \oplus}$ with
    finite domain, we can compute in PTIME,
    for any input po-relation $\OR$ such that $w(\OR) \leq k'$,
    the set $\accum_{h, \oplus}(\OR)$.
  \end{theorem}

  Indeed, by what precedes, we can assume that the query has already been
  evaluated to a po-relation; further, once the possible results are determined, it is immediate
  to solve possibility and certainty.

  To take care of this task, we need the following notions:

  \begin{definition}
    A \deft{chain partition} of a poset $P$ is a partition $\mathbf{\Lambda} =
    (\Lambda_1,
    \ldots, \Lambda_n)$ of the elements of $P$, i.e., $P = \Lambda_1 \sqcup \cdots \sqcup
    \Lambda_n$, such that each $\Lambda_i$ is a total order.
    (However, $P$ may feature comparability relations not present in the
    $\Lambda_i$,
    i.e., relating elements in $\Lambda_i$ to elements in $\Lambda_j$ for $i \neq j$.) The
    \deft{width} of the partition $\mathbf{\Lambda} = (\Lambda_1, \ldots,
    \Lambda_n)$ is~$n$.
  \end{definition}

  \begin{definition}
    Given a poset $P$, an \deft{order ideal} of~$P$ is a subset $S$ of $P$ such that,
    for all $x, y \in P$, if $x < y$ and $y \in S$ then $x \in S$.
  \end{definition}

  We also need the following known results:

  \begin{theorem}[\cite{dilworth1950decomposition}]
    \label{thm:dilworth}
    Any poset $P$ has a chain partition of width $w(P)$.
  \end{theorem}

  \begin{theorem}[\cite{fulkerson1955note}]
    \label{thm:fulkerson}
    For any poset $P$, we can compute in PTIME a chain partition of $P$ of
    minimal width.
  \end{theorem}

  We now prove Theorem~\ref{thm:aggregwb}:

  \begin{proof}[Proof of Theorem~\ref{thm:aggregwb}]
    Consider a po-relation $\OR=(ID,T,{<})$, with underlying
    poset~$P=(ID,{<})$. Using Theorems~\ref{thm:dilworth}
    and Theorem~\ref{thm:fulkerson}, compute in PTIME a chain partition $\mathbf{\Lambda}$
    of $P$ of width $k'$.
    For $1 \leq i \leq k'$,
    write $n_i \defeq \card{\Lambda_i}$, and for $0 \leq j \leq n_i$, write
    $\Lambda_i^{\leq
    j}$ to denote the subset of $\Lambda_i$ containing the first $j$ elements of the
    chain (in particular $\Lambda_i^{\leq 0} = \emptyset$).

    We now consider all vectors of the form $(m_1, \ldots, m_{k'})$, with
    $0 \leq m_i \leq n_i$, of which there are polynomially many (there are $\leq
    \card{\OR}^{k'}$, where $k'$ is constant). To each such vector $\mathbf{m}$ we
    associate the subset $s(\mathbf{m})$ of $P$ consisting of
    $\bigsqcup_{i=1}^{k'} \Lambda_i^{\leq m_i}$.

    We call such a vector $\mathbf{m}$ \deft{sane} if $s(\mathbf{m})$ is an
    order ideal. (While $s(\mathbf{m})$ is always an order ideal of the subposet
    of the comparability relations within the chains, it may not be an order
    ideal overall because of
    the additional comparability relations across the chains that may be
    featured in $P$.) For each vector
    $\mathbf{m}$, we can check in PTIME whether it is sane, by materializing
    $s(\mathbf{m})$ and checking that it is an ideal for each comparability
    relation (of which there are $O(\card{P}^2)$).

    By definition, for each sane vector $\mathbf{m}$, $s(\mathbf{m})$
    is an ideal. We now observe that the converse is true, and that for every
    ideal $S$ of $P$, there is a sane vector $\mathbf{m}$ such that
    $s(\mathbf{m}) = S$. To see why, consider an ideal $S$, and determine for
    each chain $\Lambda_i$ the last element of the chain present in the ideal; let
    $m_i$ be its position in the chain. $S$ then does not include any
    element of $\Lambda_i$ at a later position, and because $\Lambda_i$ is a chain it must
    include all elements before, hence, $S \cap \Lambda_i = \Lambda_i^{\leq m_i}$. As
    $\mathbf{\Lambda}$ is a chain partition of~$P$, this entirely determines $S$. Thus we
    have indeed $S = s(\mathbf{m})$, and the fact that $s(\mathbf{m})$ is sane
    is witnessed by~$S$.

    For any sane vector $\mathbf{m}$, we now write $t(\mathbf{m}) \defeq \accum_{h,
    \oplus}(T(s(\mathbf{m})))$ (recall that $T$ maps elements of the
    poset to tuples, and can therefore naturally be extended to map
    sub-posets to sub-po-relations). This is a subset of the
    accumulation domain
    $\calM$ (since the latter is finite, this subset is of constant
    size). It is immediate that $t((0, \ldots, 0)) = \{\epsilon\}$, the neutral
    element of the accumulation monoid, and that
    $t((n_1, \ldots, n_{k'})) = \accum_{h, \oplus}(\OR)$ is our desired answer. Denoting by $e_i$ the vector
    consisting of $n-1$ zeroes and a $1$ at position $i$, for $1 \leq i
    \leq k'$,
    we now observe that, for any sane vector $\mathbf{m}$, we have:
    \begin{equation}
      \label{eqn:tprop}
      t(\mathbf{m}) = \bigcup_{1 \leq i \leq k'} \left\{
        v \oplus h\left(T(\Lambda_i[m_i]), \sum_{i'} m_{i'}\right)
      \:\middle|\: v \in t(\mathbf{m} - e_i)\right\}
    \end{equation}
    where the operator ``$-$'' is the component-by-component integer difference
    on tuples,
    and where we define $t(\mathbf{m} - e_i)$ to be $\emptyset$ if $\mathbf{m} -
    e_i$ is not
    sane or if one of the coordinates of $\mathbf{m} - e_i$ is $< 0$.
    Equation~\ref{eqn:tprop} holds because
    any linear extension of $s(\mathbf{m})$ must end with one of the maximal
    elements of $s(\mathbf{m})$, which must be one of the $\Lambda_i[m_i]$ for $1 \leq
    i \leq m$ such that $m_i \geq 1$, and the preceding elements must be a linear
    extension of the ideal where this element was removed (which must be an
    ideal, i.e., $\mathbf{m} - e_i$ must be sane, otherwise the removed
    $\Lambda_i[m_i]$ was not actually maximal because it was comparable to (and
    smaller than) some $\Lambda_j[m_j]$ for $j \neq i$). Conversely, any sequence
    constructed in this fashion is indeed a linear extension. Thus, the possible
    accumulation results are computed according to this characterization of the
    linear extensions. We store with each possible accumulation result a
    witnessing totally ordered relation from which it can be computed in PTIME,
    namely, the linear extension prefix considered in the previous reasoning,
    so that we can use the PTIME-evaluability of the underlying monoid to ensure
    that all computations of accumulation results can be performed in PTIME.

    This last equation allows us to compute $t(n_1, \ldots, n_{k'})$ in PTIME by a
    dynamic algorithm, enumerating the vectors (of which there are polynomially
    many) in lexicographical order, and computing their image by $t$ in PTIME
    according to the equation above, from the base case $t((0, \ldots, 0)) =
    \epsilon$ and from the previously computed values of $t$. Hence, we have
    computed $\accum_{h, \oplus}(\OR)$ in PTIME, which concludes the proof.
  \end{proof}

  \subparagraph*{\Plexbold queries.}
  First note that, for queries with no accumulation, we cannot reduce \poss and
  \cert to the case with accumulation, because the monoid of tuples under
  concatenation does not satisfy the hypothesis of finite accumulation. Hence, we
  need specific arguments to prove Theorem~\ref{thm:aggregw} for queries with no
  accumulation.

  Recall that the \cert problem is in PTIME for such queries by
  Theorem~\ref{thm:certptime}, so it suffices to study the case of \poss.
  We do so by the following result, which is obtained by adapting the proof of
  Theorem~\ref{thm:aggregwb}:

  \begin{theorem}
    \label{thm:aggregwbc}
    For any constant $k \in \mathbb{N}$,
    we can determine in PTIME,
    for any input po-relation $\OR$ such that $w(\OR) \leq k$
    and list relation $L$,
    whether $L \in \pw(\OR)$.
  \end{theorem}

  \begin{proof}
    The proof of Theorem~\ref{thm:aggregwb} adapts because of
    the following: to decide instance possibility, we do not need to compute
    \emph{all} possible accumulation results (which may be exponentially
    numerous), but it suffices to store, for each sane
    vector $\mathbf{m}$, whether the prefix of the correct length of the candidate
    possible world can be achieved in the order ideal $s(\mathbf{m})$.
    More formally, we define $t((0, \ldots, 0))
    \defeq \text{true}$, and:
    \[
      t(\mathbf{m}) \defeq \bigvee_{1 \leq i \leq k'} \left(
      t(\mathbf{m} - e_i) \wedge T(L_i[m_i]) = L\left[1 + \sum_{i'} m_{i'}\right]\right)
    \]
    where $L$ is the candidate possible world. We conclude by a dynamic
    algorithm as in Theorem~\ref{thm:aggregwb}.
  \end{proof}

  This concludes the proof of Theorem~\ref{thm:aggregw}, and, as an
  immediate
  corollary, of Theorem~\ref{thm:aggregwCorr}.

\end{toappendix}

We last justify our choice of disallowing the $\times_{\dir}$ product. Indeed, if we
allow $\times_{\dir}$, then \poss is hard on totally ordered po-relations, even if we disallow
$\times_{\lex}$:

\begin{toappendix}
  \subsubsection{Hardness result: Proof of Theorem~\ref{thm:posscompextend1}}
  \label{sec:posscompextend1proof}
\end{toappendix}

\begin{theoremrep}\label{thm:posscompextend1}
  There is a \Pgen query for which the \poss problem is NP-complete
  even when the input po-database is restricted to consist only of totally ordered
  po-relations.
\end{theoremrep}

\begin{toappendix}

\begin{figure}
  \begin{tikzpicture}[scale=.95]
    \node (S) at (-1.5, -1) {$S$};
    \node[draw] (S0) at (-.5, 0) {$0$};
    \node[draw] (S1) at (-.5, -1) {$1$};
    \node[draw] (S2) at (-.5, -2) {$2$};
    \draw[->] (S0) -- (S1);
    \draw[->] (S1) -- (S2);

    \node (S') at (11, 2) {$S'$};

    \node[draw,fill=red] (R1s) at (1, 1.5) {$\s$};
    \node[draw,fill=green] (R1n) at (2, 1.5) {$\n$};
    \node[draw,fill=cyan] (R1e) at (3, 1.5) {$\e$};
    \draw[->] (R1s) -- (R1n);
    \draw[->] (R1n) -- (R1e);

    \node[draw,fill=red] (R2s) at (4, 1.5) {$\s$};
    \node[draw,fill=green] (R2n) at (5, 1.5) {$\n$};
    \node[draw,fill=cyan] (R2e) at (6, 1.5) {$\e$};
    \draw[->] (R2s) -- (R2n);
    \draw[->] (R2n) -- (R2e);

    \node[draw,fill=red] (R3s) at (7, 1.5) {$\s$};
    \node[draw,fill=green] (R3n1) at (8, 1.5) {$\n$};
    \node[draw,fill=green] (R3n2) at (9, 1.5) {$\n$};
    \node[draw,fill=cyan] (R3e) at (10, 1.5) {$\e$};
    \draw[->] (R3s) -- (R3n1);
    \draw[->] (R3n1) -- (R3n2);
    \draw[->] (R3n2) -- (R3e);

    \draw[->] (R1e) -- (R2s);
    \draw[->] (R2e) -- (R3s);

    \draw[dashed] (0.7, 1.2) rectangle (3.3, 1.8);
    \draw[dashed] (3.7, 1.2) rectangle (6.3, 1.8);
    \draw[dashed] (6.7, 1.2) rectangle (10.3, 1.8);

    \node[draw,fill=red] (R1sb) at (1, 0) {$\s$};
    \node[draw,fill=green] (R1nb) at (2, 0) {$\n$};
    \node[draw,fill=cyan] (R1eb) at (3, 0) {$\e$};
    \draw[->] (R1sb) -- (R1nb);
    \draw[->] (R1nb) -- (R1eb);

    \node[draw,fill=red] (R2sb) at (4, 0) {$\s$};
    \node[draw,fill=green] (R2nb) at (5, 0) {$\n$};
    \node[draw,fill=cyan] (R2eb) at (6, 0) {$\e$};
    \draw[->] (R2sb) -- (R2nb);
    \draw[->] (R2nb) -- (R2eb);

    \node[draw,fill=red] (R3sb) at (7, 0) {$\s$};
    \node[draw,fill=green] (R3n1b) at (8, 0) {$\n$};
    \node[draw,fill=green] (R3n2b) at (9, 0) {$\n$};
    \node[draw,fill=cyan] (R3eb) at (10, 0) {$\e$};
    \draw[->] (R3sb) -- (R3n1b);
    \draw[->] (R3n1b) -- (R3n2b);
    \draw[->] (R3n2b) -- (R3eb);

    \draw[->] (R1eb) -- (R2sb);
    \draw[->] (R2eb) -- (R3sb);

    \draw[dashed] (0.7, -1.7) rectangle (3.3, -2.3);
    \draw[dashed] (3.7, -0.7) rectangle (6.3, -1.3);
    \draw[dashed] (6.7, .3) rectangle (10.3, -.3);

    \node[draw,fill=red] (R1sc) at (1, -1) {$\s$};
    \node[draw,fill=green] (R1nc) at (2, -1) {$\n$};
    \node[draw,fill=cyan] (R1ec) at (3, -1) {$\e$};
    \draw[->] (R1sc) -- (R1nc);
    \draw[->] (R1nc) -- (R1ec);

    \node[draw,fill=red] (R2sc) at (4, -1) {$\s$};
    \node[draw,fill=green] (R2nc) at (5, -1) {$\n$};
    \node[draw,fill=cyan] (R2ec) at (6, -1) {$\e$};
    \draw[->] (R2sc) -- (R2nc);
    \draw[->] (R2nc) -- (R2ec);

    \node[draw,fill=red] (R3sc) at (7, -1) {$\s$};
    \node[draw,fill=green] (R3n1c) at (8, -1) {$\n$};
    \node[draw,fill=green] (R3n2c) at (9, -1) {$\n$};
    \node[draw,fill=cyan] (R3ec) at (10, -1) {$\e$};
    \draw[->] (R3sc) -- (R3n1c);
    \draw[->] (R3n1c) -- (R3n2c);
    \draw[->] (R3n2c) -- (R3ec);

    \draw[->] (R1ec) -- (R2sc);
    \draw[->] (R2ec) -- (R3sc);

    \node[draw,fill=red] (R1sd) at (1, -2) {$\s$};
    \node[draw,fill=green] (R1nd) at (2, -2) {$\n$};
    \node[draw,fill=cyan] (R1ed) at (3, -2) {$\e$};
    \draw[->] (R1sd) -- (R1nd);
    \draw[->] (R1nd) -- (R1ed);

    \node[draw,fill=red] (R2sd) at (4, -2) {$\s$};
    \node[draw,fill=green] (R2nd) at (5, -2) {$\n$};
    \node[draw,fill=cyan] (R2ed) at (6, -2) {$\e$};
    \draw[->] (R2sd) -- (R2nd);
    \draw[->] (R2nd) -- (R2ed);

    \node[draw,fill=red] (R3sd) at (7, -2) {$\s$};
    \node[draw,fill=green] (R3n1d) at (8, -2) {$\n$};
    \node[draw,fill=green] (R3n2d) at (9, -2) {$\n$};
    \node[draw,fill=cyan] (R3ed) at (10, -2) {$\e$};
    \draw[->] (R3sd) -- (R3n1d);
    \draw[->] (R3n1d) -- (R3n2d);
    \draw[->] (R3n2d) -- (R3ed);

    \draw[->] (R1ed) -- (R2sd);
    \draw[->] (R2ed) -- (R3sd);

    \draw[->] (R1sb) -- (R1sc);
    \draw[->] (R1nb) -- (R1nc);
    \draw[->] (R1eb) -- (R1ec);
    \draw[->] (R2sb) -- (R2sc);
    \draw[->] (R2nb) -- (R2nc);
    \draw[->] (R2eb) -- (R2ec);
    \draw[->] (R3sb) -- (R3sc);
    \draw[->] (R3n1b) -- (R3n1c);
    \draw[->] (R3n2b) -- (R3n2c);
    \draw[->] (R3eb) -- (R3ec);

    \draw[->] (R1sc) -- (R1sd);
    \draw[->] (R1nc) -- (R1nd);
    \draw[->] (R1ec) -- (R1ed);
    \draw[->] (R2sc) -- (R2sd);
    \draw[->] (R2nc) -- (R2nd);
    \draw[->] (R2ec) -- (R2ed);
    \draw[->] (R3sc) -- (R3sd);
    \draw[->] (R3n1c) -- (R3n1d);
    \draw[->] (R3n2c) -- (R3n2d);
    \draw[->] (R3ec) -- (R3ed);
    \node (Q) at (12, -1) {$\Pi_2(S \times_{\dir} S')$};

    \node (L') at (11, -3.5) {$L'$};

    \node[draw,fill=red] (R1sl) at (1, -3.5) {$\s$};
    \node[draw,fill=red] (R1nl) at (2, -3.5) {$\s$};
    \node[draw,fill=red] (R1el) at (3, -3.5) {$\s$};
    \draw[->] (R1sl) -- (R1nl);
    \draw[->] (R1nl) -- (R1el);

    \node[draw,fill=green] (R2sl) at (4, -3.5) {$\n$};
    \node[draw,fill=green] (R2nl) at (5, -3.5) {$\n$};
    \node[draw,fill=green] (R2el) at (6, -3.5) {$\n$};
    \draw[->] (R2sl) -- (R2nl);
    \draw[->] (R2nl) -- (R2el);

    \node[draw,fill=green] (R3sl) at (7, -3.5) {$\n$};
    \node[draw,fill=cyan] (R3n1l) at (8, -3.5) {$\e$};
    \node[draw,fill=cyan] (R3n2l) at (9, -3.5) {$\e$};
    \node[draw,fill=cyan] (R3el) at (10, -3.5) {$\e$};
    \draw[->] (R3sl) -- (R3n1l);
    \draw[->] (R3n1l) -- (R3n2l);
    \draw[->] (R3n2l) -- (R3el);

    \draw[->] (R1el) -- (R2sl);
    \draw[->] (R2el) -- (R3sl);
  \end{tikzpicture}
  \caption{Example for the proof of Theorem~\ref{thm:posscompextend1}, with $E =
  (1, 1, 2)$ and $B = 4$. The dashed parts of the grid represent $T$, as
  mentioned in the proof sketch.}
  \label{fig:gridpic}
\end{figure}
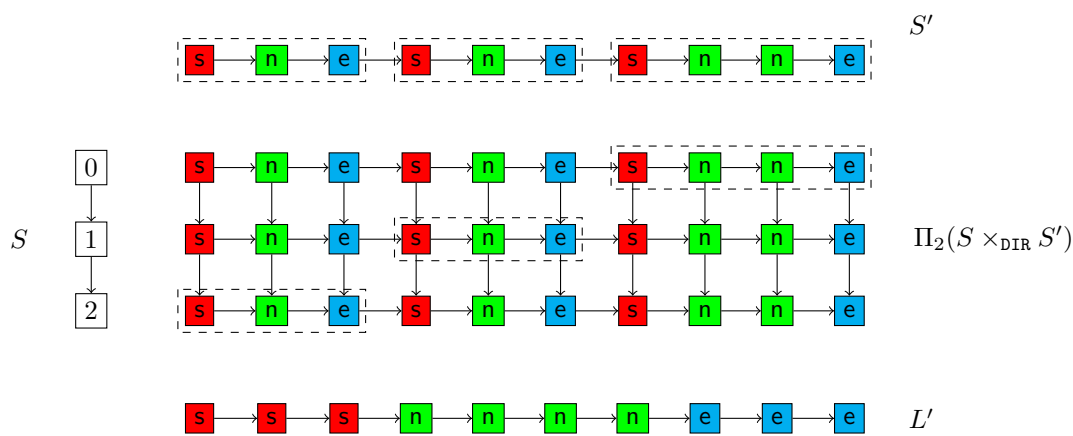

Note that, unlike Proposition~\ref{prp:posshardsimple}, this result does not
follow immediately from the results of~\cite{warmuth1984complexity}. Remember
that \cite{warmuth1984complexity} studies the \emph{shuffle problem} which asks, given a string $w$ and a tuple of strings
$s_1, \ldots, s_n$, whether there is an interleaving of $s_1, \ldots, s_n$ which
is equal to~$w$. It is easy to describe the possible interleavings of the~$s_i$
in \PosRA as a union of
totally ordered list relations, but it is more challenging to test, \emph{with a
constant query}, whether the~$s_i$ have an interleaving equal to~$w$. This is
what we do in the proof of Theorem~\ref{thm:posscompextend1}:

\begin{proof}
  The proof is an adaptation of Proposition~\ref{prp:posshardsimple}. Again, we
  reduce from the NP-hard UNARY-3-PARTITION
  problem~\cite{garey-johnson}: given $3m$ integers $E = (n_1, \ldots, n_{3m})$
  written in unary (not necessarily distinct) and a number $B$, decide if the integers can be partitioned
  in triples such that the sum of each triple is $B$. We reduce an instance
  $\mathcal{I} = (E, B)$ of UNARY-3-PARTITION to a \poss instance in PTIME.
  We fix $\calD \defeq \NN \sqcup \{\s, \n, \e\}$, with $\s$, $\n$ and $\e$
  standing for \emph{start}, \emph{inner},
  and \emph{end} as in the previous proof.
  
  Let $S$ be the totally ordered po-relation $\ordern{3m-1}$,
  and let $S'$ be the totally ordered po-relation whose one possible world is
  constructed from the instance~$\calI$ as
  follows: for $1 \leq i \leq 3m$, we consider the concatenation of one tuple
  $t^i_1$ with value $\s$, $n_i$~tuples $t^i_j$ (with $2 \leq j \leq n_i+1$)
  with value $\n$, and one tuple $t^i_{n_i+2}$ with value $\e$, and $S'$ is the
  total order formed by concatenating the $3m$ sequences of length $n_i+2$.
  Consider the query $Q \defeq \Pi_2(S \times_{\gen} S')$, where $\Pi_2$
  projects to the attribute coming from relation $S'$. See
  Figure~\ref{fig:gridpic} for an illustration, and note the similarity with
  Figure~\ref{fig:gridpic2}. Note that $S$ and $S'$ are
  \emph{input} relations, not constant expressions that would give the same
  relation.
  
  We define the candidate possible world as follows:

  \begin{itemize}
    \item $L_1$ is a list relation defined as the concatenation, for $1
  \leq i \leq 3m$, of $3m-i$ copies of the following 
  sublist: one tuple with value~$\s$, $n_i$
  tuples with value~$\n$, and one tuple with value~$\e$.
    \item $L_2$ is a list relation defined as above, except that
      $3m-i$ is replaced by $i-1$.
    \item $L'$ is the list relation defined as in the proof of
      Proposition~\ref{prp:posshardsimple}, namely, the concatenation of
      $m$ copies of the following sublist:
      three tuples with value $\s$, $B$ tuples with value $\n$, three tuples with
  value~$\e$.
    \item $L$ is the concatenation of $L_1$, $L'$, and $L_2$.
  \end{itemize}

  We now consider the \poss instance that asks whether
  $L$ is a possible world of the query $Q(S, S')$, where $S$ and $S'$ are the
  input totally ordered po-relations. We claim that this instance is positive
  iff the original UNARY-3-PARTITION instance $\calI$ is positive. As the
  reduction process described above is
  clearly PTIME, this suffices to show our desired hardness result, so all that
  remains to show our hardness result for \Pgen is to prove this claim. We now
  do so: the intuition is to eliminate parts of the grid that match to~$L_1$
  and~$L_2$, so that we are left with an order relation that allows us to re-use
  the proof of Proposition~\ref{prp:posshardsimple}.

  \medskip

  Denote by $R$ the po-relation obtained by evaluating~$Q(S, S')$,
  and note that all tuples of $R$ have value in $\{\s, \n, \e\}$.
  For $0 \leq k \leq \card{L_1}$, we write $L_1^{\leq k}$ for the prefix of $L_1$ of length $k$.
  We say that $L_1^{\leq k}$ is a \emph{whole prefix} if either $k
  = 0$ (that is, the empty prefix) or the $k$-th symbol of $L_1$ has value $\e$.
  We say that a linear extension $L''$ of $R$ \emph{realizes} $L_1^{\leq k}$ if
  the sequence of its $k$-th first values is $L_1^{\leq k}$, and that it
  realizes $L_1$ if it realizes $L_1^{\leq \card{L_1}}$. When $L''$ realizes
  $L_1^{\leq k}$, we call the \emph{matched} elements the elements of $R$ that
  occur in the first $k$ positions of $L''$, and say that the other elements are
  \emph{unmatched}. We call the \emph{$i$-th row} of $R$ the elements whose
  first component before projection was $i-1$: note that, for each $i$, $R$
  imposes a total order on the $i$-th row.

  We first observe that for any linear extension $L''$ realizing $L_1^{\leq k}$,
  for all $i$, writing the $i$-th row as $t'_1 < \ldots <
  t'_{\card{S'}}$, the unmatched elements must be all of the form $t'_j$ for $j
  > k_i$ for some $k_i$, i.e., they must be a prefix of the total order of the $i$-th row. Indeed, if they did not form a
  prefix, then some order constraint of $R$ would have been violated when
  enumerating $L''$. Further, by cardinality we clearly have $\sum_{i} k_i=k$. 

  Second, when a linear extension $L''$ of $R$ realizes $L_1^{\leq k}$, we say
  that we are in a \emph{whole situation} if for all $i$, the
  value of element $t'_{k_i+1}$ is either undefined (i.e., there are no row-$i$
  unmatched elements, which means $k_i=\card{S'}$) or it is $\s$. This 
  clearly implies that $k_i$ is of the form $\sum_{j=1}^{l_i} (n_j+2)$ for
  some~$l_i$; letting $S_i$ be the multiset of the $n_j$ for $1 \leq j \leq
  l_i$, we call $S_i$
  the 
  bag of \emph{row-$i$ consumed integers}. The \emph{row-$i$ remaining integers}
  are $E \backslash S_i$ (seeing $E$ as a multiset, and performing difference of
  multisets by subtracting the multiplicities in~$S_i$ to the multiplicities
  in~$E$).

  We now prove the following claim: for any linear extension of~$R$ realizing
  $L_1$, we are in a whole situation, and the multiset union $\biguplus_{1 \leq i
  \leq 3m} S_i$ is equal to the multiset obtained by repeating integer $n_i$ of $E$
  $3m-i$ times for all $1\leq i\leq 3m$.

  We prove the first part of the claim by showing it for all whole prefixes
  $L_1^{\leq k}$, by induction on $k$. It is certainly the case for $L_1^{\leq
  0}$ (the empty prefix). Now, assuming that it holds for prefixes of length up to
  $l$, to realize a whole prefix $L^{\leq l'}$ with $l' > l$, you must first realize a
  strictly shorter whole prefix $L^{\leq l''}$ with $l'' \leq l$ (take it to be of maximal
  length), so by induction hypothesis you are in a whole situation when
  realizing $L^{\leq l''}$. Now to realize the whole prefix $L^{\leq l'}$ having
  realized the whole prefix $L^{\leq l''}$, by construction of~$L_1$, the
  sequence~$L''$ of additional values to realize is $\s$, a certain number of $\n$'s, and
  $\e$, and it is easily seen that this must bring you from a whole situation to
  a whole situation: since there is only one $\s$ in $L''$, there is only one
  row such that an $\s$ value becomes matched; now, to match the additional
  $\n$'s and $\e$, only this particular row can be used, as any first
  unmatched element (if any) of another row is $\s$. Hence the claim is
  proved.

  To prove the second part of the claim, observe that whenever we go from a
  whole prefix to a whole prefix by additionally matching $\s$, $n_j$ times
  $\n$, and $\e$, then we add to $S_i$ the integer~$n_j$. So the claim holds by
  construction of $L_1$.

  A similar argument shows that for any linear extension $L''$ of $R$ whose first
  $\card{L_1}$ tuples achieve $L_1$ and whose last $\card{L_2}$ tuples achieve $L_2$,
  the row-$i$ unmatched elements are a contiguous sequence $t'_j$ with $k_i < j <
  m_i$ for some $k_i$ and $m_i$. In addition, if we have $k_i<m_i-1$, then $t'_{k_i}$ has value $\e$
  and $t'_{m_i}$ has value $\s$, and the unmatched values (defined in an
  analogous fashion) are a multiset corresponding exactly to the elements
  $n_1,\dots,n_{3m}$. So the unmatched elements when having read $L_1$ (at
  the beginning) and $L_2$ (at the end) are formed of $3m$
  lists, of length $n_i + 2$ for $1 \leq i \leq 3m$, of the
  form $\s$, $n_i$ times $\n$, and $\e$, with a certain order relation between
  the elements of the sequences (arising from the fact that some may be on the
  same row, or that some may be on different rows but comparable by definition
  of $\times_\gen$).

  But we now notice that we can clearly achieve $L_1$ by picking the following,
  in that order: for $1\leq j\leq 3m$, for $1\leq i\leq 3m-j$, pick the first
  $n_j+2$ unmatched tuples of row $i$. Similarly, to achieve $L_2$ at the end,
  we can pick the following, in \emph{reverse} order: for $3m\geq j\geq 1$, for
  $3m\geq i\geq 3m-j+1$, the last $n_j + 2$ unmatched tuples of row $i$. When we
  pick elements this way, the unmatched elements are $3m$ lists
  (one for each row, with that of row $i$ being $\s$, $n_i$ times $\n$
  and $\e$, for all $i$) and there are no order relations across sequences.
  Let $T$ be the sub-po-relation of $R$ that consists of exactly these unmatched
  elements. We denote the elements of~$T$ as $u_l^j$ with $1\leq j\leq 3m$
  iterating over the lists, and $1\leq l\leq n_j+2$
  iterating within each sequence. $T$ is the parallel composition of $3m$ total
  orders, namely, $u_1^j < u_2^j < \cdots < u_{n_j+2}^j$ for all $j$, having
  values $\s$ for $u_1^j$, $\e$ for $u_{n_j+2}^j$, and $\n$ for the others.

  We now claim that for any sequence $L''$, the concatenation $L_1 L'' L_2$ is a
  possible world of~$R$ if and only if $L''$ is a possible world of $T$. The
  ``only if'' direction was proved with the construction above. The ``if''
  direction comes from the fact that $T$ is the \emph{least constrained}
  possible po-relation for the unmatched sequences, since the order on the
  sequences of remaining elements when matching $L_1$ and $L_2$ is known to be
  total. Hence, to prove our original claim, it only remains to show that the
  UNARY-3-PARTITION instance $\calI$ is positive iff $L'$ is a possible world
  of~$T$. This claim is shown exactly as in the proof of
  Proposition~\ref{prp:posshardsimple}, as $L'$ is the same as in that proof,
  and $T$ is the same order relation as $\OR$ in that proof. This concludes the
  proof of the desired result.
\end{proof}

\end{toappendix}

\subparagraph*{Disallowing product.}\label{sec:product}
\begin{toappendix}
\subsection{Disallowing Product}
\end{toappendix}

We have shown the tractability of \poss when disallowing the $\times_\gen$
operator, when the input po-relations are assumed to have bounded width.
We now show that if we disallow both kinds of product, we obtain tractability for
more general input po-relations. Specifically, we will allow input po-relations
that are almost totally ordered, i.e., have bounded \emph{width}; and we will also
allow input po-relations that are almost unordered, which we measure using a new
order-theoretic notion of \emph{ia-width}. The idea of ia-width is to decompose
the relation in classes of indistinguishable sets of incomparable elements:

\begin{definition}
  \label{def:iawidth}
  Given a poset $P = (\ID, <)$, a subset $A \subseteq \ID$ is an
  \emph{indistinguishable antichain} if it is both an antichain (there are no
  $x, y \in A$ such that $x < y$) and an \emph{indistinguishable set} (or
  \emph{interval}~\cite{fraisse1984intervalle}): for all
  $x, y \in A$ and $z \in \ID \backslash A$,
  we have $x < z$ iff $y < z$, and $z < x$ iff $z < y$.

  An \emph{indistinguishable antichain partition} (ia-partition) of~$P$ is a
  partition $\ID = A_1 \sqcup \cdots \sqcup A_n$ of~$\ID$ such that each $A_i$
  for $1 \leq i \leq n$ is an indistinguishable antichain. The
  \emph{cardinality} of the partition is~$n$.
  The \emph{ia-width} of~$P$ is the
  cardinality of its smallest ia-partition.
  The \emph{ia-width} of a po-relation is that of its underlying poset,
  and the \emph{ia-width} of a po-database is the maximal ia-width of its po-relations.
\end{definition}

Hence, any po-relation~$\OR$ has ia-width at most~$\card{\OR}$, with the trivial
ia-partition consisting of singleton indistinguishable antichains, and unordered
po-relations have an ia-width of 1.
Po-relations may have low ia-width in practice if order is completely
unknown except for a
few comparability pairs given by users, or when they consist of objects from a
constant number of types that are ordered based only on some order on the types.

We can now state our tractability result when disallowing both kinds of
products, and allowing both bounded-width and bounded-ia-width relations. For
instance, this result allows us to
combine sources whose order is fully unknown or irrelevant, with
sources that are completely ordered (or almost totally ordered).

\begin{toappendix}
  \subsubsection{Tractability Result: Proof of Theorem~\ref{thm:aggregnoprod}}
  \label{sec:noprodtract}
\end{toappendix}

\begin{theoremrep}\label{thm:aggregnoprod}
  For any fixed $k \in \mathbb{N}$
  and fixed \Pnoprod query $Q$,
  the \poss problem for~$Q$ is in PTIME
  when all po-relations of the input po-database
  have either ia-width
  $\leq k$ or width~$\leq k$.
\end{theoremrep}

\begin{toappendix}
  We start by making a simple observation:

  \begin{lemma}
    \label{lem:rewritenoprod}
    Any \PosRA query $Q$ without any product can be rewritten as a union of
    projections of selections of a constant number of input relations and constant
    relations.
  \end{lemma}

  \begin{proof}
    This follows from the fact that, for the semantics that we have defined for
    operators, the following is clear:
    selection commutes with union, selection commutes with projection, and
    projection commutes with union. Hence, we can perform the desired rewriting.
  \end{proof}

We can thus rewrite the input query using this lemma. The idea is that we will
evaluate the query in PTIME using Proposition~\ref{prp:repsys}, argue that
the width bounds are preserved using Proposition~\ref{prp:lexwidth}, and compute a
chain partition of the relations using Theorem~\ref{thm:dilworth} and
  Theorem~\ref{thm:fulkerson}. However, we first
need to show analogues of Proposition~\ref{prp:lexwidth},
  Theorem~\ref{thm:dilworth}, and Theorem~\ref{thm:fulkerson} for the new notion of ia-width. We first show the
analogue of Proposition~\ref{prp:lexwidth} for the case without product:

\begin{proposition}\label{prp:lexiawidthnoprod}
  Let $k \geq 2$ and $Q$ be a \Pnoprod query. Let $k' \colonequals
  \max(k,q) \times \card{Q}$, where $\card{Q}$ denote the number of
  symbols of~$Q$, and where $q$ denotes the largest value such that
  $\ordern{q}$ appears in~$Q$. For any po-database~$D$ of ia-width~$\leq
  k$, the po-relation $Q(D)$ has ia-width $\leq k'$.
\end{proposition}

\begin{proof}
  We first show by induction on~$Q$ that the ia-width of the query output can be
  bounded as a function of the bound~$k$ on the ia-width of the query inputs.
  For the base cases:

  \begin{itemize}
    \item The input relations have ia-width at most~$k$.
    \item The constant relations have ia-width $\leq q$ with the trivial
      ia-partition consisting of singleton classes.
  \end{itemize}

For the induction step:

  \begin{itemize}
    \item Projection clearly does not change ia-width.
    \item Selection may only decrease the ia-width. Indeed, consider an
      ia-partition of the input po-relation, apply the selection to each class,
      and remove the classes that became empty. The number of classes has not
      increased, and it is clear that the result is still an ia-partition of the
      output po-relation.
    \item The union of two relations
  with ia-width $k_1$ and $k_2$ has ia-width at most $k_1 + k_2$. Indeed, we can
      obtain an ia-partition for the union as the union of ia-partitions for the
      input relations.
  \end{itemize}

  Second, we see that the bound $k' \colonequals \max(k,q) \times \card{Q}$ is
  clearly correct, because the base cases have ia-width $\leq \max(k, q)$ and
  the worst operators are unions, which amount to summing the ia-width bounds on
  all inputs, of which there are $\leq \card{Q}$. So we have shown the desired
  bound.
\end{proof}

We next show that, like chain partitions for bounded-width po-relations, we can
efficiently compute an ia-partition for a bounded-ia-width po-relation:

\begin{proposition}\label{prp:ptime-ia-partition}
  The ia-width of any poset, and a corresponding ia-partition, can be computed
  in PTIME.
\end{proposition}

To show this result, we need two preliminary observations about
indistinguishable antichains:

\begin{lemma}
  \label{lem:iasubset}
  For any poset $(\ID, <)$ and indistinguishable antichain $A$, for any $A'
  \subseteq A$, then $A'$ is an indistinguishable antichain.
\end{lemma}

\begin{proof}
  Clearly $A'$ is an antichain because $A$ is. We show that it is an
  indistinguishable set. Let $x, y \in A'$ and $z \in \ID \backslash A'$, and show
  that $x < z$ implies $y < z$ (the other three implications are symmetric). If
  $z \in \ID \backslash A$, we conclude because $A$ is an indistinguishable set.
  If $z \in A \backslash A'$, we conclude because, as $A$ is an antichain, $z$
  is incomparable both to $x$ and to $y$.
\end{proof}

\begin{lemma}
  \label{lem:indistinguishablea}
  For any poset $(\ID, <)$ and indistinguishable antichains $A_1, A_2 \subseteq
  \ID$ such that $A_1 \cap A_2 \neq \emptyset$, the union $A_1 \cup A_2$ is an
  indistinguishable antichain.
\end{lemma}

\begin{proof}
  We first show that $A_1 \cup A_2$ is an indistinguishable set.
  Let $x, y \in A_1 \cup A_2$ and
  $z \in \ID \backslash (A_1 \cup A_2)$, assume that $x < z$ and show that $y <
  z$ (again the other three implications are symmetric).
  As $A_1$ and $A_2$ are indistinguishable sets, this is immediate unless $x
  \in A_1 \backslash A_2$ and $y \in A_2 \backslash A_1$, or vice-versa. We
  assume the first case as the second one is symmetric. Consider
  $w \in A_1 \cap A_2$. As $x < z$, we know that $w < z$ because $A_1$ is an
  indistinguishable set, so that $y < z$ because $A_2$ is an indistinguishable
  set, which proves the desired implication.

  Second, we show that $A_1 \cup A_2$ is an antichain.
  Proceed by contradiction, and let $x, y \in
  A_1 \cup A_2$ such that $x < y$. As $A_1$ and $A_2$ are antichains, we must
  have $x \in A_1 \backslash A_2$ and $y \in A_2 \backslash A_1$, or vice-versa.
  Assume the first case, the second case is symmetric. As $A_1$ is an
  indistinguishable set, letting $w \in A_1 \cap A_2$,
  as $x < y$ and $x \in A_1$, we have $w < y$. But $w \in A_2$ and
  $y \in A_2$, which is impossible because $A_2$ is an antichain. We have
  reached a contradiction, so we cannot have $x < y$. Hence, $A_1 \cup A_2$ is
  an antichain, which concludes the proof.
\end{proof}

We can now show Proposition~\ref{prp:ptime-ia-partition}:

\begin{proof}
  Start with the trivial partition in singletons (which is an
  ia-partition), and for every pair of items, see if their
  current classes can be merged (i.e., merge them, check in PTIME if it
  is an antichain, and if it is an indistinguishable set, and undo the merge if
  it is not). Repeat the process
  while it is possible to merge classes (i.e., at most linearly many times).
  This greedy process concludes in PTIME and yields an ia-partition
  $\mathbf{A}$. Let $n$ be its cardinality.

  Now assume that there is an ia-partition $\mathbf{A'}$ of cardinality $m < n$.
  There has to be a class $A'$ of~$\mathbf{A'}$
  which intersects two different classes $A_1 \neq A_2$ of
  the greedy ia-partition $\mathbf{A}$, otherwise $\mathbf{A'}$ would be a refinement
  of $\mathbf{A}$ so we would have $m \geq n$.
  Now, by Lemma~\ref{lem:indistinguishablea},
  $A \cup A_1$ and $A \cup A_2$, and hence
  $A \cup A_1 \cup A_2$, are
  indistinguishable antichains.
  By Lemma~\ref{lem:iasubset},
  this implies that $A_1 \cup A_2$ is an indistinguishable
  antichain. Now, when constructing the greedy ia-partition $\mathbf{A}$,
  the algorithm has
  considered one element of $A_1$ and one element of $A_2$, attempted to merge
  the classes $A_1$ and $A_2$,
  and, since it has not merged them in~$\mathbf{A}$, the union $A_1 \cup A_2$
  cannot be an indistinguishable antichain. We have reached a contradiction, so
  we cannot have $m < n$, which concludes the proof.
\end{proof}

We have shown the preservation of ia-width bounds through selection, projection,
and union (Proposition~\ref{prp:lexiawidthnoprod}), and shown how to compute an
ia-partition in PTIME (Proposition~\ref{prp:ptime-ia-partition}). Let us now return
to the proof of Theorem~\ref{thm:aggregnoprod}. We use
Lemma~\ref{lem:rewritenoprod} to rewrite the query to a union of projection of
selections. We evaluate the selections and projections in PTIME by
Proposition~\ref{prp:repsys}. As union is clearly associative and commutative,
we evaluate the union of relations of width $\leq k$, yielding $\OR$, and the union of those
of ia-width $\leq k$, yielding~$\OR'$. The first result $\OR$ has bounded width thanks to
Proposition~\ref{prp:lexwidth}, and we can compute a chain partition of it in PTIME using
Theorem~\ref{thm:dilworth} and Theorem~\ref{thm:fulkerson}. The second result has bounded ia-width thanks to
Proposition~\ref{prp:lexiawidthnoprod}, and we can compute an ia-partition of it in
PTIME using Proposition~\ref{prp:ptime-ia-partition}.

  \subparagraph*{Queries with no accumulation.}
  We first prove Theorem~\ref{thm:aggregnoprod} for the case without
  accumulation. It suffices to show the following:

  \begin{proposition}
    \label{prp:aggregwuawb}
    For any constant $k \in \mathbb{N}$,
    we can determine in PTIME,
    for any input po-relation $\OR$ with width $\leq k$,
    input po-relation $\OR'$ with ia-width $\leq k$,
    and list relation $L$,
    whether $L \in \pw(\OR \cupgen \OR')$.
  \end{proposition}

Before proving this, we show a weaker result that restricts to a
bounded-ia-width input relation:

\begin{proposition}\label{prp:aggreguawinst}
    For any constant $k \in \mathbb{N}$,
    we can determine in PTIME,
    for any po-relation $\OR$ with ia-width $\leq k$
    and list relation $L$,
    whether $L \in \pw(\OR)$.
  \end{proposition}

\begin{proof}
  Let $\mathbf{A} = (A_1, \ldots, A_k)$ be an ia-partition of width $k$ of $\OR
  = (\ID, T, <)$, which can be computed in PTIME
  by Proposition~\ref{prp:ptime-ia-partition}. We assume that the
  length of the candidate possible world $L$ is $\card{\ID}$,
  as we can trivially reject otherwise.

  If there is a way to realize $L$ as a possible world of $\OR$,
  For any linear extension $<'$ of~$\OR$, we call the
  \emph{finishing order} $<'$ the permutation $\pi$ of $\{1, \ldots, k\}$ obtained by
  considering, for each class $A_i$ of $\mathbf{A}$, the largest position $1
  \leq n_i \leq \card{\ID}$
  in~$<'$ to which an element of $A_i$ is mapped, and sorting the
  class indexes by ascending finishing order. We say we can realize $L$ with
  finishing order $\pi$ if there is a linear extension of~$\OR$ that realizes
  $L$ and whose finishing order
  is~$\pi$. Hence, it suffices to check, for every possible permutation $\pi$ of
  $\{1, \ldots, k\}$,
  whether $L$ can be realized from $\OR$ with finishing order $\pi$: this does not
  make the complexity worse because the number of finishing orders depends only on
  $k$ and not on $\OR$, so it is constant. (Note that the order relations across
  classes may imply that some finishing orders are impossible to realize
  altogether.)

  We now claim that to determine whether $L$ can be realized with finishing
  order $\pi$, the following greedy algorithm works. Read $L$ linearly. At any
  point, maintain the set of elements of $\OR$ which have already been used
  (distinguish the \emph{used} and \emph{unused} elements; initially all
  elements are unused), and distinguish the classes of~$\mathbf{A}$ in three
  kinds: the \emph{exhausted
  classes}, where all elements are used; the \emph{open classes},
  the ones where some elements are unused and all ancestor elements outside of
  the class are used;
  and the \emph{blocked
  classes}, where some ancestor element outside of the class is not used.
  Initially, the
  open classes are those which are roots in the poset obtained from the
  underlying poset of~$\OR$ by quotienting by the equivalence relation induced by
  $\mathbf{A}$; and the other classes are blocked.

  When reading a value $t$ from $L$, consider all open classes. If none of these
  classes have an unused element with value $t$, reject, i.e., conclude that we cannot realize $L$
  as a possible world of $\OR$ with finishing order $\pi$. Otherwise, take the
  open class that comes first in the finishing order, and use
  an arbitrary suitable element from it. Update the class to be \emph{exhausted} if it
  is: in this case, check that the class was the next one in the finishing
  order~$\pi$ (and reject otherwise), and 
  update from \emph{blocked} to \emph{open} the classes that
  must be. Once $L$ has been completely read, accept: as $\card{L} = \card{\ID}$
  we know that all elements are now used.

  It is clear by construction that if this greedy algorithm accepts then there
  is a linear extension of~$\OR$ that realizes $L$ with finishing order~$\pi$;
  indeed, when the algorithm suceeds then it has clearly respected the finishing
  order $\pi$, and whenever an identifier $\id$ of~$\OR$ is marked as
  \emph{used} by the algorithm, then $\id$ has
  the right value relative to the element of~$L$ that has just been read, and
  $\id$
  is in an open class so no order relations of~$\OR$ are violated by enumerating
  $\id$ at this point of the linear extension. The interesting direction is
  the converse: show that if $L$
  can be realized by a linear extension $<'$ of~$\OR$ with finishing order $\pi$,
  then the algorithm accepts when considering $\pi$.
  To do so, we must show that if there is such a linear extension,
  then there is such a linear extension where identifiers are enumerated as in
  the greedy algorithm, i.e., we always choose an identifier with the right
  value and in the open class with the smallest finishing time: we call this a
  \emph{minimal} identifier. (Note that we do not need to worry about which
  identifier is chosen: once we have decided on the value of the identifier and
  on its class, then it does not matter which element we choose, because all
  elements in the class are unordered and have the same order relations to
  elements outside the class thanks to indistinguishability.)
  If we can prove this, then this justifies the existence of a linear extension
  that the greedy algorithm will construct, which we call a \emph{greedy linear
  extension}.

  Hence, let us see why it is always possible to enumerate minimal identifiers. 
  Consider a linear extension $<'$ and take the smallest position in~$L$ where
  $<'$ chooses an identifier $\id$ which is non-minimal. We know that $\id$ must
  still have the correct value, i.e., $T(\id)$ is determined, and by definition
  of a linear extension, we know that $\id$ must be in an open class. Hence, we
  know that the class $A$ of~$\id$ is non-minimal, i.e., there is another open
  class $A'$ containing an unused element with value $T(\id)$, and $A'$ is
  before~$A$ in the finishing order~$\pi$. Let us take for $A'$ the first open
  class with such an unused element in the finishing order~$\pi$, and let $\id'$
  be a minimal element, i.e., an element of~$A'$ with $T(\id') = T(\id)$. Let us
  now construct a different linear extension $<''$ by swapping $\id$ and $\id'$,
  i.e., enumerating $\id'$ instead of~$\id$, and enumerating $\id$ in~$<''$ at
  the point where $<'$ enumerates $\id'$. It is clear that the sequence of
  values (images by~$T$) of the identifiers in~$<''$ is still the same as
  in~$<'$. Hence, if we can show that $<''$ additionally satisfies the order
  constraints of~$\OR$, then we will have justified the existence of a linear
  extension that enumerates minimal identifiers until a later position; so,
  reapplying the rewriting argument, we will deduce the existence of a greedy
  linear extension. So it only remains to show that $<''$ satisfies the order
  constraints of~$\OR$.

  Let us assume by way of contradiction that $<''$
  violates an order constraint of~$\OR$. There are two possible kinds of
  violation. The first kind is if $<'$ enumerates an element~$\id''$
  between~$\id$ and~$\id'$ for which $\id < \id''$, so that having $\id'' <''
  \id$ in~$<''$ is a violation. The second kind is if $<'$ enumerates an
  element~$\id''$ between $\id$ and $\id'$ for which $\id'' < \id'$, so that
  having $\id'' <'' \id'$  in~$<''$ is a violation. The second kind of violation
  cannot happen because we know that $\id'$ is in an open class when $<'$
  considers~$\id$, i.e., we have ensured that $\id'$ can be enumerated instead
  of~$\id$. Hence, we focus on violations of the first kind. Consider $\id''$
  such that $\id <' \id'' <' \id'$ and let us show that we do not have $\id <
  \id''$. Letting $A''$ be the class of~$\id''$, we assume that $A'' \neq A$, as
  otherwise there is nothing to show because the classes are antichains.
  Now, we know from~$<'$ that we do not have $\id'
  <' \id''$, and that the class~$A'$ of~$\id'$ is not exhausted when $<'$
  enumerates~$\id''$. As~$<'$ respects the finishing order~$\pi$, and $A'$ comes
  before~$A$ in~$\pi$, we know that $A$ is not exhausted either when~$<'$
  enumerates~$\id''$. Letting $\id_A$ be an element of~$A$ which is still unused
  when~$<'$ enumerates~$\id''$, we know that we do not have $\id_A < \id''$. So
  as $\id'' \notin A$ we know by indistinguishability that we do not have $\id <
  \id''$ either. This is what we wanted to show, so $\id''$ cannot witness a
  violation of the first kind.  Hence $<''$  does not violate the order
  constraints of~$\OR$, and repeating this rewriting argument shows that there
  is a greedy linear extension that the greedy algorithm will find,
  contradicting the proof.
\end{proof}

We now extend this proof to show Proposition~\ref{prp:aggregwuawb}:

\begin{proof}
  As in the proof of Proposition~\ref{prp:aggreguawinst}, we will enumerate all possible finishing
  orders for the classes of $\OR'$, of which there are constantly many, and apply
  an algorithm for each finishing order $\pi$, with the algorithm succeeding iff it
  succeeds for some finishing order.

  We first observe that if there is a way to achieve $L$ as a possible world of
  $\OR \cup \OR'$ for a finishing order $\pi$, then there is one where the subsequence of
  the tuples that are matched to $\OR'$ are matched following a greedy strategy as
  in Proposition~\ref{prp:aggreguawinst}. This is simply because $L$ must then be an
  interleaving of a possible world of $\OR$ and a possible world of $\OR'$, and a
  match for the possible world of $\OR'$ can be found as a greedy match, by what
  was shown in the proof of Proposition~\ref{prp:aggreguawinst}. So it suffices to assume
  that the tuples matched to $\OR'$ are matched following the greedy algorithm of
  Proposition~\ref{prp:aggreguawinst}.

  Second, we observe the following: for any prefix $L'$ of $L$ and order ideal
  $\OR''$ of $\OR$, if we realize $L'$ by matching exactly the tuples of
  $\OR''$ in $\OR$,
  and by matching the other tuples to $\OR'$ following a greedy strategy, then the
  matched tuples in $\OR'$ are entirely determined (up to replacing tuples in a
  class by other tuples with the same value). This is because, while there may
  be multiple ways to match parts of $L'$ to $\OR''$ in a way that leaves a
  different sequence of tuples to be matched to $\OR'$, all these ways make us
  match the same bag of tuples to $\OR'$; now the state of $\OR'$ after matching a bag
  of tuples following the greedy strategy (for a fixed finishing order) is the
  same, no matter the order in which these tuples are matched, assuming that the
  match does not fail.

  This justifies that we can solve the problem with a dynamic algorithm again.
  The state contains the position $\mathbf{b}$ in each chain of $\OR$, and a position $i$
  in the candidate possible world. As in the proof of
  Theorem~\ref{thm:aggregwb}, we filter the configurations so that they are sane
  with respect to the order constraints between the chains of $\OR$. For each
  state, we will store a Boolean value indicating whether the prefix of length
  $i$ of $L$ can be realized by $\OR \cupgen \OR'$ such that the tuples
  of $\OR$ that
  are matched is the order ideal $s(\mathbf{b})$
  described by $\mathbf{b}$, and such that the
  other tuples of the prefix are matched to $\OR'$ following a greedy strategy with
  finishing order $\pi$. By
  our second remark above, when the Boolean is true, the state of $\OR'$ is
  uniquely determined, and we also store it as part of the state (it is
  polynomial) so that we do not have to recompute it each time.

  From each state we can make progress by consuming the next tuple from the
  candidate possible world, increasing the length of the prefix, and reaching
  one of the following states: either match the tuple to a chain of $\OR$,
  in which case we make progress in one chain and the consumed tuples in
  $\OR'$
  remain the same; or make progress in $\OR'$, in which case
  we look at the previous state of $\OR'$ that was stored and consume a tuple from
  $\OR'$ following the greedy algorithm of Proposition~\ref{prp:aggreguawinst}:
  more specifically, we find an unused tuple with
  the right label which is in the open class that appears first in the finishing
  order, 
  if the class is now exhausted we verify that it was supposed to be the next
  one according to the finishing order, and
  we update the open, exhausted and blocked status of the
  classes.

  Applying the dynamic algorithm allows us to conclude whether $L$ can be
  realized by matching all tuples of $\OR$, and matching tuples in $\OR'$ following
  the greedy algorithm with finishing order $\pi$ (and checking cardinality
  suffices to ensure that we have matched all tuples of $\OR'$). If the answer of
  the dynamic algorithm is YES, then it is clear that, following
  the path from the initial to the final state found by the dynamic algorithm,
  we can realize $L$. Conversely, if $L$ can be realized, then by our
  preliminary remark it can be realized in a way that matches tuples in
  $\OR'$
  following the greedy algorithm for some finishing order. Now, for that
  finishing order, the path of the dynamic algorithm that matches tuples
  to $\OR$
  or to $\OR'$ following that match will answer YES.
\end{proof}

  \subparagraph*{\Pnoprodacc queries with finite and position-invariant
  accumulation.}
  We now prove the result for the case of a query with accumulation. In this
  setting, the
  results for \poss and \cert follow from the following claim:

  \begin{theorem}
    \label{thm:aggregwuawb}
    For any constant $k \in \mathbb{N}$,
    and position-invariant accumulation operator $\accum_{h, \oplus}$ with
    finite domain, we can compute in PTIME,
    for any input po-relation $\OR$ with width $\leq k$
    and input po-relation $\OR'$ with ia-width $\leq k$,
    the set $\accum_{h, \oplus}(\OR \cupgen \OR')$.
  \end{theorem}

  \begin{proof}
    We use Theorem~\ref{thm:dilworth} and Theorem~\ref{thm:fulkerson} to compute in PTIME a chain partition of
  $\OR$, and we use Proposition~\ref{prp:ptime-ia-partition} to compute in PTIME
  an ia-partition $A_1 \sqcup \cdots \sqcup A_n$ of minimal cardinality of
  $\OR'$, with $n \leq k$.

  We then apply a dynamic algorithm whose state
  consists of:
  \begin{itemize}
  \item for each chain in the partition of $\OR$, the position in
    the chain;
  \item for each class $A$ of the ia-partition of $\OR'$, for each element $m$ of the
    monoid, the number of identifiers $\id$ of~$A$ such that $h(T(\id), 1)
      = m$ that have already been used.
  \end{itemize}

  There are polynomially many possible states; for the second bullet point, this
  uses the fact that the monoid is finite, so its size is constant because it is
  fixed as part of the query. Also note that we use the rank-invariance of~$h$
  in the second bullet point.

  The possible accumulation results for each of the possible states can then be
  computed by a dynamic algorithm. At each state, we can decide to make progress
  either in a chain of $\OR$ (ensuring that the element that we enumerate has
  the right image by~$h$, and that the new vector of positions of the chains is
  still sane, i.e., yields an order ideal of~$\OR$) or in a class of $\OR'$
  (ensuring that this class is open, i.e., it has no ancestors in~$\OR'$ that
  were not enumerated yet, and that it contains an element which has the right
  image by~$h$). The correctness of this algorithm is because these is a
  bijection between the ideals of $\OR \cupgen \OR'$ and the pairs of ideals of
  $\OR$ and of ideals of~$\OR'$. Now, the dynamic algorithm considers all ideals
  of~$\OR$ as in the proof of Theorem~\ref{thm:aggregwb}, and it clearly
  considers all possible ideals of~$\OR'$ except that we identify ideals that
  only differ by elements in the same class which are mapped to the same value
  by~$h$ (but this choice does not matter because the class is an antichain and
  these elements are indistinguishable outside the class).

  As in the proof of Theorem~\ref{thm:aggregwb}, we can ensure that all
  accumulation operations are in PTIME, using PTIME-evaluability of the
  accumulation operator, up to the technicality of storing at each state,
  for each of the possible accumulation results, a witnessing totally ordered
  relation from which to compute it in PTIME.
  \end{proof}
\end{toappendix}

Disallowing product is 
severe, but we can still integrate sources by taking the
\emph{union} of their tuples, selecting subsets, and modifying tuple
values with projection. In fact, allowing product makes
\poss intractable when allowing both unordered and totally ordered input:

\begin{toappendix}
  \subsubsection{Hardness result: Proof of Theorem~\ref{thm:posscompextended}}
  \label{sec:posscompextendedproof}
\end{toappendix}

\begin{theoremrep}\label{thm:posscompextended}
  There is a \Plex query and a \Pgen query for which the \poss problem is NP-complete
  even when the input po-database is restricted to consist only of one totally
  ordered and one unordered po-relation.
\end{theoremrep}

\begin{toappendix}
  The proof is by adapting the proof of Theorem~\ref{thm:posscompextend1}.
  The argument is exactly the
  same, except that we take relation $S$ to be \emph{unordered} rather than
  totally ordered. Intuitively, in Figure~\ref{fig:gridpic}, this means that we
  drop the vertical edges. The proof adapts, because it only used
  the fact that $t'_j < t'_k$ for $j < k$ within a row-$i$; we never used
  the comparability across groups.
\end{toappendix}

\section{Tractable Cases for Accumulation Queries}\label{sec:fpt}
We next study tractable cases for \poss and
\cert in presence of accumulation.

\subparagraph*{Cancellative monoids.}
\begin{toappendix}
  \subsection{Cancellative monoids}
\end{toappendix}
We first consider a natural restriction on the accumulation function:

\begin{definition}[\cite{howie1995fundamentals}]
    \label{def:cancellative}
    For any monoid $(\calM, \oplus, \epsilon)$,
    we call $a \in \calM$ \deft{cancellable} if, for all~$b, c \in \calM$, we
    have that $a \oplus b = a \oplus c$ implies $b = c$, and we also have that $b \oplus a = c
    \oplus a$ implies $b = c$.
    We call $\calM$ a \deft{cancellative monoid} if all
    its elements are cancellable.
\end{definition}

Many interesting monoids are cancellative; in particular, this is the
case of both monoids in
Example~\ref{exa:aggreg}. More generally, all \emph{groups} are cancellative
monoids (but some infinite cancellative monoids are not groups, e.g., the monoid
of concatenation).
For this large class of accumulation functions, we design an efficient algorithm for certainty.

\begin{theoremrep}\label{thm:certaintyptimec}
    \cert is in PTIME for any \PosRAagg{} query
    that performs accumulation in a cancellative monoid.
\end{theoremrep}

\begin{proofsketch}
  We show that the accumulation result in cancellative monoids is certain iff the
  po-relation on which we apply accumulation respects the following \emph{safe
  swaps} criterion: for all tuples $t_1$ and $t_2$ and consecutive positions $p$ and $p+1$
  where they may appear, we have $h(t_1, p) \oplus h(t_2, p+1) = h(t_2, p) \oplus h(t_1,
  p+1)$. We can check this in PTIME.
\end{proofsketch}

\begin{toappendix}
    \label{sec:proof-cancellative-monoids}
We formalize the definition of possible ranks for pairs of incomparable
elements, and of the \emph{safe swaps} property:

\begin{definition}
    \label{def:pr2}
    Given two \emph{incomparable} elements $x$ and $y$ in $\OR$, their
    \deft{possible ranks} $\pr_\OR(x, y)$ is the interval $[a+1, \card{\OR} - d]$,
    where $a$ is the number of elements that are either ancestors of~$x$ or
    of~$y$ in~$\OR$ (not including $x$ and $y$), and $d$ is the number of elements
    that are either descendants of $x$ or of $y$ (again excluding $x$ and $y$
    themselves).

    Let $(\calM, \oplus, \epsilon)$ be an accumulation monoid and let $h :
    \calD \times \mathbb{N} \to \calM$ be an accumulation map.
        The po-relation $\OR$ has the \emph{safe swaps} property with respect
    to $\calM$ and $h$ if the following holds: for any pair $t_1 \neq t_2$ of
    incomparable tuples of $\OR$, for any pair $p, p+1$ of \emph{consecutive}
    integers in $\pr_\OR(t_1, t_2)$, we have:
    \[
    h(t_1, p) \oplus h(t_2, p+1) = h(t_2, p) \oplus h(t_1, p+1)
    \]
\end{definition}

We first show the following soundness result for possible ranks:

\begin{lemma}
  \label{lem:achieverank}
  For any poset $P$ and incomparable elements $x, y \in P$, for any $p \neq q \in \pr_P(x,
  y)$, there exists a linear extension~$\Lambda$ of~$P$ such that element $x$ is
  enumerated at position $p$ in $\Lambda$, and element $y$ is enumerated at position
  $q$, and we can compute it in PTIME from~$P$.
\end{lemma}

\begin{proof}
  We can construct the desired linear extension $\Lambda$ by starting to enumerate
  all elements which are ancestors of either~$x$ or~$y$ in any order, and
  finishing by enumerating all elements which are descendants of either~$x$
  or~$y$, in any order: that this can be done without enumerating either $x$
  or $y$ follows from the fact that $x$ and $y$ are incomparable.

  Call $p' = p - a$, and $q' = q - a$; it follows from the definition of
  $\pr_P(x, y)$ that $1 \leq p', q' \leq \card{P} - d - a$, and clearly $p'
  \neq q'$.

  All unenumerated elements are either $x$, $y$, or incomparable to both $x$
  and $y$. Consider any linear extension of the unenumerated elements except
  $x$ and $y$; it has length $\card{P} - d - a - 2$.
  Now, as $p' \neq q'$, if $p' < q'$, we can enumerate $p' - 1$ of these
  elements, enumerate $x$, enumerate $q' - p' - 1$ of these elements,
  enumerate $y$, and enumerate the remaining elements, following the linear
  extension. We proceed similarly, reversing the roles of $x$ and $y$, if $q'
  < p'$. The overall process is clearly in PTIME.
\end{proof}

We can then show:

\begin{lemma}
  \label{lem:safeptime}
  For any fixed (PTIME-evaluable) accumulation operator $\accum_{h, \oplus}$ we
  can determine in PTIME, given a po-relation $\OR$, whether $\OR$ has safe
  swaps with respect to~$h$.
\end{lemma}

\begin{proof}
  Consider each pair $(\id_1, \id_2)$ of elements of $\OR$, of which there are quadratically
  many. Check in PTIME whether they are incomparable. If yes, compute in PTIME
  $\pr_\OR(\id_1, \id_2)$, and consider each pair $p$, $p+1$ of consecutive
  integers (there are linearly many). For each such pair, compute 
  $h(T(\id_1), p) \oplus h(T(\id_2), p+1)$ and $h(T(\id_2), p) \oplus
  h(T(\id_1), p+1)$, and check whether are equal.

  We must only argue that these expressions can be evaluated in PTIME, but this
  follows from the PTIME-evaluability of the accumulation operator.
  Specifically, to
  evaluate, e.g., $h(T(\id_1), p) \oplus h(T(\id_2), p+1)$, we build in PTIME from~$\OR$ a list relation
  $L$ with $p+1$ tuples that are all labeled with the neutral element of the
  monoid of~$h$ except the two last ones which are labeled respectively with
  $T(\id_1)$ and~$T(\id_2)$. We then
  evaluate the accumulation operator in PTIME on~$L$ and obtain the desired
  value.
\end{proof}

Now it is easily seen that Theorem~\ref{thm:certaintyptimec} is implied by
the following claim.

\begin{proposition}
  \label{prp:auxcancel}
  If the monoid $(\calM, \oplus, \epsilon)$ is cancellative, then, for any
  po-relation $\OR$, we have $\card{\accum_{h,\oplus}(\OR)} = 1$ iff $\OR$ has
  safe swaps with respect to $\oplus$ and $h$.
\end{proposition}

Indeed, given an instance $(D, v)$ of the \cert problem for query $Q$, we can
find $\OR$ such that $\pw(\OR)=Q(D)$ in PTIME by Proposition~\ref{prp:repsys}, and we can
test in PTIME by Lemma~\ref{lem:safeptime} whether $\OR$ has safe swaps with
respect to $\oplus$ and $h$. If it does not, then, by the above claim, we know
that $v$ cannot be certain, so $(D, v)$ is not a positive instance of \cert. If
it does, then, by the above claim, $Q(D)$ has only one possible result, so to
determine whether $v$ is certain it suffices to compute any linear extension of
$\OR$, obtaining one possible world $L$ of~$Q(D)$, and checking whether
accumulation on $L$ yields $v$. If it does not, then $(D, v)$ is not a positive
instance of \cert. If it does, then as this is the only possible result, $(D,
v)$ is a positive instance of \cert.

We now prove this claim:

\begin{proof}[Proof of Proposition~\ref{prp:auxcancel}]
  For one direction, assume that $\OR$ does \emph{not} have the safe swaps
  property. Hence, there exist two incomparable elements $t_1$ and $t_2$ in
  $\OR$ and a pair of consecutive integers $p, p+1$ in $\pr_\OR(t_1, t_2)$ such
  that the following disequality holds:
  \[h(t_1, p) \oplus h(t_2, p+1) \neq h(t_2, p) \oplus h(t_1, p+1)\]
  We use Lemma~\ref{lem:achieverank}
  to compute
  two possible worlds $L$ and $L'$ of $\OR$ that are identical except that $t_1$
  and $t_2$ occur respectively at positions $p$ and $p+1$ in~$L$, and at
  positions $p+1$ and $p$ respectively in~$L'$. We then use cancellativity (as
  in the same proof) to deduce that $L$ and $L'$ are possible worlds of~$\OR$
  that yield different accumulation results $w \neq w'$, so we conclude that $\card{\accum_{h,\oplus}(\OR)} > 1$.

  \medskip

  For the converse direction, assume that $\OR$ has the safe swaps property.
  Assume by way of contradiction that there are two possible worlds $L_1$ and
  $L_2$ of~$\OR$ such that the result of accumulation on $L_1$ and on $L_2$,
  respectively $w_1$ and $w_2$, are different, i.e., $w_1 \neq w_2$.
  Take $L_1$ and $L_2$ to have the longest possible common prefix,
  i.e., the first position $i$ such that tuple $i$ of $L_1$ and tuple $i$ of
  $L_2$ are different is as large as possible. Let $i_0$ be the length of the
  common prefix. Let $\OR'$ be $\OR$ but removing the
  elements enumerated in the common prefix of $L_1$ and $L_2$, and let $L_1'$
  and $L_2'$ be $L_1$ and $L_2$ without their common prefix. Let $t_1$ and
  $t_2$, $t_1 \neq t_2$, be the first elements respectively of $L_1'$ and
  $L_2'$; it is immediate that $t_1$ and $t_2$ are roots of $\OR'$, that is, no
  element of $\OR'$ is less than them. Further, it is clear
  that accumulation over $L_2'$ (but offsetting all ranks by $i_0$) and
  accumulation over $L_1'$ (also offsetting all ranks by $i_0$),
  respectively $w_1'$ and $w_2'$, are different, because, by the contrapositive
  of cancellativity, combining them with the accumulation result of the common
  prefix leads to the different accumulation results $w_1$ and $w_2$.

  Our goal is to
  construct a possible world $L_3'$ of $\OR'$ whose first element is $t_1$
  but such that the result of accumulation on $L_3'$ is $w_2'$. If we can build
  such an $L_3'$, then combining it
  with the common prefix will give a possible world $L_3$ of $\OR$ such that the
  result of accumulation on $L_3$ is $w_2 \neq w_1$, yet $L_1$ and $L_3$
  have a common prefix of length $>i_0$, contradicting minimality. Hence, it
  suffices to show how to construct such a $L_3'$.

  As $t_1$ is a root of $\OR'$, $L_2'$ must enumerate $t_1$, and all elements
  before $t_1$ in $L_2'$ must be incomparable to $t_1$. Write these elements as
  $L_2'' = s_1, \ldots, s_m$, and write $L_2'''$ the
  sequence following $t_1$, so that $L_2'$ is the concatenation of $L_2''$,
  $\singleton{t_1}$, and $L_2'''$. We now consider the following sequence of
  list relations, which are clearly possible worlds of $\OR'$:

  \begin{itemize}
    \item $s_1 \ldots s_m t_1 L_2'''$
    \item $s_1 \ldots s_{m-1} t_1 s_m L_2'''$
    \item $s_1 \ldots s_{m-2} t_1 s_{m-1} s_m L_2'''$
    \item $s_1 \ldots s_{m-3} t_1 s_{m-2} \ldots  s_m L_2'''$
    \item $\vdots$
    \item $s_1 \ldots s_{3} t_1 s_{4} \ldots  s_m L_2'''$
    \item $s_1 s_{2} t_1 s_{3} \ldots  s_m L_2'''$
    \item $s_1 t_1 s_{2} \ldots  s_m L_2'''$
    \item $t_1 s_{1} \ldots  s_m L_2'''$
  \end{itemize}

  We can see that any consecutive pair in this list achieves the same
  accumulation result. Indeed, it suffices to show that the accumulation result for the only
  two contiguous indices where they differ is the same, and this is exactly what
  the safe swaps property for $t_1$ and $s_j$ says, as it is easily checked that
  $j, j+1 \in \pr_{\OR'}(s_j, t_1)$, so that $j+i_0, j+i_0+1 \in \pr_\OR(s_j, t_1)$. 
  Now, the first list relation in the list is $L_2'$, and the last list relation in this list is our desired $L_3'$. This
  concludes the second direction of the proof.

  Hence, the desired equivalence is shown.
\end{proof}

This finishes the proof of Proposition~\ref{prp:auxcancel}, which, as we argued,
concludes the proof of Theorem~\ref{thm:certaintyptimec}.

\end{toappendix}

Hence, \cert is tractable for
\PosRA (Theorem~\ref{thm:certptime}), via the concatenation monoid, and 
\cert is also tractable for top-$k$ (defined in Example~\ref{ex:variants}).
The hardness of \poss for \PosRA
(Theorem~\ref{thm:posscomp1}) then implies that \poss, unlike \cert,
is hard even on cancellative monoids.

\subparagraph*{Other restrictions on accumulation.}
\begin{toappendix}
  \subsection{Other restrictions on accumulation}\label{sec:restracc}
\end{toappendix}

We next revisit the results of Section~\ref{sec:fpt2} for
\PosRAacc. However, we need to make other
assumptions on accumulation (besides PTIME-evaluability).
First, in the next results in this section, we 
assume that the accumulation monoid is \emph{finite}:

\begin{definition}
    \label{def:accfin}
    A \PosRAacc query is said to perform \emph{finite} accumulation if the
    accumulation monoid $(\calM, \oplus, \epsilon)$ is finite.
\end{definition}

For instance, if the domain of the output is assumed to be fixed (e.g., ratings
in $\{1, \ldots, 10\}$), then
select-at-$k$ and top-$k$ (the latter for fixed $k$),
as defined in Example~\ref{ex:variants}, are finite.

Second, for some of the next results, we require \emph{position-invariant
accumulation}, namely, that the accumulation map does not depend on
the absolute position of tuples:

\begin{definition}
    \label{def:ri}
    Recall that the accumulation map $h$ has in general two inputs: a tuple and
    its position. A \PosRAacc query is said to be \emph{position-invariant} if its
    accumulation map ignores the
    second input, so that effectively its only input is the tuple itself.
\end{definition}

Note that 
accumulation in the monoid is still performed in order, so we can still perform,
e.g., concatenation.
These two restrictions do not suffice to make \poss and \cert tractable (see
Appendix~\ref{apx:possfrihypo}), but we will use them
to lift the results of Section~\ref{sec:fpt2}.

\begin{toappendix}
  \label{apx:possfrihypo}

  We show the additional claim that assuming finiteness and position-invariance
  of accumulation does not suffice to make \poss or \cert tractable.
  Specifically, we show the following two results:

  \begin{theoremrep}\label{thm:possfrihypoposs}
  There is a \PosRAagg{} query performing finite and position-invariant accumulation
  for which \poss is NP-hard 
        even assuming that the input po-database contains only totally ordered po-relations.
  \end{theoremrep}

  \begin{theoremrep}\label{thm:possfrihypocert}
  There is a \PosRAagg{} query performing finite and position-invariant accumulation
  for which \cert is coNP-hard 
        even assuming that the input po-database contains only totally ordered po-relations.
  \end{theoremrep}

  We will first show the result about \poss (Theorem~\ref{thm:possfrihypoposs}),
  and then use it to show the result about \cert
  (Theorem~\ref{thm:possfrihypocert}).

    \subsubsection{Proof of Theorem~\ref{thm:possfrihypoposs} for \poss}
    \label{sec:possfrihypoposs}
    We show the following strengthening of Theorem~\ref{thm:possfrihypoposs}, which 
will be useful to prove the result for \cert in
Appendix~\ref{sec:possfrihypocert}.

\begin{proposition}
  \label{prp:posscerthard}
  There is a \PosRAagg{} query $Q_{\a}$ with finite and position-invariant accumulation such that the \poss problem is NP-hard
  for~$Q_{\a}$, even assuming that all input po-relations are totally ordered.
  Further, for any input po-database $D$ (no matter whether the relations are
  totally ordered or not), we have $\card{Q_{\a}(D)} \leq 2$.
\end{proposition}

Define the following finite domains:

\begin{itemize}
  \item $\calD_- \defeq \{\s_-, \n_-, \e_-\}$;
  \item $\calD_+ \defeq \{\s_+, \n_+, \e_+\}$;
  \item $\calD_{\pm} \defeq \calD_- \sqcup \calD_+ \sqcup \{\l, \r\}$ (the additional
    elements stand for ``left'' and ``right'').
\end{itemize}

Define the following regular expression on $\calD_{\pm}^*$, and call
\deft{balanced} a word that satisfies it:
\[
  e \defeq \l \left(\s_- \s_+ | \n_- \n_+ | \e_- \e_+\right)^* \r
\]

We now define the following problem for any \PosRA query:

\begin{definition}
  The \deft{balanced checking problem} for a \PosRA query $Q$ asks, given a po-database $D$ of
  po-relations over $\calD_{\pm}$, whether there is $L \in \pw(Q(D))$ such that $L$
  is balanced (i.e., can be seen as a word over $\calD_{\pm}$
  that satisfies $e$).
\end{definition}

Note that the balanced checking problem only makes sense (i.e., is not vacuously
false) for unary queries (i.e., queries whose output arity is~$1$) whose output
tuples have value in~$\calD_{\pm}$.

We also introduce the following regular
expression: $e' \defeq \l\, \calD_{\pm}^*\, \r$, which we will use later to guarantee that
there are only two possible worlds. We show the following lemma:

\begin{lemma}
  \label{lem:balhard}
  There exists a \PosRA query $Q_\b$ over po-databases with domain in
  $\calD_{\pm}$ such that the balanced checking problem for $Q_\b$ is
  NP-hard, even when all input po-relations are totally ordered. Further, $Q_\b$ is such that, for any input po-database~$D$,
  all possible worlds of~$Q_\b(D)$ satisfy $e'$.
\end{lemma}

To prove this lemma, we construct the query $Q'_\b(R, T) \defeq
\singleton{\l}\cupcat((R \cupgen T)\cupcat \singleton{\r})$,
i.e., $Q'_\b(R, T)$ is
the parallel composition of $R$ and $T$, preceded by $\l$ and followed by
$\r$. Recall the definition of $\cupcat$ (Definition~\ref{def:concat}), and
recall from Lemma~\ref{lem:lexconcat} that $\cupcat$ can be
expressed by a \PosRA query.

We write $L^-_w$ for any word $w \in \calD_+^*$ to be the unary list
relation defined by mapping each
letter of $w$ to the corresponding letter in $\calD_-$.
We define $\OR^-_w$ as the totally ordered po-relation with $\pw(\OR^-_w) =
\{L^-_w\}$. We claim the following:

\begin{lemma}
  \label{lem:balred}
  For any $w \in \calD_+^*$ and unary po-relation $T$ over $\calD_+$, we have $w \in
  \pw(T)$ iff $\{R \mapsto \OR^-_w, T \mapsto T\}$ is a positive instance to the
  balanced checking problem for $Q'_\b$; in other words, iff $Q'_\b(\OR^-_w, T)$ has some
  balanced possible world.
\end{lemma}

\begin{proof}
  For the first direction, assume that $w$ is indeed a possible world~$L$ of
  $T$ and let us construct a balanced possible world~$L'$ of $Q'_\b(\OR^-_w,T)$.
  $L'$ starts with $\l$. Then, $L'$ successively contains alternatively
  one
  tuple from $\OR^-_w$ (in their total order) and one from $T$ (taken in the
  order of the linear extension that yields $L$).
  Finally, $L'$ ends with $\r$. $L'$ is clearly balanced.

  For the converse direction, observe that a
  balanced possible world of $Q'_\b(\OR^-_w,T)$ must consist of first $\l$, last $\r$, and, between the two,
  tuples alternatively enumerated from  $\OR^-_w$ from one of the possible worlds
  of~$T$, with that possible world of $T$ 
  achieving $w$.
\end{proof}

We now use Lemma~\ref{lem:balred} to prove Lemma~\ref{lem:balhard}:

\begin{proof}[Proof of Lemma~\ref{lem:balhard}]
  By Theorem~\ref{thm:posscompextend1} and its proof, there is a unary query
  $Q_0$ in \PosRA{} such
  that the \poss problem for~$Q_0$ is
  NP-hard, even for input relations over $\calD_+$ (this is by observing that the
  proof uses $\{\s, \n, \e\}$ and renaming the alphabet), and even assuming that
  $D$ contains only totally
  ordered relations. Consider the query
  $Q_\b(R, D) \defeq Q'_\b(R, Q_0(D))$; $Q_\b$ is a \PosRA{} query, and by definition
  of $Q'_\b$ it satisfies the additional condition of all possible worlds satisfying
  $e'$.
  
  We reduce the \poss problem for $Q_0$ to the balanced checking problem for
  $Q_\b$ in PTIME: more specifically, we claim that $(D, w)$ is a positive instance to
  \poss for $Q_0$ iff $D'$, obtained by adding to $D$ the relation name $R$ that
  maps to the totally ordered $\OR^-_w$, is a positive instance of the balanced checking problem for
  $Q_\b$. This
  is exactly what Lemma~\ref{lem:balred} shows. This concludes the reduction, so
  we have shown that the balanced checking problem for $Q_\b$ is NP-hard, even
  assuming that the input po-database (here, $D'$) contains only totally ordered
  po-relations.
\end{proof}

Hence, all that remains to show is to prove Proposition~\ref{prp:posscerthard}
(and hence Theorem~\ref{thm:possfrihypoposs})
using Lemma~\ref{lem:balhard}. The idea is that we will reduce the balanced
checking problem to \poss, using an accumulation operator to do
the job, which will allow us to ensure that there are at most two possible
results.
To do this, we need to introduce some new concepts.

Let $A$ be the deterministic complete finite automaton defined as
follows, which clearly recognizes the language of the regular
expression $e$, and let $S$ be its state space:

\begin{itemize}
  \item there is a $\l$-transition from the initial state $q_{\i}$
    to a state $q_0$;
  \item there is a $\r$-transition from $q_0$ to the final state $q_{\f}$;
  \item for $\alpha \in \{\s, \n, \e\}$:
    \begin{itemize}
      \item there is an $\alpha_+$-transition from $q_0$ to a state $q_\alpha$;
      \item there is an $\alpha_-$-transition from $q_\alpha$ to $q_0$;
    \end{itemize}
\item all other transitions go to a sink state $q_{\bot}$.
\end{itemize}

We now define the \deft{transition monoid} of this automaton, which is a finite
monoid (so we are indeed performing finite accumulation). 
Let $\calF_S$ be the finite set of
total functions from~$S$ to $S$, and consider the monoid defined on
$\calF_S$ with the identity function $\id$ as the neutral element, and
with function composition $\circ$ as the (associative) binary operation. We
define inductively a mapping $h$ from $\calD_{\pm}^*$ to
$\calF_S$ as follows, which can be understood as a homomorphism from
the free monoid $\calD_{\pm}^*$ to the transition monoid of~$A$:

\begin{itemize}
  \item For $\epsilon$ the empty word, $h(\epsilon)$ is the identity
    function $\mathrm{id}$.
  \item For $a \in \calD_{\pm}$, $h(a)$ is the transition table for symbol
    $a$ for the automaton $A$, i.e., the function that maps each state $q \in S$
    to the one state $q'$ such that there is an $a$-labeled transition from $q$
    to $q'$; the fact that $A$ is deterministic and complete is what ensures
    that this is well-defined.
  \item For $w \in \calD_{\pm}^*$ and $w \neq \epsilon$, writing $w = a w'$ with
    $a \in \calD_{\pm}$, we define $h(w) \defeq h(w') \circ h(a)$.
\end{itemize}

It is easy to show inductively that, for any $w \in \calD_{\pm}^*$,
for any $q \in S$, $(h(w))(q)$ is the state that we reach in $A$
when reading word $w$ from state $q$. We will identify two special
elements of $\calF_S$:

\begin{itemize}
  \item $f_0$, the function mapping every state of $S$ to the sink state
    $q_{\bot}$;
  \item $f_1$, the function mapping the initial state $q_{\i}$ to the final
    state $q_{\f}$, and mapping every other state in $S \backslash \{q_{\i}\}$
    to $q_{\bot}$.
\end{itemize}

Recall the definition of the regular expression $e'$ earlier. We
claim the following property on the automaton $A$:

\begin{lemma}
  \label{lem:transmon}
  For any word $w \in \calD_{\pm}^*$ that matches $e'$, we have $h(w) =
  f_1$ if $w$ is balanced (i.e., satisfies $e$) and $h(w) = f_0$ otherwise.
\end{lemma}

\begin{proof}
  By definition of~$A$, for any state $q \neq q_{\i}$, we have $(h(\l))(q) =
  q_{\bot}$,
  so that, as $q_{\bot}$ is a sink state, we have
  $(h(w))(q) = q_{\bot}$ for any $w$ that satisfies $e'$. Further, by
  definition of~$A$, for any state $q$, we have $(h(\r))(q) \in \{q_{\bot},
  q_{\f}\}$, so that, for any state $q$ and $w$ that satisfies $e'$, we have
  $(h(w))(q) \in \{q_{\bot}, q_{\f}\}$. This implies that, for any word $w$
  that satisfies $e'$, we have $h(w) \in \{f_0, f_1\}$.

  Now, as we know that $A$ recognizes the language of $e$, we have the desired
  property, because, for any $w$ satisfying $e'$, $h(w)(q_{\i})$ is $q_{\f}$
  or not depending on whether $w$ satisfies $e$ or not, so $h(w)$ is $f_1$ or
  $f_0$ depending on whether $w$ satisfies $e$ or not.
\end{proof}

Hence, consider the query $Q_\b$ whose existence is guaranteed by
Lemma~\ref{lem:balhard}, and such that all its possible worlds satisfy $e'$, and
construct the query $Q_{\a} \defeq \accum_{h,\circ} (Q_\b)$ -- we see $h$ as
a position-invariant accumulation map. We
conclude the proof of Proposition~\ref{prp:posscerthard} by showing that \poss
is NP-hard for $Q_{\a}$, even when the input po-database consists only of totally ordered
po-relations; and that $\card{Q_\a(D)}\leq 2$ in any case:

\begin{proof}[Proof of Proposition~\ref{prp:posscerthard}]
  To see that $Q_{\a}$ has at most two possible results on $D$, observe that, for any
  po-database $D$, writing $Q_{\b}(D)$ as a word $w \in \calD_{\pm}$, we know that
  $w$ matches $e'$. Hence, by Lemma~\ref{lem:transmon}, we
  have $h(w) \in \{f_0, f_1\}$, so that $Q_{\a}(D) \in \{f_0, f_1\}$.

  To see that \poss in NP-hard for $Q_{\a}$ even on totally ordered po-relations, we reduce the balanced checking
  problem for $Q_{\b}$ to \poss for $Q_{\a}$ with the trivial reduction: we claim
  that for any po-database $D$, there is a balanced possible world in~$Q_{\b}(D)$ iff $f_1 \in Q_{\a}(D)$,
  which is proved by Lemma~\ref{lem:transmon} again. Hence, $Q_{\b}(D)$ is balanced
  iff $(D, f_1)$ is a positive instance of \poss for~$Q_{\a}$. This concludes the reduction.
\end{proof}

This concludes the proof of Proposition~\ref{prp:posscerthard}, hence of
Theorem~\ref{thm:possfrihypoposs}.

    \subsubsection{Proof of Theorem~\ref{thm:possfrihypocert} for \cert}
    \label{sec:possfrihypocert}
    We prove Theorem~\ref{thm:possfrihypocert} by relying on Proposition~\ref{prp:posscerthard}, proven in
Appendix~\ref{sec:possfrihypoposs}:

\begin{proof}[Proof of Theorem~\ref{thm:possfrihypocert}]
  Consider the query $Q_{\a}$ from Proposition~\ref{prp:posscerthard}. We show a
  PTIME reduction from the NP-hard problem of \poss for $Q_{\a}$ (for totally ordered
  input po-databases) to the negation of
  the \cert problem for $Q_{\a}$ (for input po-databases of the same kind). The
  query $Q_{\a}$ uses accumulation, so it is of the
  form $\accum_{h,\oplus}(Q')$.

  Consider an instance of \poss for $Q_{\a}$ consisting of an input po-database $D$
  and candidate result $v \in \calM$. Evaluate $R=Q'(D)$ in
  PTIME by Proposition~\ref{prp:repsys}, and compute in PTIME an
  arbitrary possible world $L'$ of $R$: this can be done by a topological
  sort of $R$. Let $v'=\accum_{h,\oplus}(L')$. If $v=v'$ then $(D, v)$ is a positive
  instance for \poss for $Q_{\a}$. Otherwise, we have $v\neq v'$. Now, solve
  the \cert problem for $Q_{\a}$ on the input $(D, v')$.
    If the answer is YES, then $(D, v)$ is a negative instance for \poss
    for $Q_{\a}$. Otherwise, there must exist a possible world $L''$
    in $\pw(R)$ with $v''=\accum_{h,\oplus}(L'')$ and $v''\neq
    v'$. However, we know that $\card{\pw(Q_{\a}(D))} \leq 2$
    by
      Proposition~\ref{prp:posscerthard}. Hence, as $v \neq v'$ and $v'
      \neq v''$,
        we must have $v = v''$. So $(D, v)$ is a
        positive
          instance for \poss for $Q_{\a}$.

  Thus, we have reduced \poss for $Q_{\a}$ in PTIME to the negation of \cert for
  $Q_{\a}$,
  showing that \cert for $Q_{\a}$ is coNP-hard.
\end{proof}

\end{toappendix}

\subparagraph*{Revisiting Section~\ref{sec:fpt2}.}
\begin{toappendix}
  \subsection{Revisiting Section~\ref{sec:fpt2}}
\end{toappendix}

We now extend our previous results to queries with accumulation,
for \poss and \cert, under the additional assumptions on
accumulation that we presented. We call \Plexacc and \Pnoprodacc
the
extension of \Plex and \Pnoprod
with accumulation.

\begin{toappendix}
  For the proof of the results of this paragraph, refer to the proof of
  the corresponding results in Section~\ref{sec:fpt}: Theorem~\ref{thm:aggregwa}
  is proven together with Theorem~\ref{thm:aggregw} in
  Appendix~\ref{sec:totaltract}, and
  Theorem~\ref{thm:aggregnoproda} is proven together with
  Theorem~\ref{thm:aggregnoprod} in Appendix~\ref{sec:noprodtract}.
\end{toappendix}

We can first generalize Theorem~\ref{thm:aggregw} to \Plexacc queries with
\emph{finite} accumulation:
\begin{theorem}
  \label{thm:aggregwa}
  For any \Plexacc query performing \emph{finite}
  accumulation,
  \poss and \cert are in PTIME on po-databases of bounded width.
\end{theorem}

We can then adapt the tractability result for queries without product (Theorem~\ref{thm:aggregnoprod}):

\begin{theorem}
  \label{thm:aggregnoproda}
  For any \Pnoprodacc query performing \emph{finite} and \emph{position-invariant}
  accumulation,
  \poss and \cert are in PTIME on po-databases whose relations
  have either bounded width or bounded ia-width.
\end{theorem}

The finiteness assumption is important, as the previous result
does not hold otherwise. Specifically, there exists a query
that performs \emph{position-invariant} but not \emph{finite}
accumulation, for which \poss is NP-hard even on unordered
po-relations (see Appendix~\ref{app:hardinf}).

\begin{toappendix}
  \subsection{Hardness Without the Finiteness Assumption}
  \label{app:hardinf}
  We show the additional claim that \poss for \Pnoprodacc queries can be hard if
  we do not assume finiteness. Namely, we show:

\begin{theoremrep}\label{thm:possgri}
  There is a position-invariant accumulation operator $\accum_{h, \oplus}$
  such that \poss is NP-hard for the \Pnoprodacc query $Q \defeq \accum_{h,
  \oplus}(\OR)$ (i.e., accumulation applied directly to an input po-relation
  $\OR$), even on input po-databases where $\OR$ is restricted to be an unordered relation.
\end{theoremrep}

\begin{proof}
  We consider the NP-hard partition problem: given a multiset $S$ of
  integers, decide whether it can be partitioned as $S = S_1 \sqcup S_2$ such
  that $S_1$ and $S_2$ have the same sum.
  Let us reduce an instance of the partition problem with
  this restriction to an instance of the \poss problem, in PTIME.

  Let $\calM$ be the monoid generated by
  the functions $f : x \mapsto -x$ and $g_a : x \mapsto x + a$ for $a \in \ZZ$
  under the function composition operation. We have $g_a \circ g_b = g_{a+b}$
  for all $a, b \in \NN$, $f \circ f = \id$, and $f \circ g_a = g_{-a} \circ f$,
  so we actually have $\calD = \{g_a \mid a \in \ZZ\} \sqcup \{f \circ g_a \mid
  a \in \ZZ\}$. Further, $\calM$ is actually a group, as we can define
  $(g_a)^{-1} = g_{-a}$ and $(f \circ g_a)^{-1} = f \circ g_a$ for all $a \in
  \ZZ$.

  We fix $\calD = \NN \sqcup \{-1\}$.
  We define the position-invariant accumulation map $h$ as mapping $-1$ to $f$ and $a
  \in \NN$ to $g_a$.
  We encode the partition problem instance~$S$ in PTIME
  to an unordered po-relation $\OR_S$ with a single attribute, that contains
  one tuple with value $s$ for each $s \in S$, plus two tuples with value~$-1$.
  Let the candidate result~$v$ be $\id \in \calM$, and consider the \poss
  instance for the query $\accum_{h,+}(\OR)$, on the po-database $D$ where $\OR$
  is the relation $\OR_S$, and the candidate result~$v$.

  We claim that this \poss instance is positive iff the partition problem has a
  solution. Indeed, if $S$ has a partition, 
  let $s = \sum_{i \in S_1} i = \sum_{i \in S_2} i$.
  Consider the total order on
  $\OR_S$ which enumerates the tuples corresponding to the elements of $S_1$, then
  one tuple $-1$, then the tuples corresponding to the elements of $S_2$, then
  one tuple $-1$. The result of accumulation is then $g_s \circ f \circ g_s \circ
  f$, which is $\id$.

  Conversely, assume that the \poss problem has a solution. Consider a witness
  total order of~$\OR_S$; it must a (possibly empty) sequence of tuples
  corresponding to a subset $S_1$ of~$S$, then a tuple $-1$, then a (possibly
  empty) sequence corresponding to $S_2 \subseteq S$, then a tuple~$-1$, then a
  (possibly empty) sequence corresponding to $S_1' \subseteq S$, with $S = S_1
  \sqcup S_1' \sqcup S_2$. Let $s_1$, $s_1'$ and $s_2$ respectively be the sums
  of these subsets of $S$. The result of accumulation is then $g_{s_1} \circ f
  \circ g_{s_2} \circ f \circ g_{s_1'}$, which simplifies to $g_{s_1 + s_1' -
  s_2}$. Hence, we have $s_1 + s_1' = s_2$, so that $(S_1 \sqcup S_1')$ and
  $S_2$ are a partition witnessing that $S$ is a positive instance of the
  partition problem.

  As the reduction is in PTIME, this concludes the proof.
\end{proof}

\end{toappendix}

\subparagraph*{Other definitions.}
\begin{toappendix}
  \subsection{Other definitions}
\end{toappendix}
Finally, recall that we can use accumulation as in Example~\ref{ex:variants} to
capture \emph{position-based selection}
($\text{top-}k$, $\text{select-at-}k$) and \emph{tuple-level comparison}
(whether the first occurrence of a tuple precedes all occurrences of another
tuple) for \PosRA queries. Using a direct construction for these problems, we
can show that they are tractable:

\begin{propositionrep}\label{prop:otherdef}
        For any \PosRA query $Q$, the following problems are in PTIME:
        \begin{description}
          \item[select-at-$k$:] Given a po-database $D$, tuple value $t$, and
            position $k \in \NN$, whether it is \mbox{possible/certain} that $Q(D)$ has
            value $t$ at position~$k$;
          \item[top-$k$:] For any \emph{fixed} $k \in \NN$, given a po-database
            $D$ and list relation $L$ of length $k$, whether it is
            possible/certain that the top-$k$ values in $Q(D)$ are exactly $L$;
          \item[tuple-level comparison:] Given a po-database $D$ and two tuple
            values $t_1$ and $t_2$,
            whether it is possible/certain that the first occurrence of~$t_1$
            precedes all occurrences of~$t_2$.
        \end{description}
\end{propositionrep}

\begin{proof}
  To solve each problem, we first compute the po-relation $\OR \defeq Q(D)$ in
  PTIME by Proposition~\ref{prp:repsys}.
  We now address each problem in turn.

  \begin{description}
    \item[select-at-$k$:]
    Considering the po-relation $\OR = (\ID, T, <)$, we can compute in PTIME,
    for every element $\id \in \ID$, its \emph{earliest
		index} $\i^-(\id)$, which is its number of ancestors by $<$ plus one, and its \emph{latest
		index} $\i^+(\id)$, which is the number of elements of $\OR$ minus the number of
	descendants of~$\id$. It is easily seen that for any element $\id \in \ID$, there is a
        linear extension of~$\OR$ where $\id$ appears at position $\i^-(\id)$ (by
        enumerating first exactly the ancestors of~$\id$), or at position
        $\i^+(\id)$ (by enumerating first everything except the descendants
        of~$\id$), or in fact at any position of $[\i^-(\id), \i^+(\id)]$, the
        \emph{interval} of~$\id$ (this is by enumerating first the ancestors
        of~$\id$, and then as many elements as needed that are incomparable
        to~$\id$, along a linear extension of these elements).
	
	Hence, select-at-$k$ possibility for tuple $t$ and position $k$ can be
        decided by checking, for each $\id \in \ID$ such that $T(\id) = t$,
        whether $k \in [\i^-(\id), \i^+(\id)]$, and answering YES iff we can find
        such an~$\id$. For select-at-$k$ certainty, we answer NO iff we can find
        an $\id \in \ID$ such that $k \in [\i^-(\id), \i^+(\id)]$ but we have
        $T(\id) \neq t$.

      \item[top-$k$:]
        Considering the po-relation $\OR = (\ID, T, <)$,
  we consider each sequence of $k$ elements of~$\OR$, of which there are at most
      $\card{ID}^k$, i.e., polynomially many, as $k$ is fixed. To solve possibility for top-$k$, we
  consider each such sequence $\id_1, \ldots, \id_k$ such that $(T(\id_1),
  \ldots, T(\id_k))$ is equal to the candidate list relation $L$, and we check if this sequence is indeed a prefix of a linear extension
  of~$\OR$, i.e., whether, for each $i \in \{1, \ldots, k\}$, for any $\id \in
  \ID$ such that $\id < \id_i$, if $\id_i \in \{\id_1, \ldots, \id_{i-1}\}$,
  which we can do in PTIME. We answer YES iff we can find such a sequence.

  For certainty, we consider each sequence $\id_1, \ldots, \id_k$ such that
  $(T(\id_1), \ldots, T(\id_k)) \neq L$, and we check whether it is a prefix of
  a linear extension in the same way: we answer NO iff we can find such a
  sequence.
	
\item[tuple-level comparison:]
  We are given the two tuple values $t_1$ and $t_2$, and we assume that both are
  in the image of~$T$, as the tuple-level comparison problem is vacuous
  otherwise.

  For possibility, given the two tuple values $t_1$ and $t_2$, we consider each
  $\id \in \ID$ such that $T(\id) = t_1$, and for each of them, we construct
  $\OR_\id \defeq (\ID, T, {<_\id})$ where ${<_\id}$ is the transitive closure
  of ${<} \cup \{(\id, \id') \mid \id' \in \ID, T(\id') = t_2\}$. We answer YES iff one of the
  $\OR_\id$ is indeed a po-relation, i.e., if $<_\id$ as defined does not
  contain a cycle. This is correct, because it is possible that the first
  occurrence of~$t_1$ precedes all occurrences of~$t_2$ iff there is some
  identifier $\id$ with tuple value $t_1$ that precedes all identifiers with
  tuple value $t_2$, i.e., iff one of the $\OR_\id$ has a linear extension.

  For certainty, given $t_1$ and $t_2$, we answer the negation of possibility
  for $t_2$ and $t_1$. This is correct because certainty is false iff there is a
  linear extension of~$\OR$ where the first occurrence of~$t_1$ does not precede
  all occurrences of~$t_2$, i.e., iff there is a linear extension where the
  first occurrence of $t_2$ is not after an occurrence of~$t_1$, i.e., iff some
  linear extension is such that the first occurrence of $t_2$ precedes all
  occurrences of~$t_1$, i.e., iff possibility is true for $t_2$
      and~$t_1$. \qedhere
  \end{description}
\end{proof}

\section{Extensions}\label{sec:extensions}
We next briefly consider two extensions to our model: group-by and
duplicate elimination. 

\subparagraph*{Group-by.} First, we extend accumulation with a \emph{group-by}
operator, inspired by SQL.
\begin{definition}
    Let $(\calM, \oplus, \epsilon)$ be a monoid and $h : \calD^k \to \calM$ be
    an accumulation map (cf.\ Definition \ref{def:aggregationb}), and
    let 
    $\mathbf{A} = A_1,...,A_n$ be a
    sequence of attributes: we call $\accumgby_{h, \oplus, \mathbf{A}}$ an \emph{accumulation
        operator with group-by}.
    Letting $L$ be a list relation with compatible schema,
    we define $\accumgby_{h, \oplus, \mathbf{A}}(L)$ as
    an \emph{unordered} relation
    that has, for each tuple value $t
    \in \pi_{\mathbf{A}}(L)$, one tuple $\langle t, v_t \rangle$ where $v_t$ is $\accum_{h,
        \oplus}(\sigma_{A_1 = t.A_1,...A_n=t.A_n}(L))$ with $\pi$ and $\sigma$ on
    the list relation $L$ having the expected semantics. The result on a
    po-relation $\OR$ is the set of unordered relations $\{\accumgby_{h, \oplus, \mathbf{A}}(L) \mid L \in
    \pw(\OR)\}$.
\end{definition}

In other words, the operator ``groups by'' the values of $A_1,...,A_n$, and
performs accumulation within each group, forgetting the order
across groups. As for standard accumulation, we only allow 
group-by as an outermost operation,
calling \PosRAaccgby the language of \PosRA queries followed by one accumulation
operator with group-by. Note that the set of possible results is generally not 
a po-relation,
    because the underlying bag relation is not certain.

We next study the complexity of \poss and \cert for \PosRAaccgby queries. Of course,
whenever \poss and \cert are hard for some \PosRAacc query $Q$ on some kind of
input po-relations, then there is a corresponding \PosRAaccgby query
for which hardness also holds (with empty $\mathbf{A}$). The main point of this section is to show that
the converse is not true: the addition of group-by increases
complexity. Specifically, we show that the \poss problem for \PosRAaccgby is
hard even on totally ordered po-relations and without the $\times_\dir$
operator:

\begin{toappendix}
    \subsection{Proof of Theorem~\ref{thm:hardpossgby}: Hardness of \poss with
        Group-By}
\end{toappendix}

\begin{theoremrep}\label{thm:hardpossgby}
    There is a \PosRAaccgby query $Q$ 
    with finite and position-invariant
    accumulation, not using $\times_\dir$, such that \poss
    for $Q$ is NP-hard even on totally ordered po-relations.
\end{theoremrep}

This result contrasts with the tractability of \poss for \Plex queries
(Theorem~\ref{thm:aggregwCorr})
 and for \Plexacc queries with finite
accumulation (Theorem~\ref{thm:aggregwa})
on totally ordered po-relations.

\begin{proof}
  Let $Q$ be the query $\accumgby_{\oplus, h, \{1\}}(Q')$, where we define:
  \[
    Q' \defeq \Pi_{3,4}(\sigma_{1=2}(R \times_\lex
    S_1 \cup S_2 \cup S_3))
  \]
  In the accumulation operator, the accumulation map $h$ maps each tuple $t$ to its
  second component. Further, we define the finite monoid $\calM$ to be the
  \emph{syntactic monoid} \cite{pin1997syntactic} of the language defined by the regular expression $\s (\l_+\l_- | \l_-
  \l_+)^* \e$, where $\s$ (for ``start''), $\l_-$ and $\l_+$, and $\e$ (for
  ``end'') are fresh values from $\calD$: this monoid 
  ensures that, 
  for any non-empty word $w$ on the alphabet $\{\s, \l_-, \l_+, \e\}$ that
  starts with $\s$ and ends with $\e$, the word $w$
  evaluates to $\epsilon$ in $\calM$ iff $w$ matches this regular expression.

  We reduce from the NP-hard 3-SAT problem: we are given a conjunction of
  clauses $C_1, \ldots, C_n$, with each clause being a disjunction of three
  literals, namely, a variable or negated variable among $x_1, \ldots, x_m$, and
  we ask whether there is a valuation of the variables such that the clause is
  true. We fix an instance of this problem. We assume without loss of generality
  that the instance has been preprocessed to ensure that no clause contained two
  occurrences of the same variable (neither with the same polarity nor with
  different polarities).

  We define the relation $R$ to be $\ordern{m+3}$. The totally ordered relations
  $S_1$, $S_2$, and $S_3$ consist of $3m+2n$ tuple values, which we define in a
  piecewise fashion:

  \begin{itemize}
    \item First, for the tuples with positions from $1$ to $m$ (the ``opening gadget''):
      \begin{itemize}
        \item The first coordinate is $1$ for all tuples in $S_1$ and $0$ for
          all tuples in $S_2$ and $S_3$ (which do not join with $R$);
        \item The second coordinate is $i$ for the $i$-th tuple in $S_1$ (and
          irrelevant for tuples in $S_2$ and $S_3$);
        \item The third coordinate is $\s$ for all these tuples.
      \end{itemize}
      The intuition for the opening gadget is that it ensures that accumulation in
      each of the $m$ groups will start with the start value $\s$, used to
      disambiguate the possible monoid values and ensure that there is exactly
      one correct value.
    \item For the tuples with positions from $m+1$ to $2m$ (the ``variable
      choice'' gadget):
      \begin{itemize}
        \item The first coordinate is $2$ for all tuples in~$S_1$ and~$S_2$ and
          $0$ for all tuples in~$S_3$ (which do not join with $R$):
        \item The second coordinate is $i$ for the $(m+i)$-th tuple in $S_1$ and in
          $S_2$;
        \item The third coordinate is $\l_-$ for all tuples in~$S_1$ and $\l_+$
          for all tuples in~$S_2$.
      \end{itemize}
      The intuition for the variable choice gadget is that, for each group, we
      have two incomparable elements, one labeled $\l_-$ and one labeled $\l_+$.
      Hence, any linear extension must choose to enumerate one after the other,
      committing to a valuation of the variables in the 3-SAT instance;
      to achieve the candidate possible world, the linear extension will then
      have to continue enumerating the elements of this group in the correct
      order.
    \item For the tuples with positions from $2m+1$ to $2m + 2n$ (the ``clause
      check'' gadget), for each $1 \leq j \leq n$, letting $j' \defeq 2n + j +
      1$, we describe tuples $j'$ and $j'+1$ in $S_1$, $S_2$, $S_3$:
      \begin{itemize}
        \item The first coordinate is $j+2$;
        \item The second coordinate carries values in
          $\{a, b, c\}$, where we write clause $C_j$ as $\pm x_a \vee \pm x_b
          \vee \pm x_c$. Specifically:
          \begin{itemize}
            \item Value $a$ is assigned to tuple $j'+1$ in relation $S_1$ and
              tuple $j'$ in relation $S_2$;
            \item Value $b$ is assigned to tuple $j'+1$ in relation $S_1$ and
              tuple $j'$ in relation $S_2$;
            \item Value $c$ is assigned to tuple $j'+1$ in relation $S_1$ and
              tuple $j'$ in relation $S_2$;
          \end{itemize}
        \item The third
          coordinate carries values in $\{\l_-, \l_+\}$; namely, writing $C_j$ as above:
        \begin{itemize}
          \item Tuple $j'+1$ in relation $S_1$ carries $\l_+$ if
            variable $x_a$ occurs positively in $C_j$, and $\l_-$ otherwise; tuple
            $j'$ in relation $S_2$ carries the other value;
          \item The elements at the same positions in relation $S_2$ and $S_3$,
            respectively in $S_3$ and $S_1$, are defined in the same way depending
            on the sign of $x_b$, respectively of $x_c$.
        \end{itemize}
      \end{itemize}
      The intuition for the clause check gadget is that, for each $1 \leq j \leq
      n$, the tuples at levels $j'$ and $j'+1$ check that clause $C_j$ is
      satisfied by the valuation chosen in the variable choice gadget.
      Specifically, if we consider the order constraints on the two elements
      from the same group (i.e., second coordinate) which are implied by the
      order chosen for this variable in the variable choice gadget, the
      construction ensures that these order constraints plus the comparability
      relations of the chains imply a cycle (that is, an impossibility) iff the
      clause is violated by the chosen valuation.
    \item For the tuples with positions from $2n + 2m + 1$ to $3n + 2m$ (the
      ``closing gadget''), the definition is like the opening gadget but
      replacing $\e$ by $\s$, namely:
      \begin{itemize}
        \item The first coordinate is $m+3$ for all tuples in $S_1$ and $0$ for
          all tuples in $S_2$ and $S_3$ (which again do not join with $R$);
        \item The second coordinate is $i$ for the $i$-th tuple in $S_1$;
        \item The third coordinate is $\e$ for all these tuples.
      \end{itemize}
      The intuition for the closing gadget is that it ensures that accumulation
      in each group ends with value $\e$.
  \end{itemize}

  We define the candidate possible world to consist of a list relation of $n$
  tuples; the $i$-th tuple carries value $i$ as its first component and the
  acceptation value from the monoid $\calM$. The reduction that we described is
  clearly in PTIME, so all that remains is to show correctness of the reduction.
  
  To do so, we first describe the result of evaluating $\OR \defeq Q'(R, S_1, S_2, S_3)$ on
  the relations described above. Intuitively, it is just like $\Pi_{2, 3} (\sigma_{2 \neq
  \text{``0''}}(S_1 \cup S_2 \cup S_3))))$, but with the following additional
  comparability relations: all tuples in all chains whose first coordinate
  carried a value $i$ are less than all tuples in all chains whose first
  coordinate carried a value $j > i$. In other words, we add comparability
  relations across chains as we move from one ``first component'' value to the
  next. The point of this is that it forces us to enumerate the tuples of the
  chains in a way that ``synchronizes'' across all chains whenever we change the
  first component value. Observe that, in keeping with
  Proposition~\ref{prp:lexwidth}, the width of $\OR$ has a constant bound,
  namely, it is~$3$.

  Let us now show the correctness of the reduction. For the forward direction,
  consider a valuation $\nu$ that satisfies the 3-SAT instance. Construct the
  linear extension of~$\OR$ as follows:

  \begin{itemize}
    \item For the start gadget, enumerate all tuples of $S_1$ in the prescribed
      order. Hence, the current accumulation result in all $n$ groups is $\s$.
    \item For the variable choice gadget, for all $i$, enumerate the $i$-th
      tuples of $S_1$ and $S_2$ of the gadget in an order depending on
      $\nu(x_i)$: if $\nu(x_i)$ is $1$, enumerate first the tuple of~$S_1$ and
      then the tuple of~$S_2$, and do the converse if $\nu(x_i) = 0$. Hence, for
      all $1 \leq i \leq n$, the
      current accumulation result in group~$i$ is $\s \l_- \l_+$ if $\nu(x_i)$
      is $1$ and $\s \l_+ \l_-$ otherwise.
    \item For the clause check gadget, we consider each clause in order, for $1
      \leq j \leq n$, maintaining the property that, for each group $1 \leq i
      \leq n$, the current accumulation result in group $i$ is of the form $\s
      (\l_- \l_+)^*$ if $\nu(x_i) = 1$ and $\s (\l_+ \l_-)^*$ otherwise.
      
      Fix
      a clause $C_j$, let $j' \defeq 2n + j + 1$ as before, and study the tuples
      $j'$ and $j'+1$ of~$S_1, S_2, S_3$. As $C_j$ is satisfied under $\nu$, let
      $x_d$ be the witnessing literal (with $d \in \{a, b, c\}$), and let $d'$
      be the index (in $\{1, 2, 3\}$) of variable~$d$. Assume that $x_d$
      occurs positively; the argument is symmetric if it occurs negatively. By
      definition, we must have $\nu(x_d) = 1$, and by construction tuple $j'$ in relation
      $S_{1 + (d' + 1 \text{ mod } 3)}$ must carry value $\l_-$ and it is in
      group $d$. Hence, we can enumerate it and group $d$ now carries a value of
      the form $\s (\l_- \l_+)^* \l_-$. Now, letting $x_e$ be the $1 + (d' + 1
      \text{ mod } 3)$-th variable of $\{x_a, x_b, x_c\}$, the two elements of
      group $e$ (tuple $j'+1$ of
      $S_{1 + (d' + 1 \text{ mod } 3)}$ and tuple $j'$ of $S_{1 + (d' + 1 \text{
        mod } 3)}$) both had all their predecessors enumerated; so we can
      enumerate them in the order that we prefer to satisfy the condition on the
      accumulation values; then we enumerate likewise the two elements in the
      remaining group in the order that we prefer, and last we enumerate the
      second element of group~$d$; so we have satisfied the invariants.
    \item Last, for the end gadget, we enumerate all tuples of~$S_1$ and we have
      indeed obtained the desired accumulation result.
  \end{itemize}

  This concludes the proof of the forward direction.

  \medskip

  For the backward direction, consider any linear extension of~$\OR$. Thanks to
  the order constraints of $\OR$, the linear extension must enumerate tuples in
  the following order:

  \begin{itemize}
    \item First, all tuples of the start gadget.
    \item Then, all tuples of the variable choice gadget. We use this to define
      a valuation $\nu$: for each variable $x_i$, we set $\nu(x_i) = 1$ if the
      tuple of $S_1$ in group $i$ was enumerated before the one in group~$S_2$,
      and we set $\nu(x_i) = 0$ otherwise.
    \item Then, for each $1 \leq j \leq n$, in order, tuples $2n + j + 1$ of
      $S_1$, $S_2$, $S_3$.

      Observe that this implies that, whenever we enumerate such tuples, it must
      be the case that the current accumulation value for any variable $x_i$ is
      of the form $\s (\l_- \l_+)^*$ if $\nu(x_i) = 1$, and $\s (\l_+ \l_-)^*$
      otherwise. Indeed, fixing $1 \leq i \leq n$, assume that we are in the
      first case (the second one is symmetric). In this case, the accumulation
      state for $x_i$ after the variable choice gadget was $\s \l_- \l_+$, and
      each pair of levels in the clause check gadget made us enumerate either
      $\epsilon$ (variable $x_i$ did not occur in the clause) or one of
      $\l_-\l_+$ or $\l_+\l_-$ (variable $x_i$ occurred in the clause); as the
      3-SAT instance was preprocessed to ensure that each variable occurred only
      at most once in each clause, this case enumeration is exhaustive. Hence,
      the only way to obtain the correct accumulation result is to always
      enumerate $\l_- \l_+$, as if we ever do the contrary the accumulation
      result can never satisfy the regular expression that it should satisfy.

    \item Last, all tuples of the end gadget.
  \end{itemize}

  What we have to show is that the valuation $\nu$ thus defined indeed satisfies
  the formula of the 3-SAT instance. Indeed, fix $1 \leq j \leq n$ and consider
  clause $C_j$. Let $S_i$ be the first relation where the linear extension
  enumerated a tuple for the clause check of~$C_j$, and let $x_d$ be its
  variable (where $d$ is its group index). If $\nu(x_d) = 1$, then the
  observation above implies that the label of the enumerated element must be
  $\l_-$, as otherwise the accumulation result cannot be correct. Hence, by
  construction, it means that variable $x_d$ must occur positively in $C_j$, so
  it witnesses that $\nu$ satisfies $C_j$. If $\nu(x_d) = 0$, the reasoning is
  symmetric. This concludes the proof in the backwards direction, so we have
  established correctness of the reduction, which concludes the proof.
\end{proof}

By contrast, it is not hard to see that the \cert problem for \PosRAaccgby
reduces to \cert for the same query without group-by, so it is no harder than
the latter problem. Specifically:

\begin{toappendix}
    \subsection{Proof of Theorem~\ref{thm:easycertgby}: Tractability of \cert with
        Group-By}
\end{toappendix}

\begin{theoremrep}\label{thm:easycertgby}
    All \cert tractability results from Section \ref{sec:fpt} extend to \PosRAaccgby
    when imposing the same restrictions on query operators, accumulation, and input
    po-relations.
\end{theoremrep}

\begin{toappendix}
  We show the following auxiliary result:

  \begin{proposition}
    For any \PosRAaccgby query $Q \defeq \accumgby_{h, \oplus, P} Q'$ and family
    $\calD$ of po-databases, the \cert problem for $Q$ on input po-databases
    from $\calD$
    reduces in PTIME to the \cert problem for $\accum_{h, \oplus} R$ (where
    $\OR$ is a po-relation name), on the family $\calD'$ of po-databases mapping
    the name $\OR$ to a subset of a po-relation of $\{Q'(D) \mid D \in \calD\}$.
  \end{proposition}

\begin{proof}
  To prove that, consider an instance of \cert for $Q$, defined by an input
  po-database $D$ of $\calD$
  and candidate possible world $L$. We first evaluate $\OR' \defeq Q'(D)$ in
  PTIME. Now, for each tuple value $t$ in $\pi_P(\OR')$, let $\OR_t$ be the
  restriction of~$\OR'$ to the elements matching this value; note that the
  po-database mapping $R$ to $\OR_t$ is indeed in the family $\calD'$. We solve \cert for
  each $\accum_{h, \oplus} \OR_t$ in PTIME with the candidate possible world
  obtained from $L$ by extracting the accumulation value for that group, and
  answer YES to the original \cert instance iff all these invocations answer
  YES. As this process is clearly in PTIME, we must show correctness of the
  relation.

  For one direction, assume that each of the invocations answers YES,
  but the initial instance to \cert was negative. Consider two linear extensions
  of $\OR'$ that achieve different accumulation results and witness that the
  initial instance was negative, and consider a group $t$ where these
  accumulation results for these two linear extensions differ. Considering the
  restriction of these linear extensions to that group, we obtain the two
  different accumulation values for that group, so that the \cert invocation for
  $\OR_t$ should not have answered YES.

  For the other direction, assume that invocation for tuple~$t$ does not answer YES, then
  considering two witnessing linear extensions for that invocation, and
  extending them two linear extensions of~$\OR'$ by enumerating other tuples in
  an indifferent way, we obtain two different accumulation results for~$Q$ which
  differ in their result for~$t$. This concludes the proof.
\end{proof}

This allows us to show Theorem~\ref{thm:easycertgby} by considering all results
of Section~\ref{sec:fpt} in turn, and showing that they extend to \PosRAaccgby
queries, under the same restrictions on operators, accumulation, and input
po-relations:

\begin{itemize}
\item Theorem~\ref{thm:certaintyptimec} extends, because
\cert is
tractable on any family $\calD'$ of input po-databases, so tractability for
\PosRAaccgby holds for any family $\calD$ of input po-databases. 
\item Theorem~\ref{thm:aggregwa} extends, because, for any
family $\calD$ of po-databases whose po-relations have width at most $k$ for some
$k \in \Nat$, we know by 
Proposition~\ref{prp:lexwidth} that the result $Q'(D)$ for $D \in \calD$ also
has width depending only on~$Q'$ and on~$k$, and we know that restricting to a
subset of $Q'(D)$ (namely, each group) does not increase the width (this is like
the case of selection in the proof of Proposition~\ref{prp:lexwidth}). Hence, the
family $\calD'$ also has bounded width.
\item Theorem~\ref{thm:aggregnoproda} extends because we know (see
Lemma~\ref{lem:rewritenoprod} and subsequent observations) that the result
$Q'(D)$ for $D \in \calD$ is a union of a po-relation of bounded width and of a
po-relation with bounded ia-width. Restricting to a subset (i.e., a group), this
    property is preserved (as in the case of selection in the proof of
    Proposition~\ref{prp:lexwidth} and of Proposition~\ref{prp:lexiawidthnoprod}),
    which allows us to
    conclude.\qedhere
\end{itemize}
\end{toappendix}

\subparagraph*{Duplicate elimination.} We last study the problem of consolidating tuples with
\emph{duplicate values}. To this end, we define a new operator, $\dupelim$, and introduce a
semantics for it. The main problem is that tuples with the same
values may be ordered differently relative to other
tuples. To mitigate this, we introduce the notion
of \emph{id-sets}:

\begin{definition}
    \label{def:idset}
    Given a totally ordered po-relation $(\ID, T, <)$, a subset $\ID'$ of $\ID$
    is an \emph{indistinguishable duplicate set} (or
    \deft{id-set}) if for every $\id_1, \id_2 \in \ID'$,
    we have $T(\id_1) = T(\id_2)$,
    and for every $\id \in \ID \backslash \ID'$,
    we have $\id < \id_1$ iff $\id < \id_2$, and $\id_1 < \id$ iff $\id_2 < \id$.
\end{definition}

\begin{example}
    \label{exa:dup1}
    Consider the totally ordered relation
    $\OR_1 \defeq \Pi_{\mathit{hotelname}}(\mathit{Hotel})$, with
    $\mathit{Hotel}$ as in Figure~\ref{fig:examplerels}. The two
    ``Mercure''
    tuples are not an id-set: they disagree on their
    ordering with ``Balzac''.
    Consider now a totally ordered relation $\OR_2 = (\ID_2, T_2, <_2)$ whose only possible world is
    a list
    relation $(A, B, B,
    C)$ for some tuples $A$, $B$, and $C$ over $\calD$. The set
    $\{\id \in \ID_2 \mid T_2(\id) = B\}$ is an id-set in~$\OR_2$. Note that a singleton is always
    an id-set.
\end{example}
We define a semantics for $\dupelim$ on a totally ordered po-relation $\OR =
(\ID, T, <)$
via id-sets.
First, check that for every
tuple value $t$ in the image of~$T$, the set $\{\id \in \ID \mid T(\id) = t\}$ is an id-set
in~$\OR$. If this
holds, we call $\OR$ \emph{safe}, and set $\dupelim(\OR)$ to be the singleton
$\{L\}$ of the only possible world of
the restriction of~$\OR$ obtained by picking one representative
element per id-set (clearly $L$ does not depend on the chosen
representatives).
Otherwise, we call $\OR$ \emph{unsafe} and say that
duplicate consolidation has \emph{failed}; we then set $\dupelim(\OR)$ to
be an empty set of possible worlds. Intuitively, duplicate
consolidation tries to reconcile (or ``synchronize'') order
constraints for tuples with the same values, and fails when it
cannot be done.

\begin{example}
    In Example~\ref{exa:dup1}, we have $\dupelim(\OR_1)=\emptyset$
    but $\dupelim(\OR_2) = (A, B, C)$.
\end{example}
We then extend 
$\dupelim$ to po-relations by
considering all possible
results of duplicate elimination on the possible worlds,
ignoring the unsafe possible worlds. If no possible worlds are safe, then
we \emph{completely fail}:

\begin{definition}
    For each list relation $L$, we let $\OR_L$ be a po-relation such that
    $\pw(\OR_L) = \{L\}$.
    Letting $\OR$ be a po-relation, we set $\dupelim(\OR) \defeq
    \bigcup_{L \in pw(\OR)}\dupelim(\OR_L)$. We say that $\dupelim(\OR)$
    \deft{completely fails} if $\dupelim(\OR) = \emptyset$, i.e.,
    $\dupelim(\OR_L) = \emptyset$ for every $L\in
    pw(\OR)$.
\end{definition}

\begin{example}
    Consider the totally ordered po-relation
    $\mathit{Rest}$ from
    Figure~\ref{fig:examplerels}, and a totally ordered po-relation
    $\mathit{Rest}_2$ whose only possible
    world is $(\langlem
    \textup{Tsukizi} \ranglem,$ $\langlem \textup{Gagnaire} \ranglem)$.
    Consider 
    $Q \defeq
    \dupelim(\Pi_{\mathit{restname}}(\mathit{Rest}) \cupgen \mathit{Rest}_2)$.
    Intuitively, $Q$ combines restaurant rankings,
    using duplicate consolidation to collapse two occurrences of the
    same 
    name to a single tuple.
    The only possible world of $Q$ is
    (\textup{Tsukizi}, \textup{Gagnaire}, \textup{TourArgent}), since
    duplicate elimination fails in the other possible worlds: 
    indeed, this is the only possible way to combine the rankings.
\end{example}
We next show that the result of $\dupelim$ can still be represented as a
po-relation,
up to complete failure (which may be efficiently identified).

\begin{toappendix}
  \subsection{Proof of Theorems~\ref{thm:duelim-por}
  and~\ref{thm:incomparable2} and Proposition~\ref{prp:dupunion}}
    We first define the notion of \emph{quotient} of a
    po-relation by \emph{value equality}:

    \begin{definition}
        For a po-relation $\OR=(\ID,T,{<})$, we define the \emph{value-equality
            quotient of~$\OR$} as
        the directed graph $\mathrm{G}_\OR=(\ID', E)$ where:
        \begin{itemize}
            \item $\ID'$ is the quotient of
            $\ID$ by the equivalence relation $\id_1\sim\id_2\Leftrightarrow
            T(id_1)=T(id_2)$;
            \item $E\defeq\{(\id'_1,\id'_2)\in{\ID'}^2\mid \id_1' \neq
            \id_2'\land \exists
            (\id_1,\id_2) \in \id_1'\times \id_2'\text{~s.t.~}id_1<id_2\}$.
        \end{itemize}
    \end{definition}

    We claim that cycles in the value-equality quotient of~$\OR$ precisely
    characterize complete failure of $\dupelim$.

    \begin{propositionrep}\label{prp:pocycle}
        For any po-relation $\OR$, $\dupelim(\OR)$ completely fails iff
        $\mathrm{G}_\OR$ has a cycle.
    \end{propositionrep}

    \begin{proof}
      Fix the input po-relation $\OR = (\ID, T, <)$.
        We first show that the existence of a cycle implies complete failure of
        $\dupelim$. Let $\id'_1,\dots,\id'_n,\id'_1$ be a simple cycle of
        $\mathrm{G}_\OR$.
        For all $1\leq i\leq n$, there exists $\id_{1i},\id_{2i}\in\id'_1$ such
        that $\id_{2i}<\id_{1(i+1)}$
        (with the convention $\id_{1(n+1)}=\id_{11}$)
        and the $T(\id_{2i})$ are pairwise distinct.

        Let $L$ be a possible world of $\OR$ and let us show that $\dupelim$
        fails on any po-relation $\OR_L$ that \emph{represents}~$L$, i.e.,
        $\OR_L = (\ID_L, T_L, {<_L})$ is totally ordered and $\pw(\OR_L) = \{L\}$.
        Assume by contradiction that
        for all $1\leq i\leq n$, $\id'_i$ forms an id-set of $\OR_L$. Let us show
        by induction on~$j$ that for all $1\leq j\leq n$,
        $\id_{21}\leq_{L}\id_{2j}$, where $\leq_L$ denotes the non-strict order
        defined from $<_L$ in the expected fashion. The base case is trivial. Assume this holds
        for $j$ and let us show it for $j+1$. Since $\id_{2j}<\id_{1(j+1)}$, we
        have $\id_{21}\leq\id_{2j}<_L\id_{1(j+1)}$. Now, if $\id_{2(j+1)}<_{L}\id_{21}$,
        then $\id_{2(j+1)}<_{L}\id_{21}<_L\id_{1(j+1)}$ with
        $T(\id_{2(j+1)})=T(\id_{1(j+1)})\neq T(\id_{21})$, so this contradicts the fact that
        $\id'_{j+1}$ is an id-set. Hence, as $L$ is a total order, we must have
        $\id_{21} \leq_L \id_{2(j+1)}$, which
        proves the induction case. Now the claim proved by induction implies that
        $\id_{21}\leq_{L}\id_{2n}$, and we had $\id_{2n}<\id_{11}$ in~$\OR$ and
        therefore $\id_{2n}<_L\id_{11}$, so this contradicts the fact that
        $\id'_1$ is an id-set. Thus, $\dupelim$ fails in $\OR_L$. We have thus shown that
        $\dupelim$ fails in every possible world of~$\OR$, so that it completely fails.

        \medskip

        Conversely, let us assume that $\mathrm{G}_\OR$ is acyclic. Consider a
        topological sort of $\mathrm{G}_\OR$ as $\id'_1,\dots,\id'_n$. For $1\leq
        j\leq n$, let $L_j$ be a linear extension of the poset
        $(\id'_j,\restr{<}{\id'_j})$. Let $L$ be the concatenation of $L_1,\dots L_n$.
        We claim $L$ is a linear extension of $\OR$ such that $\dupelim$ does not
        fail in~$\OR_L = (\ID_L, T_L, {<_L})$; this latter fact is clear by construction of $L$, so we must only show
        that $L$ obeys the comparability relations of~$\OR$. Now, let
        $\id_1<\id_2$ in $\OR$. Either for some $1\leq j\leq n$,
        $\id_1, \id_2\in\id'_j$ and then the tuple for $\id_1$ precedes the one
        for $\id_2$ in~$L_j$ by construction, so
        means $t_1<_L t_2$; or they are in different classes $\id'_{j_1}$ and
        $\id'_{j_2}$
        and this is
        reflected in $\mathrm{G}_\OR$, which means that $j_1<j_2$ and
        $\id_1<_L \id_2$. Hence, $L$ is a linear extension, which concludes the proof.
    \end{proof}
    We can now state and prove the result:
\end{toappendix}

\begin{theoremrep}\label{thm:duelim-por}
    For any po-relation $\OR$, we can test in PTIME if
    $\dupelim(\OR)$ completely fails; if it does not,
    we can compute in PTIME a po-relation $\OR'$ such that
    $pw(\OR')=\dupelim(\OR)$.
\end{theoremrep}

\begin{proof}
    We first observe that $\mathrm{G}_\OR$ can be constructed in PTIME,
    and that testing that $\mathrm{G}_\OR$ is acyclic is also done in
    PTIME. Thus, using Proposition~\ref{prp:pocycle}, we can determine
    in PTIME whether $\dupelim(\OR)$ fails.

    If it does not, we let $\mathrm{G}_\OR=(\ID',E)$ and construct the
    relation $\OR'$ that will stand for $\dupelim(\OR)$ as
    $(\ID',T',<')$ where  $T'(\id')$ is the unique $T'(\id)$ for
    $\id\in\id'$ and $<'$ is the transitive closure of $E$, which is
    antisymmetric because $\mathrm{G}_\OR$ is acyclic. Observe that
    the underlying bag relation of $\OR'$ has one identifier for each distinct
    tuple value in $\OR$, but has no duplicates.

    Now, it is easy to check that $\pw(\OR')=\dupelim(\OR)$. Indeed, any
    possible world $L$ of~$\OR'$ can be achieved in $\dupelim(\OR)$ by
    considering, as in the proof of Proposition~\ref{prp:pocycle}, some
    possible world of $\OR$ obtained following the topological sort of
    $\mathrm{G}_\OR$ defined by $L$. This implies that $\pw(\OR')
    \subseteq \dupelim(\OR)$.

    Conversely, for any possible world $L$ of $\OR$, $\dupelim(\OR_L)$ (for
    $\OR_L$ a po-relation that represents $L$) fails
    unless, for each tuple value, the occurrences of that tuple value in
    $\OR_L$ is an id-set. Now, in such an $L$, as the occurrences of each
    value are contiguous and the order relations reflected in
    $\mathrm{G}_\OR$ must be respected, $L$ is defined by a topological
    sort of $\mathrm{G}_\OR$ (and some topological sort of each id-set
    within each set of duplicates), so that $\dupelim(\OR_L)$ can also be
    obtained as the corresponding linear extension of $\OR'$. Hence, we
    have $\dupelim(\OR) \subseteq \pw(\OR')$, proving their equality and
    concluding the proof.
\end{proof}

We note that $\dupelim$ is not redundant with any of the other PosRA
operators, generalizing Theorem~\ref{thm:incomparable}:

\begin{theoremrep}\label{thm:incomparable2}
  No operator among those of \PosRA and
  $\dupelim$ can be expressed through a combination of the
others.
\end{theoremrep}

\begin{proof}
  This is shown in the proof of Theorem~\ref{thm:incomparable} in
  Appendix~\ref{apx:incomparableproof}.
\end{proof}

Last, we observe that $\dupelim$ can indeed be used to undo some of the effects
of bag semantics. For instance, we can show the following:

\begin{toappendix}
    We also use the value-equality quotient to show:
\end{toappendix}

\begin{propositionrep}\label{prp:dupunion}
    For any po-relation $\OR$, we have $\dupelim(\OR \cup \OR) = \dupelim(\OR)$:
    in particular, one completely fails iff the other does.
\end{propositionrep}

\begin{proof}
    Let $G_\OR$ be the value-equality quotient of $\OR$ and $G'_\OR$ be the
    value-equality quotient of $\OR \cup \OR$. It is easy to see that these two
    graphs are identical: any edge of~$G_\OR$ witnesses the existence of the same
    edge in~$G'_\OR$, and conversely any edge in $G'_\OR$ must correspond to a
    comparability relation between two tuples of one of the copies of~$\OR$ (and
    also in the other copy, because they are two copies of the same relation), so
    that it also witnesses the existence of the same edge in~$\OR$. Hence, one
    duplicate elimination operation completely fails iff the other does, because
    this is characterized by acyclicity of the value-equality quotient (see
    Proposition~\ref{prp:pocycle}). Further, by Theorem~\ref{thm:duelim-por}, as
    duplicate elimination is constructed from the value-equality quotient, we have
    indeed the equality that we claimed.
\end{proof}

We can also show that most of our previous tractability results still apply when
the duplicate elimination operator is added:

\begin{toappendix}
  \subsection{Possibility and Certainty Results}
    \label{apx:posscertdupes}
    
    We first clarify the semantics of query evaluation when complete failure
    occurs: given a query~$Q$ in \PosRA{} extended with $\dupelim$, and given a po-database $D$,
    if complete failure occurs at any
    occurrence of the $\dupelim$ operator when evaluating $Q(D)$, we set
    $\pw(Q(D)) \defeq \emptyset$,
    pursuant to our choice of defining query evaluation on po-relations as
    yielding all possible results on all possible worlds.
    If $Q$ is
    a \posRAagg{} query extended with $\dupelim$, we likewise say that its
    possible accumulation results are $\emptyset$.

    This implies that for any \PosRA query $Q$ extended with $\dupelim$,
    for any input po-database $D$, and
    for any candidate possible world $v$, the \poss and \cert problems for $Q$ are
    vacuously false on instance $(D, v)$
    if complete failure occurs at any stage when evaluating $Q(D)$. The same holds
    for \PosRAacc queries.

    \subsubsection{Proof of Theorem~\ref{thm:easypossdupelim}: Adapting the Results of
        Section~\ref{sec:posscert}--\ref{sec:fpt}}
\end{toappendix}

\begin{theoremrep}\label{thm:easypossdupelim}
  All \poss and \cert tractability results of
  Sections~\ref{sec:posscert}--\ref{sec:fpt}, 
  except Theorem~\ref{thm:aggregnoprod} and
  Theorem~\ref{thm:aggregnoproda},
  extend to \PosRA and \PosRAagg where we allow $\dupelim$ (but impose
  the same restrictions on query operators, accumulation, and input
  po-relations).
\end{theoremrep}

\begin{toappendix}
    All complexity upper bounds in Sections~\ref{sec:posscert}--\ref{sec:fpt} are proved by first
    evaluating the query result in PTIME using Proposition~\ref{prp:repsys}. So we
    can still evaluate the query in PTIME, using in addition
    Theorem~\ref{thm:duelim-por}. Either complete failure occurs at some point in
    the evaluation, and we can
    immediately solve \poss and \cert by our initial remark above, or no complete failure occurs and we obtain
    in PTIME a po-relation on which to solve \poss and \cert. Hence, in what
    follows, we can assume that no complete failure occurs at any stage.

    Now, except Theorems~\ref{thm:aggregnoprod} and
    Theorem~\ref{thm:aggregnoproda}, the only assumptions that are
    made on the po-relation obtained from query evaluation are proved using the
    following facts:

    \begin{itemize}
        \item For all theorems in Section~\ref{sec:posscert}, for
          Theorem~\ref{thm:certaintyptimec}, and for
          Proposition~\ref{prop:otherdef}, no assumptions are
        made, so the theorems continue to hold.
        \item For Theorem~\ref{thm:aggregw} and Theorem~\ref{thm:aggregwa}, that the property of having a constant
        width is preserved during \Plex query evaluation, using
        Proposition~\ref{prp:lexwidth}.
    \end{itemize}

    Hence, Theorem~\ref{thm:easypossdupelim} follows from the following
        width preservation result:

    \begin{proposition}
        \label{prp:wdupe}
        For any constant $k \in \NN$ and po-relation $\OR$ of width $\leq k$, if
        $\dupelim(\OR)$ does not completely fail then it has width $\leq k$.
    \end{proposition}

    \begin{proof}
        It suffices to show that to every antichain $A$ of $\dupelim(\OR)$ corresponds an
        antichain $A'$ of the same cardinality in $\OR$. Construct $A'$ by picking a
        member of each of the classes of~$A$. Assume by contradiction that $A'$ is
        not an antichain, hence, there are two tuples $t_1 < t_2$ in~$A'$, and
        consider the corresponding classes $\id_1$ and $\id_2$ in $A$. By our
        characterization of the possible worlds of $\dupelim(\OR)$
        in the proof of Theorem~\ref{thm:duelim-por}
        as obtained from the topological sorts
        of the value-equality quotient $\mathrm{G}_\OR$ of $\OR$, as $t_1 < t_2$ implies
        that $(\id_1, \id_2)$ is an edge of $\mathrm{G}_\OR$, we conclude that we have $\id_1
        < \id_2$ in $A$, contradicting the fact that it is an antichain.
    \end{proof}

    We conclude by illustrating that Theorem~\ref{thm:aggregnoprod} cannot be
    adapted as-is, because the preservation result that it uses does not adapt to
    the $\dupelim$ operator.
    \begin{example}
        \label{exa:noprodfails}
        Fix $n \in \NN$. Consider the totally ordered relation $R \defeq \ordern{n}$; it
        has width~$1$.
        Consider the po-relation $S = (\ID, T, <)$ that consists of $n$ pairwise incomparable
        identifiers $\id^\uparrow_1,
        \ldots, \id^\uparrow_n$ whose images by $T$ are respectively $1, \ldots, n$, and $n$ pairwise
        incomparable identifiers $\id_1^\downarrow, \ldots, \id_n^\downarrow$ with pairwise distinct fresh
        values, with the order relation $\id_i^\uparrow < \id_j^\downarrow$ for all $1 \leq i, j \leq n$;
        The po-relation $S$ has ia-width~$2$, with the partition
        $(\{\id^\uparrow_i \mid 1 \leq i \leq n\},
        \{\id^\downarrow_i \mid 1 \leq i \leq n\})$. Hence, $R \cupgen S$ would satisfy the
        hypotheses of Theorem~\ref{thm:aggregnoprod}. However, $R' \defeq \dupelim(R
        \cupgen S)$
        is the po-relation consisting of tuples $\id'_1, \ldots, \id'_n$ with values
        respectively $1, \ldots, n$, tuples $\id''_1, \ldots, \id''_n$ with the values of
        the $\id^\downarrow_i$, and the order relation $\id'_i < \id'_j$ iff $i
        < j$ and $\id'_i < \id''_j$
        for all $1 \leq i, j \leq n$.

        We now observe that, for every partition of $R'$ into two sets, there is a
        comparability relation going from one set to the other. Hence, $R'$ cannot
        be written as the union of two non-empty po-relations. Yet, $R'$ has width
        $n$, as witnessed by the $\id''_i$, and it has ia-width $n$, as witnessed by
        the $\id'_i$.

        This illustrates that, when performing duplicate consolidation on the union of a
        constant-width po-relation and of a constant-ia-width po-relation, we cannot
        hope that the result has constant width, or constant ia-width, or can be
        written as the union of two relations where each has one of these properties.
    \end{example}
\end{toappendix}

Furthermore, if in a set-semantics spirit we {\em require} that the
query output has no duplicates, \poss and \cert are always tractable (as this
avoids the technical difficulty of Example~\ref{exa:notposet}):

\begin{toappendix}
    \subsubsection{Proof of Theorem~\ref{thm:posscertnodupes}: \poss and \cert
    After Removing Duplicates}
\end{toappendix}

\begin{theoremrep}\label{thm:posscertnodupes}
    For any \PosRA query $Q$, \poss and \cert for $\dupelim(Q)$
    are in PTIME.
\end{theoremrep}

\begin{proof}
    Let $D$ be an input po-relation, and $L$ be the candidate possible world
    (a list relation).
    We compute the po-relation $\OR'$ such that $\pw(\OR')= Q(D)$ in PTIME using
    Proposition~\ref{prp:repsys}
    and the po-relation $\OR \defeq \dupelim(\OR')$ in PTIME using
    Theorem~\ref{thm:duelim-por}. If duplicate elimination fails, we vacuously
    reject for \poss and \cert, following the remark at the beginning of
    Appendix~\ref{apx:posscertdupes}.
    Otherwise, the result is a po-relation $\OR$, with the property that each tuple
    value is realized exactly once, by definition of $\dupelim$. Note that we can
    reject immediately if $L$ contains multiple occurrences of the same tuple, or
    does not have the same underlying set of tuples as~$\OR$; so we assume that $L$ has the
    same underlying set of tuples as~$\OR$ and no duplicate tuples.

    The \cert problem is in PTIME on $\OR$ by Theorem~\ref{thm:certptime}, so
    we need only study the case of \poss, namely, decide whether $L \in
    \pw(\OR)$.
    Let $\OR_L$ be a po-relation that represents~$L$.
    As $\OR_L$ and $\OR$ have no duplicate tuples, there is only one way to
    match each identifier of $\OR_L$ to an identifier of $\OR$.
    Build $\OR''$ from $\OR$ by adding, for each pair $\id_i <_L \id_{i+1}$ of consecutive
    tuples of~$\OR_L$, the order constraint $\id_i'' {<''} \id''_{i+1}$ on the
    corresponding identifiers in $\OR''$. We claim that $L \in \pw(\OR)$ iff the
    resulting $\OR''$
    is a po-relation, i.e., its transitive closure is still
    antisymmetric, which can be tested in PTIME by computing the strongly connected
    components of $\OR''$ and checking that they are all trivial.

    To see why this
    works, observe that, if the result $\OR''$ is a po-relation, it is a total
    order, and so it describes a way to achieve $L$ as a linear extension of
    $\OR$ because it doesn't contradict any of the comparability relations
    of $\OR$.
    Conversely, if $L \in \pw(\OR)$, assuming to the contrary the existence of a
    cycle in $\OR''$, we observe that such a cycle must consist of order relations
    of $\OR$ and $\OR_L$, and the order relations of $\OR$ are reflected in
    $\OR_L$ as it is a
    linear extension of $\OR$, so we deduce the existence of a cycle in $\OR_L$,
    which is impossible by construction. Hence, we have reached a contradiction,
    and we deduce the desired result.
\end{proof}

\begin{toappendix}
    \subsection{Alternative Semantics for Duplicate Elimination}

A main downside of our
proposed semantics for $\dupelim$ is the fact that complete failure is
allowed. We conclude 
by briefly considering alternative
semantics that avoid failure, and illustrate the other
problems that they have.

A first possibility is to do a \emph{weak} form of duplicate
elimination: keep one 
element for each \emph{maximal
id-set}, rather than for each value, and leave some duplicates in
the output:

\begin{example}
  \label{exa:dupelim2}
  Letting $A \neq B$ be two tuples,
  let us consider a po-relation $\OR_L$ representing the list relation $L \defeq (A, B, B, A)$. With
  weak duplicate elimination, 
  we would have $\dupelim(\OR_L) =
  (A, B, A)$.
\end{example}

However, when generalizing this semantics from totally ordered relations to
po-relations, we notice that the result of $\dupelim$ on a
po-relation may not be representable as a po-relation, since possible worlds differ in their tuples and not only on their order:

\begin{example}
  Consider the po-relation $\OR = (\{a_1, b, a_2\}, T, {<})$ with $T(a_1) = T(a_2) = A$
  and $T(b) = B$, where $A \neq B$ are tuples, and $<$ defined by $a_1
  < b$ and $a_1 < a_2$.
  We have $\pw(\OR) = \{(A, B, A), (A, A, B)\}$ and 
  $\dupelim(\OR) = \{(A, B, A), (A, B)\}$ for \emph{weak} duplicate
  elimination:
  we cannot represent it as a po-relation (the underlying relation is
  not certain).
\end{example}

A second possibility is to do an \emph{aggressive} form of duplicate
elimination: define $\dupelim(L)$ for a list relation~$L$ as the set
of \emph{all} totally ordered relations that we can obtain by
picking one representative element for each value, even
when the representatives are not indistinguishable. In other
words, we do not fail even if we cannot reconcile the order
between duplicate tuples:

\begin{example}
 Applying \emph{aggressive} $\dupelim$ to $\OR_L$ from Example~\ref{exa:dupelim2} yields
  $\{(A, B), (B, A)\}$.
\end{example}

However, again
$\dupelim(\OR_L)$ may not be
representable as a po-relation, this time because the set of possible orders may not correspond to a partial order:

\begin{example}
  Consider  a po-relation $\OR_L$ representing the list relation
  $L \defeq (A, C, B, C, A)$
  with distinct tuples $A$, $B$, $C$. Then
  $\dupelim(\OR_L)$ is $\{(A, C, B),
  (A, B, C), (B, C, A), (C, B, A)\}$. No po-relation $\OR$ satisfies
  $\pw(\OR) = \dupelim(L)$, because no comparability pair holds in all
  possible worlds, so $\OR$ must be unordered, but then all
  permutations of $\{A, B, C\}$ are possible worlds of~$\OR$, which is
  unsuitable because
  some of the six permutations of $\{A, B, C\}$ are not possible worlds.
\end{example}

We leave for future work the question of designing a practical semantics
for duplicate consolidation that can be incorporated in our framework
and avoids failure.

\end{toappendix}

\subparagraph*{Discussion.}  The introduced group-by and duplicate elimination operators have some  shortcomings: the result of group-by is in general not representable by po-relations, and duplicate elimination may fail. These are both consequences of our design choices, where we capture only uncertainty on order (but not on tuple values) and design each operator so that its result corresponds to the result of applying it to each individual world of the input (see further discussion in Section \ref{sec:compare}). Avoiding these shortcomings is left for future work.

\section{Comparison With Other Formalisms}\label{sec:compare}
\begin{toappendix}
  \label{apx:relatedwork}
\end{toappendix}

We next compare our formalism to previously proposed formalisms:
query languages over bags (with no order); a query language for
partially ordered multisets; and other related work. 
To our knowledge, however, none of these works studied the possibility or
certainty problems for partially ordered data, so that our technical results do not follow from them.

\subparagraph*{Standard bag semantics.}
We first compare to related work on the \emph{bag semantics} for relational
algebra. Indeed, a natural desideratum for our semantics on
(partially) ordered relations is that it should be a faithful extension
of bag semantics. We first consider the $\BALG^{1}$ language on bags
\cite{GrumbachMiloBags} (the ``flat fragment'' of their language $\BALG$ on
nested relations).
We denote by $\BALG_{+}^{1}$ the fragment of~$\BALG^{1}$,
that includes the standard extension of positive relational algebra
operations to bags: additive union, cross product, selection,
and projection.
We observe that, indeed, our semantics
faithfully extends $\BALG_{+}^{1}$:
\emph{query
evaluation commutes with ``forgetting'' the order}. Formally, for a po-relation $\OR$, we
denote by $\bag(\OR)$ its underlying bag relation, and define likewise 
$\bag(D)$ for a po-database~$D$ as the database of the underlying bag relations.
For the
following comparison, we identify $\times_{\dir}$ and
$\times_{\lex}$ with the $\times$ of~\cite{GrumbachMiloBags} and our
union with the additive union of~\cite{GrumbachMiloBags}, and then the following trivially holds:
\begin{propositionrep}\label{prop:faithful}
    For any \PosRA query $Q$ and a po-relation $D$, $\bag(Q(D))=Q(\bag(D))$ where
$Q(D)$ is defined according to our semantics and $Q(\bag(D))$ is defined by
$\BALG_{+}^{1}$.
\end{propositionrep}

\begin{proof}
    There is an exact correspondence in terms of the output bags between
    additive union and our union; between cross product and $\times_{\dir}$ and
    $\times_{\lex}$ (both our product operations yield the same bag as output, for
    any input); between our selection and that of $\BALG_{+}^{1}$, and similarly
    for projection (as noted before the statement of
    Proposition~\ref{prop:faithful} in the main text, a technical subtlety is that the projection
    of $BALG$ can only project on a single attribute, but one can encode
    ``standard'' projection on multiple attributes).
     The proposition follows by induction on the query structure.
\end{proof}

The full $\BALG^{1}$ language includes additional operators, such as bag
intersection and subtraction, which are non-monotone and as such may not be
expressed in our language: it is also unclear how they could be extended to our
setting (see further discussion in ``\textsf{Algebra on pomsets}'' below).
On the other hand, $\BALG^{1}$ does not include
aggregation, and so \posRAagg and $\BALG^{1}$ are incomparable in terms of
expressive power.

A better yardstick to compare against for accumulation could be~\cite{libkin1997query}:
they show that their basic language $\BQL$ is equivalent
to $\BALG$, and then further extend the language with aggregate operators, to
define a language called $\NRLaggr$ on nested relations.
On flat relations,
$\NRLaggr$
captures
functions that cannot be captured in our
language: in particular the average function \textsf{AVG} is non-associative and thus cannot be captured by
our accumulation function (which anyway focuses on order-dependent functions, as
\poss/\cert are trivial otherwise).
On the other hand, $\NRLaggr$ cannot test 
parity 
(Corollary~5.7 in \cite{libkin1997query}) whereas this is easily captured by our accumulation operator.
We conclude that $\NRLaggr$ and \posRAagg are incomparable in terms of captured transformations on bags, even when restricted to flat relations.

\subparagraph*{Algebra on pomsets.} We now compare our work to algebras defined
on \emph{pomsets}~\cite{grumbach1995algebra,grumbach1999algebra}, which also 
attempt to bridge partial order theory and data management (although, again, they do not
study possibility and certainty). 
\emph{Pomsets} are labeled posets quotiented by
isomorphism (i.e., renaming of identifiers), like po-relations.
A
major
conceptual difference between our formalism and that
of~\cite{grumbach1995algebra,grumbach1999algebra} is that their language focuses on
processing {\em connected components} of the partial order graph,
and their operators are tailored for that semantics. As a
consequence, their semantics is {\em not} a faithful extension of bag
semantics, i.e., their language would not satisfy the counterpart of Proposition \ref{prop:faithful} (see for instance the semantics of union in \cite{grumbach1995algebra}). 
By contrast, we manipulate po-relations that stand for sets of possible list
relations, and our operators are designed accordingly, unlike those of
\cite{grumbach1995algebra} where transformations take into account the structure
(connected components) of the entire poset graph.
Because of this choice, \cite{grumbach1995algebra} introduces
non-monotone operators that we cannot express, and 
can design a duplicate elimination operator that cannot fail. Indeed, the possible failure of our duplicate elimination operator is a direct consequence of its semantics of operating on each possible world, possibly leading to contradictions.

If we consequently disallow duplicate elimination in both languages for the sake
of comparison, we note that the resulting fragment $\PomAlgEps$ of the language
of \cite{grumbach1995algebra} can yield only series-parallel output
(Proposition~4.1 of~\cite{grumbach1995algebra}), unlike \PosRA
queries whose output order may be arbitrary (see
Appendix~\ref{apx:relatedwork}). Hence, $\PomAlgEps$ does not subsume
\PosRA.

\begin{toappendix}
We formally prove that the output of a \PosRA query can be arbitrary:

\begin{propositionrep}\label{prp:gen}
  For any po-relation $\OR$, there is a \PosRA query $Q$ with no inputs
  s.t. $Q() = \OR$.
\end{propositionrep}

  To prove the result, we will need the notion of a \emph{realizer} of a
  poset:
  
  \begin{definition} \label{def:realizer}
  \cite{schroder2003ordered}
  Letting $P = (V, <)$ be a poset, we say that a set of total orders
  $(V, {<_1}), \ldots, (V, {<_n})$ is a \deft{realizer} of $P$ if for every
  $x, y \in V$, we have $x < y$ iff $x <_i y$ for all~$i$.
\end{definition}

  We will use this notion for the following lemma. This lemma is given as Theorem~9.6 of~\cite{hiraguchi1955dimension},
see also~\cite{ore1962theory}; we rephrase it in our vocabulary, and for
  convenience we also give a self-contained proof.

  \begin{lemma}
    \label{lem:dimprodoneway}
    Let $n \in \NN$, and let
    $(P, <_P)$ be a poset that has a realizer $(L_1, \ldots, L_n)$ of size~$n$.
    Then $P$ is isomorphic to a subset $\OR'$ of $\OR = \ordern{l} \times_{\gen}
    \cdots \times_{\gen} \ordern{l}$, with $n$ factors in the product, for some integer $l \in \mathbb{N}$ (the
    order on $\OR'$ being the restriction on that of $\OR$).
  \end{lemma}

  \begin{proof}
    We define $\OR$ by taking $l \defeq \card{P}$,
    and we
    identify each
    element $x$ of $P$ to $f(x) \defeq (n_1^x, \ldots, n_n^x)$, where
$n^x_i$ is
    the position where $x$ occurs in $L_i$. Now, for any $x, y \in P$, we have $x <_P y$ iff $n_i^x < n_i^y$ for
    all $1 \leq i \leq n$ (that is, $x <_{L_i} y$), hence iff
    $f(x) <_\OR f(y)$: this uses the fact that there are no two elements $x \neq y$
    and $1 \leq i \leq n$
    such that the $i$-th components of $f(x)$ and of $f(y)$ are the same.
    Hence, taking $\OR'$ to be the image of $f$ (which is
    injective), $\OR'$ is indeed isomorphic to $P$.
  \end{proof}

  We are now ready to prove Proposition~\ref{prp:gen}:
\begin{proof}[Proof of Proposition~\ref{prp:gen}]
  We first show that for any
  poset $(P, <)$, there exists a \Pgen query $Q$
  such that the tuples of $\OR' \defeq Q()$ all have unique
  values and the underlying poset of $\OR'$ is $(P, <)$.
  Indeed, we can take $d$ to be the \emph{order dimension} of
  $P$, which is necessarily finite~\cite{schroder2003ordered}, and then by
  definition $P$ has a realizer of size~$d$. By
  Lemma~\ref{lem:dimprodoneway}, there is an integer
  $l \in \mathbb{N}$ such that $\OR'' \defeq \ordern{l} \times_\gen \cdots
  \times_\gen \ordern{l}$ (with $n$ factors in the product)
  has a subset $S$ isomorphic to $(P, <)$. Hence, letting $\psi$ be a tuple predicate
  such that $\sigma_\psi(\OR'') = S$ (which can clearly be constructed by
  enumerating the elements of $S$), the query $Q' \defeq \sigma_\psi(\OR'')$ proves the
  claim, with $\OR''$
  expressed as above.

  Now, to prove the desired result from this claim, build $Q$ from $Q'$ by
  taking its join (i.e., $\times_\lex$-product, selection,
projection) with a union of singleton
  constant expressions that map each unique tuple value of $Q'()$ to the desired
  value of the corresponding tuple in the desired po-relation~$\OR$. This
  concludes the proof.
\end{proof}
\end{toappendix}

\label{sec:related}
\subparagraph*{Incompleteness in databases.}
Our work is inspired by the field of incomplete information management, which
has been studied for
various models~\cite{Incompletexml,Libkin06}, in particular
relational databases~\cite{IL84}. This field inspires our design
of po-relations
and our study of
possibility and certainty~\cite{KO06,lipski1979semantic}. However, uncertainty in
these settings typically focuses on \emph{whether} tuples exist
or on what their \emph{values} are (e.g., with nulls~\cite{codd1979extending},
including the novel approach of~\cite{libkin2014incompleteness,libkin2015icdt};
with c-tables~\cite{IL84}, probabilistic
databases~\cite{probdbbook} or fuzzy numerical values as
in~\cite{soliman2009ranking}). To our knowledge, though, our work is the first to study possible and
certain answers in the context of 
{\em
order}-incomplete data. Combining order incompleteness with standard
tuple-level uncertainty is left as a challenge for future work.
Note that some works~\cite{BunemanJO91,libkin1998semantics,libkin2015icdt} use
partial orders on \emph{relations} to compare the
informativeness of representations. This is
unrelated to our partial orders on \emph{tuples}.

\subparagraph*{Ordered domains.} Another line of work has studied
relational data management where the \emph{domain elements} are
(partially) ordered \cite{Immermanptime,Ng,van1997complexity}. However, the perspective
is different: we see order on tuples as part of the relations, and
as being constructed by applying our operators; these works see
order as being given \emph{outside} of the query, hence do not study the propagation of uncertainty through queries.
Also, queries in such works can often directly access the order
relation ~\cite{van1997complexity,BenS2009}. Some works also study uncertainty on totally ordered
\emph{numerical}
domains~\cite{soliman2009ranking,soliman2010supporting}, while
we look at general order relations.

\subparagraph*{Temporal databases.}
\emph{Temporal
databases}~\cite{chomicki2005time,snodgrass2000developing} consider
order on facts, but it is usually induced by timestamps, hence
total. A notable exception is~\cite{fan2012determining} which
considers that some facts may be \emph{more current} than others,
with constraints leading to a partial order. In particular, they
study the complexity of retrieving query answers that are certainly
current, for a rich query class. In contrast, we can
{\em manipulate} the order via queries, and we can also ask
about aspects beyond currency, as shown throughout the paper
(e.g., via accumulation).

\subparagraph*{Using preference information.} Order theory has been also used to handle \emph{preference information} in
database systems~\cite{jacob2014system,arvanitis2014preferences,
kiessling2002foundations,alexe2014preference, stefanidis2011survey},
with some operators being the same as ours, and for \emph{rank
aggregation}~\cite{Fagin,jacob2014system,dwork2001rank}, i.e. retrieving top-$k$ query answers given multiple
rankings. However, such works typically try to \emph{resolve}
uncertainty by reconciling many conflicting representations (e.g.
via knowledge on the individual scores given by different sources
and a function to aggregate them \cite{Fagin}, or a preference
function~\cite{alexe2014preference}). In contrast, we focus on maintaining a faithful model of
\emph{all} possible worlds without reconciling them, studying possible and certain answers in this respect.

\section{Conclusion}
\nosectionappendix
\label{sec:conclusion}
This paper introduced an algebra for order-incomplete data. We have studied the complexity of possible and certain
answers for this algebra, have shown the problems to be generally intractable, and identified
several tractable cases. In future work we plan to study the incorporation of additional operators (in particular non-monotone ones), investigate how to combine order-uncertainty with uncertainty on values, and study additional semantics for $\dupelim$. Last, it would be interesting to establish a dichotomy
result for the complexity of \poss, and a
complete syntactic characterization of cases where \poss is tractable.

\subparagraph*{Acknowledgements.} We are grateful to Marzio De
Biasi, P\'alv\"olgyi D\"om\"ot\"or, and Mikhail Rudoy, from
\url{cstheory.stackexchange.com}, for helpful suggestions.
This research was partially supported by the Israeli Science Foundation (grant 1636/13) and the Blavatnik~ICRC.

\begin{toappendix}
  \section{Discussion of Changes in this Version}
\label{apx:errata}

In the process of preparing a journal version of this
paper~\cite{amarilli2018computing}, we have discovered a flaw in the proof of
some of our tractability results on ia-width. We have accordingly
removed these results from~\cite{amarilli2018computing}
and from the present version of this paper.
However, the results still survive in the first version of this paper on
arXiv~\cite{amarilli2017possiblev1} and in the published
version in the TIME proceedings~\cite{amarilli2017possible}. In this appendix,
we list the affected theorems, point out the source of the error, and discuss our
current understanding of their correctness.

\subparagraph*{Affected theorems.}
The affected theorems are numbered as follows in the TIME proceedings version~\cite{amarilli2017possible} and in
the main text of~\cite{amarilli2017possiblev1}:

\begin{itemize}
  \item Theorem 19: tractability of \poss for any \PosRA query on po-databases
  of unordered po-relations.
  \item Theorem 22: tractability of \poss for any \PosRA query on po-databases
  of bounded-ia-width po-relations.
  \item Theorem 30: tractability of \poss and \cert for any \PosRAacc query on 
  po-databases of bounded-ia-width po-relations.
\end{itemize}

\subparagraph*{Source of the error.}
The error is in Proposition~66 of~\cite{amarilli2017possiblev1}.
This proposition claims
that, for any \PosRA query $Q$ and $k \in \NN$, there is a bound $k' \in \NN$
such that, for any po-database $D$ of po-relations of ia-width $\leq k$, the
po-relation $Q(D)$ has ia-width $\leq k'$. The proof is by induction, but in the
case of the product operators $\times_\lex$ and $\times_\gen$, the argument
does not correctly reflect the behavior of the product operators. For this
reason, the proof of the proposition is incorrect, and this affects the theorems
listed previously, because their proofs rely on Proposition~66.

\subparagraph*{Status of the results.}
It is easy to see that the \emph{statement} of Proposition~66 fails to hold:

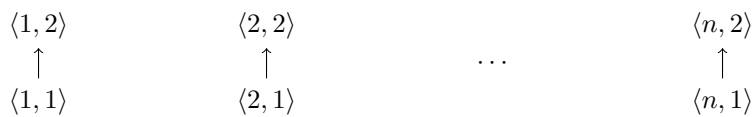
\begin{figure}
  \centering
  \begin{tikzpicture}[xscale=3,->]
    \node (a11) at (0, 0) {$\langle 1, 1 \rangle$};
    \node (a12) at (0, 1) {$\langle 1, 2 \rangle$};
    \node (a21) at (1, 0) {$\langle 2, 1 \rangle$};
    \node (a22) at (1, 1) {$\langle 2, 2 \rangle$};
    \node (d) at (2, .5) {$\cdots$}; , 
    \node (an1) at (3, 0) {$\langle n, 1 \rangle$};
    \node (an2) at (3, 1) {$\langle n, 2 \rangle$};
    \draw (a11) -- (a12);
    \draw (a21) -- (a22);
    \draw (an1) -- (an2);
  \end{tikzpicture}
  \caption{Illustration of the Hasse diagram of $Q_\lex(D_n)$ in Example~\ref{exa:noiawidth}}
  \label{fig:noiawidth}
\end{figure}

\begin{example}
  \label{exa:noiawidth}
  For any $n \in \NN$, consider the po-relation $\OR_n = (\ID_n, T_n,
  <_n)$ with $\ID_n = \{1, \ldots, n\}$, with $T_n$ being the identity function,
  and with $<_n$ being empty. As $\OR_n$ is unordered, it has ia-width~$1$.
  Consider the \PosRA query $Q_\lex \colonequals R \times_\lex \ordern{2}$.
  Call $D_n$ the po-database
  interpreting relation name~$R$ by~$\OR_n$, and let
  $\OR_{n,\lex} \colonequals Q_\lex(D_n)$.
  The set of identifiers of~$\OR_{n,\lex}$ is $\{(i, j) \mid 1 \leq i
  \leq n, 1 \leq j \leq 2\}$ (where we use $1$ and $2$ as the identifiers
  of the tuples in~$\ordern{2}$), and the order relation~$<$ is defined as
  follows (see Figure~\ref{fig:noiawidth} for an illustration):

  \begin{itemize}
  \item For all $1 \leq i \leq n$, we have $(i, 1) < (i, 2)$;
  \item For all $i \neq j$ in $\{1, \ldots, n\}$, for all $p, q \in \{1, 2\}$,
  the tuples $(i, p)$ and $(j, q)$ are incomparable.
  \end{itemize}

  We now show that the ia-width of $\OR_{n,\lex}$
  is equal to~$2n$, by arguing that there is no indistinguishable antichain
  containing two different identifiers. Indeed, consider any two identifiers
  $(i,p) \neq (j,q)$, assume 
  that there is an indistinguishable antichain $A$ that contains both of them,
  and let us show a contradiction.
  If $i = j$, then the identifiers are comparable, so they cannot both
  occur in~$A$, contradicting our assumption. Otherwise, letting $p' \colonequals 3-p$, we know that $(i,p)$ and $(i,p')$
  are comparable, so $(i,p')$ cannot be in~$A$. We now see that $(i,p')$
  violates indistinguishability for~$A$: we know that it is comparable
  to~$(i,p)$, but it is not comparable to~$(j,q)$ because $i \neq j$. Hence, we
  have a contradiction, and $(i,p)$ and $(j,q)$ cannot both occur in~$A$. So
  indeed the ia-width of $\OR_{n,\lex}$ is equal to~$2n$.

  Hence, we have an example of a \PosRA query using only the $\times_\lex$
  product for which the query
  result on a po-database of ia-width~$1$ can have unbounded ia-width. This
  contradicts the statement of Proposition~66.
  
  We note that we can also
  use the $\times_\gen$ product instead of~$\times_\lex$, e.g., with the query 
  $Q_\gen \colonequals R \times_\gen \ordern{2}$ and with the same
  construction: it is easy to see that $\OR_{n,\gen} \colonequals
  Q_\gen(D_n)$ is exactly equal to~$\OR_{n,\lex}$. Hence, Proposition~66
  fails even when restricted to \Plex or to \Pgen, which concludes the example.
\end{example}

We can also show the following result, which contradicts Theorem~19 and
Theorem~22 under the assumption that P is different from NP:

\begin{theorem}
  \label{thm:intractable}
  There is a \PosRA query $Q$ for which the \poss problem is NP-complete even
  when the input po-database is restricted to consist only of unordered
  po-relations.
\end{theorem}

As for Theorem~30, we do
not know whether a corresponding intractability result can be shown, i.e.,
whether we can adapt Theorem~\ref{thm:intractable} to perform accumulation in a
\emph{finite} monoid rather than in the free monoid. We also note that the query
used to prove Theorem~\ref{thm:intractable} will use both~$\times_\dir$ and $\times_\lex$,
so we do not know whether a restriction of Theorem~19 or Theorem~22 to \Plex or
\Pgen could hold.

We will show Theorem~\ref{thm:intractable} in the rest of this appendix. 
Let $\a \neq \b$ be two distinguished domain values of~$\calD$.
We will reduce from an NP-hard problem on so-called \emph{$\a\b$-bipartite}
po-relations:

\begin{definition}
  Let $\OR = (\ID, T, <)$ be a po-relation. We say that $\OR$ is \emph{bipartite}
  if we can partition $\ID = U \sqcup V$ such that, for any pair $\id < \id'$ of
  comparable identifiers, we have $\id \in U$ and $\id \in V$. (Equivalently,
  the Hasse diagram of the poset $(\ID, <)$ is a directed bipartite graph.) We
  say that $\OR$ is \emph{$\a\b$-bipartite} if the partition can be chosen as
  $U \colonequals \{\id \in \ID \mid T(\id) = \a\}$ and 
  $V \colonequals \{\id \in \ID \mid T(\id) = \b\}$. Note that, in this case,
  the domain of~$\OR$ is necessarily $\{\a, \b\}$, and the partition can be
  computed in PTIME simply by looking at the element labels.
\end{definition}

We show hardness of \poss on $\a\b$-bipartite po-relations for a specific kind
of possible worlds:

\begin{proposition}
  \label{prp:bipartiteposs}
  The following problem is NP-hard: given an $\a\b$-bipartite po-relation $\OR$
  with partition $U \sqcup V$, and two integers $0 \leq p \leq \card{U}$ and $0
  \leq q \leq \card{V}$,
  decide whether the totally ordered relation $L_{p,q} = \a^p \b^q \a^{\card{U}
  - p} \b^{\card{V} - q}$ on~$\{\a, \b\}$ is a possible world of~$\OR$.
\end{proposition}

\begin{proof}
  We reduce from the NP-hard \emph{$k$-clique problem}:
  given an undirected graph $G =
  (X, E)$
  and an integer $k \in \NN$, decide whether $G$ contains a clique of $k$
  vertices. Given the undirected graph $G$ and the integer~$k$, we construct the
  po-relation $\OR$ by creating one $\a$-labeled identifier in~$U$ for each
  vertex of~$X$ (that we identify to the vertex), creating one $\b$-labeled identifier in~$V$ for each edge
  of~$X$ (that we identify to the edge), and defining the order as follows: for
  any edge $e = \{x, y\}$ of~$E$, we set $x < e$ and $y < e$.
  It is immediate that $\OR$ is indeed $\a\b$-bipartite. We set $p \colonequals
  k$ and set $q \colonequals {k \choose 2}$. The construction is clearly in
  PTIME.

  Now, to show correctness, if $G$ contains a $k$-clique $X' \subseteq X$, we
  achieve the totally ordered relation $L_{p,q}$ by first enumerating all the
  $p$ identifiers of~$X'$ (they are $\a$-labeled so they are incomparable and have
  no ancestors), then enumerating the $q$ edges of the clique between the
  vertices of~$X'$ (they are $\b$-labeled, so incomparable, and their ancestors
  are all in~$X'$ so they have already been enumerated), then enumerating all
  remaining vertices (they are $\a$-labeled, so incomparable and have no
  ancestors) and edges (they are $\b$-labeled, so incomparable, and their
  ancestors have already been enumerated).

  Conversely, assume that there is a topological sort of~$\OR$ that achieves
  $L_{p,q}$. We define $X'$ to contain the vertices that were enumerated to
  achieve the prefix $\a^p$. We know that, afterwards, we have enumerated $q$
  identifiers that were $\b$-labeled, and the corresponding edges must have been
  between vertices of~$X'$, otherwise the order constraints prevent us from
  enumerating them. So the induced subgraph of~$G$ on~$X'$ contains $q$ edges,
  i.e., it is a clique. This concludes the correctness proof and establishes
  NP-hardness of our problem.
\end{proof}

We are now ready to prove Theorem~\ref{thm:intractable}:

\begin{figure}
  \centering
  \begin{tikzpicture}[xscale=1,->]
    \node (title) at (0, 2) {$\OR_R' \colonequals \sigma_R\big(\ordern{2}
    \times_\lex R\big)$};
    \node (a20) at (0, 1) {$\langle 2, 0 \rangle$};
    \node (a11) at (-2, 0) {$\langle 1, 1 \rangle$};
    \node (a12) at (-1, 0) {$\langle 1, 2 \rangle$};
    \node (a13) at (0, 0) {$\langle 1, 3 \rangle$};
    \node (d) at (1, 0) {$\cdots$}; 
    \node (a1u) at (2, 0) {$\langle 1, n \rangle$};
    \draw (a11.north) -- (a20);
    \draw (a12.north) -- (a20);
    \draw (a13.north) -- (a20);
    \draw (a1u.north) -- (a20);
  \end{tikzpicture}
  \begin{tikzpicture}[xscale=1,->]
    \node (title) at (0, 2) {$\OR_S' \colonequals \sigma_S\big(\ordern{2}
    \times_\lex S\big)$};
    \node (a10) at (0, 0) {$\langle 1, 0 \rangle$};
    \node (a21) at (-2, 1) {$\langle 2, 1 \rangle$};
    \node (a22) at (-1, 1) {$\langle 2, 2 \rangle$};
    \node (a23) at (0, 1) {$\langle 2, 3 \rangle$};
    \node (d) at (1, 1) {$\cdots$}; 
    \node (a2u) at (2, 1) {$\langle 2, m \rangle$};
    \draw (a10) -- (a21.south);
    \draw (a10) -- (a22.south);
    \draw (a10) -- (a23.south);
    \draw (a10) -- (a2u.south);
  \end{tikzpicture}
  \caption{Illustration of the Hasse diagram of~$\OR_R'$ and $\OR_S'$ in the proof of Theorem~\ref{thm:intractable}}
  \label{fig:prf1}
\end{figure}
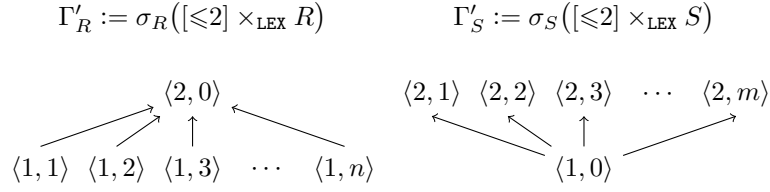

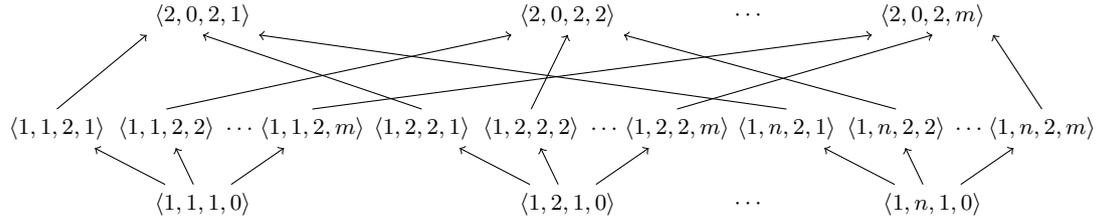
\begin{figure}
  \centering
  \begin{tikzpicture}[xscale=.48,->]
    \footnotesize
    \node (a1110) at (0, 0) {$\langle 1, 1, 1, 0 \rangle$};
    \node (a1210) at (10, 0) {$\langle 1, 2, 1, 0 \rangle$};
    \node (a1d) at (15, 0) {$\cdots$};
    \node (a1u10) at (20, 0) {$\langle 1, n, 1, 0 \rangle$};
    \node (a1121) at (-4, 1) {$\langle 1, 1, 2, 1 \rangle$};
    \node (a1122) at (-1, 1) {$\langle 1, 1, 2, 2 \rangle$};
    \node (a112d) at (1, 1) {$\ldots$};
    \node (a112v) at (3, 1) {$\langle 1, 1, 2, m \rangle$};
    \draw (a1110) -- (a1121);
    \draw (a1110) -- (a1122);
    \draw (a1110) -- (a112v);
    \node (a1221) at (6, 1) {$\langle 1, 2, 2, 1 \rangle$};
    \node (a1222) at (9, 1) {$\langle 1, 2, 2, 2 \rangle$};
    \node (a122d) at (11, 1) {$\ldots$};
    \node (a122v) at (13, 1) {$\langle 1, 2, 2, m \rangle$};
    \draw (a1210) -- (a1221);
    \draw (a1210) -- (a1222);
    \draw (a1210) -- (a122v);
    \node (a1u21) at (16, 1) {$\langle 1, n, 2, 1 \rangle$};
    \node (a1u22) at (19, 1) {$\langle 1, n, 2, 2 \rangle$};
    \node (a1u2d) at (21, 1) {$\ldots$};
    \node (a1u2v) at (23, 1) {$\langle 1, n, 2, m \rangle$};
    \draw (a1u10) -- (a1u21);
    \draw (a1u10) -- (a1u22);
    \draw (a1u10) -- (a1u2v);
    \node (a2021) at (0, 2.5) {$\langle 2, 0, 2, 1 \rangle$};
    \node (a2022) at (10, 2.5) {$\langle 2, 0, 2, 2 \rangle$};
    \node (a2d) at (15, 2.5) {$\cdots$};
    \node (a202v) at (20, 2.5) {$\langle 2, 0, 2, m \rangle$};
    \draw (a1121.north) -- (a2021.south west);
    \draw (a1221.north) -- (a2021.south);
    \draw (a1u21.north) -- (a2021.south east);
    \draw (a1122.north) -- (a2022.south west);
    \draw (a1222.north) -- (a2022.south);
    \draw (a1u22.north) -- (a2022.south east);
    \draw (a112v.north) -- (a202v.south west);
    \draw (a122v.north) -- (a202v.south);
    \draw (a1u2v.north) -- (a202v.south east);
  \end{tikzpicture}
  \caption{Illustration of the Hasse diagram of $\OR'$ (omitting tuple $\langle 2,0,1,0 \rangle$) in the proof of Theorem~\ref{thm:intractable}}
  \label{fig:prf2}
\end{figure}

\begin{proof}
  We will reduce from the NP-hard problem of
  Proposition~\ref{prp:bipartiteposs}. We start by describing formally the
  construction used in the reduction, i.e., the fixed query and input
  unordered po-relations, but the reader may find it more informative to
  digest the query bottom-up
  by reading the explanation of query evaluation given at the beginning of the
  correctness proof.

  The fixed query is as follows:
  \[Q \colonequals W \cup
  \Pi\Bigg(\sigma_=\bigg(\sigma_R\Big(\ordern{2} \times_\lex
  R\Big) \quad\times_\gen\quad \sigma_S\Big(\ordern{2}
  \times_\lex S\Big) \quad \times_\lex \quad T\bigg)\Bigg)\]
  where:

  \begin{itemize}
    \item The selection $\sigma_R$ selects tuples with the criterion
      $(.1 = 1 \land .2 \neq 0) \lor (.1 = 2 \land .2 = 0)$
    \item The selection $\sigma_S$ selects tuples with the criterion
      $(.1 = 2 \land .2 \neq 0) \lor (.1 = 1 \land .2 = 0)$
    \item The selection $\sigma_=$ selects tuples with the criterion $.1 = .5
      \land .2 = .6 \land .3 = .7 \land .4 = .8$
    \item The projection $\Pi$ projects on attribute~$9$.
  \end{itemize}

  We now explain, given the $\a\b$-bipartite po-relation $\OR_{\a\b} = (\ID, T, <)$,
  how we interpret the relation names $R$, $S$, $T$, and $W$ with unordered
  relations. Let $U \sqcup V$ be the partition of~$\ID$ into $\a$-labeled and
  $\b$-labeled elements, let $n \colonequals \card{U}$ and $m \colonequals
  \card{V}$, and write $U = (u_1, \ldots, u_n)$ and $V = (v_1, \ldots, v_m)$
  following some arbitrary order. We define the input po-database $D$ as follows:

  \begin{itemize}
    \item $R$ is interpreted as an unordered po-relation containing tuples
      labeled $0, 1, \ldots, n$
    \item $S$ is interpreted as an unordered po-relation containing tuples
      labeled $0, 1, \ldots, m$
    \item $T$ is interpreted as an unordered po-relation containing the
      following tuples:
      \begin{itemize}
        \item $\{\langle 1,i,1,0,\a\rangle \mid 1 \leq i \leq n\}$
        \item $\{\langle 1,i,2,j,\cn\rangle \mid 1 \leq i \leq n, 1 \leq j \leq
          m, \text{such that we have~} u_i < v_j \text{~in $\OR_{\a\b}$}\}$
        \item $\{\langle 2,0,2,j,\b\rangle \mid 1 \leq j \leq m\}$
      \end{itemize}
    \item $W$ is interpreted as an unordered po-relation containing $i\times j$
      tuples with label $\cn$.
  \end{itemize}

  The construction that we have described is clearly in PTIME.

  \bigskip

  Towards showing correctness, we first explain how query evaluation proceeds.
  First, 
  $\sigma_R(\ordern{2} \times_\lex R)$ creates a po-relation $\OR_R'$
  (illustrated in Figure~\ref{fig:prf1}) with a
  tuple $\id_{1,i}$ labeled $\langle 1, i\rangle$
  for $1 \leq i \leq n$
  and a tuple $\id_{2,0}$ labeled $\langle 2, 0\rangle$,
  with $\id_{1,i} < \id_{2,0}$ for all~$1 \leq i \leq n$
  and no other comparability pairs.
  Likewise, 
  $\sigma_S(\ordern{2} \times_\lex S)$ creates a po-relation $\OR_S'$ (also
  illustrated in Figure~\ref{fig:prf1}) with a
  tuple $\id_{1,0}$ labeled $\langle 1, 0\rangle$ and tuples $\id_{2,j}$ labeled $\langle 2, j\rangle$
  for $1 \leq j \leq m$, with $\id_{1,0} < \id_{2,j}$ for all~$1 \leq j \leq m$
  and no other comparability pairs.

  We now do the $\times_\gen$ product of $\OR_R'$ and $\OR_S'$, and write $\OR'
  \colonequals \OR_R' \times_\gen \OR_S'$: see 
  Figure~\ref{fig:prf2} for an illustration. Formally, $\OR'$ has
  four kinds of identifiers:

  \begin{itemize}
    \item Tuples $\id_{1,i,1,0}$ with label $\langle 1,i,1,0\rangle$ for $1 \leq
      i \leq n$. These are pairwise incomparable, because the
      $\id_{1,i}$ were.
    \item Tuples $\id_{2,0,2,j}$ with label $\langle 2,0,2,j\rangle$ for $1 \leq
      j \leq m$, which are also pairwise incomparable.
    \item Tuples $\id_{1,i,2,j}$ with label $\langle 1,i,2,j\rangle$ for $1 \leq
      i \leq n$ and $1 \leq j \leq m$, which are pairwise incomparable.
    \item One tuple $\id_{2,0,1,0}$ with label $\langle 2,0,1,0 \rangle$.
  \end{itemize}

  Intuitively, the identifier $\id_{1,i,1,0}$ will
  represent element $u_i$, the identifier $\id_{2,0,2,j}$ will represent element
  $v_j$, the identifier $\id_{1,i,2,j}$ will represent an edge between $u_i$ and
  $v_j$ (denoted $e_{i,j}$), and the identifier $\id_{2,0,1,0}$ is not important
  and will be removed soon by the selection~$\sigma_=$.
  Note that we have an element $e_{i,j}$ for all
pairs of identifiers in $U \times V$, no matter whether they are comparable
in~$\OR_{\a\b}$.

  As for the order relations across identifiers of different kinds,
  they are as follows:

  \begin{itemize}
    \item For $1 \leq i \leq n$, the identifier $\id_{1,i,1,0}$ is less than
      the identifier $\id_{2,0,1,0}$ and it is less than the
      identifiers $\id_{1,i,2,j}$ and the identifiers $\id_{2,0,2,j}$ for all $1
      \leq j \leq m$.
    \item The identifier $\id_{2,0,1,0}$ is less than the identifiers
      $\id_{2,0,2,j}$ for all $1 \leq j \leq m$.
    \item For all $1 \leq i \leq n$ and $1 \leq j \leq m$, the identifier
      $\id_{1,i,2,j}$ is less than the identifier $\id_{2,0,2,j}$.
  \end{itemize}

  Forgetting about $\id_{2,0,1,0}$ which is not important, the intuition is that
  we have $u_i < e_{i,j} < v_j$ for all $i$ and $j$.

  We have described $\OR' = \OR_R' \times_\gen \OR_S'$,
  and we continue describing how query evaluation proceeds.
  The intuition now is that we wish
  to only keep the $e_{i,j}$ such that $u_i < v_j$
  in~$\OR_{\a\b}$. We cannot express this directly using a selection, because
  the selection criterion would depend on~$\OR_{\a\b}$ so it would not be fixed.
  Instead, we take the product of~$\OR'$
  with the unordered po-relation $T$ that specifies which identifiers we wish to
  keep, and then we perform the fixed selection $\sigma_=$.
  Specifically, we do the $\times_\lex$ product of~$\OR'$ with~$T$, followed by the
  selection~$\sigma_=$, which intuitively
  replaces each element of~$\OR'$ by the contents of relation $T$, i.e.,
  unordered identifiers, so the order relation in the result of the
  $\times_\lex$ product is entirely defined by the first component, i.e.,
  by~$\OR'$. The
  selection~$\sigma_=$ then keeps the $\id_{1,i,2,j}$ such that $u_i < u_j$, and it also
  keeps the $\id_{1,i,1,0}$ and the $\id_{2,0,2,j}$; it discards the other
  $\id_{1,i,2,j}$ as well as the unimportant identifier $\id_{2,0,1,0}$.
  After the selection, we perform a
  projection~$\Pi$ to rename the identifiers using the last component
  of the tuple labels in~$T$: identifiers that come from the $\id_{1,i,1,0}$ are relabeled
  $\a$, identifiers that come from the $\id_{2,0,2,j}$ are relabeled~$\b$, and
  identifiers that come from the $\id_{1,i,2,j}$ are relabeled~$\cn$. 
  Last, we do the union with $W$ to add $n \times m$ unordered identifiers
  labeled $\cn$.

  To summarize, the po-relation $Q(D)$ contains the following identifiers:

  \begin{itemize}
    \item $m \times n$ unordered identifiers labeled $\cn$, each of which is
      incomparable to all other identifiers.
    \item $n$ identifiers corresponding to the $\id_{1,i,1,0}$ in~$\OR'$ for $1
      \leq i \leq n$, that are labeled $\a$, and that are incomparable among
      themselves: we identify each $\id_{1,i,1,0}$ to the identifier~$u_i$ in~$\OR_{\a\b}$.
    \item $m$ identifiers corresponding to the $\id_{2,0,2,j}$ in~$\OR'$ for
      $1 \leq j \leq m$, that are labeled $\b$, and that are incomparable among
      themselves: we identify each $\id_{2,0,2,j}$ to the identifier~$v_j$ in~$\OR_{\a\b}$.
    \item One identifier corresponding to $\id_{1,i,2,j}$ in~$\OR'$ for each
      $1 \leq i \leq n$ and $1 \leq j \leq m$ such that $u_i < v_j$
      is a comparability pair in~$\OR_{\a\b}$: we call each of them $e_{i,j}$ for
      brevity.
  \end{itemize}

  The comparability pairs across these identifiers are simply the following:
  $u_i < v_j$ for all $1 \leq i \leq n$ and $1 \leq j \leq n$, and $u_i < e_{i,j} <
  v_j$ for all $i$ and $j$ such that~$e_{i,j}$ exists.
  In particular, note that the order between the $u_i$ and $v_j$
  is \emph{not} like in~$\OR_{\a\b}$, because all $u_i$ are less than all~$v_j$.
  We will work around this issue when defining our candidate possible world to
  read the comparability relation from the~$e_{i,j}$.

  To define the candidate possible world,
  consider now the integers $p,q \in \NN$ that were given as input to the
  NP-hard problem of Proposition~\ref{prp:bipartiteposs} along with $\OR_{\a\b}$.
  Let $0 \leq \pi \leq m \times n$ be the number of comparability pairs of~$\Gamma$.
  Construct the totally ordered po-relation $L'_{p,q} \colonequals \a^p
  \cn^{m \times n} \a^{u-p} \b^q \cn^\pi \b^{v-q}$, which we will use as our
  candidate possible world.
  We claim that the \poss problem for $L_{p,q}$ and $\OR_{\a\b}$ reduces to the
  same problem for $L'_{p,q}$ and
  $Q(D)$, which suffices to conclude the proof.

  In one direction, assume that $L_{p,q} \in \pw(\OR_{\a\b})$, and consider a
  witnessing linear extension. We build a linear extension of $Q(D)$ achieving
  $L_{p,q}'$ as follows:
  
  \begin{enumerate}
    \item Enumerate the same $\a$-labeled identifiers in~$Q(D)$ as the ones in
      the witnessing linear
      extension of~$\OR_{\a\b}$ that achieves the factor~$\a^p$ of~$L_{p,q}$.
    \item Enumerate all the $e_{i,j}$ that can be enumerated: there are at most $m \times
  n$ in total so we can enumerate all that are available at this point.
    \item Enumerate some $\cn$-labeled identifiers from~$W$ afterwards if
      necessary, to enumerate $m \times n$ $\cn$-labeled identifiers in total.
    \item Enumerate all remaining $\a$-labeled identifiers.
    \item Enumerate the same $\b$-labeled identifiers in~$Q(D)$ as the ones in
      the witnessing linear
      extension of~$\OR_{\a\b}$ that achieves the factor~$\b^q$ of~$L_{p,q}$.
      To see why
      these identifiers 
      can be enumerated at this stage in~$Q(D)$,
      assume by way of contradiction that we try to enumerate $v_j$ in~$Q(D)$
      but that this violates an order constraint of~$Q(D)$, i.e., $v_j$ is
      greater than
      another identifier that has not been enumerated yet. 
      As all $\a$-labeled
      identifiers of~$Q(D)$ have been enumerated in steps~1 and~4,
      the only identifiers that can block the
      $\b$-labeled identifier $v_j$ from being enumerated
      are the $\cn$-labeled identifiers $e_{i,j}$ that have not been enumerated
      at step~2. So there must be $1 \leq i \leq n$ such the
      identifier~$e_{i,j}$ exists in~$Q(D)$ and was not
      enumerated at step~2. Now, the only way for this to happen is if
      the~$\a$-labeled element~$u_i$ was not enumerated at step~1.
      However, the
      existence of~$e_{i,j}$ in~$Q(D)$ witnesses that $u_i < v_j$
      in~$\OR_{\a\b}$, 
      and in the linear extension of~$\OR_{\a\b}$
      we must have enumerated $u_i$ before~$v_j$. Hence, 
      $u_i$ was enumerated at step~1 and $e_{i,j}$ was enumerated at step~2
      and $v_j$ can now be enumerated, a contradiction.
    \item Enumerate the remaining $\cn$-labeled identifiers and $\b$-labeled
      identifiers arbitrarily, which is clearly possible as no comparability
      pairs between unenumerated elements remain.
  \end{enumerate}

  In the converse direction, assume that $L'_{p,j} \in \pw(Q(D))$, and consider
  a witnessing linear extension. We build a linear extension of~$\OR_{\a\b}$ achieving
  $L_i$ by matching the factors $\a^p$ and~$\b^q$ to the elements matched to
  these factors in~$Q(D)$, and finishing by enumerating the remaining
  $\a$-labeled and~$\b$-labeled elements in some arbitrary way.
  The only thing to show is that we do not violate the order constraints
  of~$\OR_{\a\b}$ while achieving the factors $\a^p$ and $\b^q$.
  To show this, assume by way of contradiction
  that we try to enumerate some identifier $v_j$ when
  achieving~$\b^q$ but we have $u_i < v_j$  for some identifier $u_i$ that was not enumerated when
  achieving $\a^p$. In this case, the comparability pair $u_i < v_j$
  of~$\OR_{\a\b}$ witnesses the
  existence of an element $e_{i,j}$ in~$Q(D)$ such that $u_i < e_{i,j} < v_j$
  in~$Q(D)$. Now, as we did not enumerate $u_i$ to achieve $\a^p$ in~$Q(D)$, we
  cannot have enumerated $e_{i,j}$ when achieving $\cn^{m \times n}$, hence
  $e_{i,j}$ witnesses that we cannot have enumerated $v_j$ when achieving $\b^q$
  in~$Q(D)$, a contradiction. Hence, the order constraints of~$\OR_{\a\b}$ are respected.
  
  This concludes the correctness argument, so we have shown the NP-hardness of
  \poss in our context, which concludes the proof of
  Theorem~\ref{thm:intractable}.
\end{proof}

\end{toappendix}

\clearpage
\bibliographystyle{plainurl}
\bibliography{main}

\end{document}